\documentclass[12pt,english]{article}
\usepackage{lmodern}
\usepackage{lmodern}
\usepackage[T1]{fontenc}
\usepackage[cp1252]{inputenc}
\usepackage{geometry}
\usepackage{color}
\usepackage{babel}
\usepackage{amsmath}
\usepackage{amsthm}
\usepackage{amssymb}
\usepackage{setspace}
\usepackage{siunitx}
\usepackage[authoryear,comma, longnamesfirst]{natbib}
\usepackage{microtype}
\usepackage[unicode=true,
 bookmarks=false,
 breaklinks=true,pdfborder={0 0 1},backref=section,colorlinks=true]
 {hyperref}
\hypersetup{pdftitle={"DRUM"},
 pdfauthor={"KashaevAguiarPlavalaGauthier"},
 pdfnewwindow=true,pdfstartview=FitH,urlcolor=blue!90!red!45!black,citecolor=blue!90!red!45!black,linkcolor=red!90!black}

\makeatletter
\theoremstyle{plain}
\newtheorem{thm}{\protect\theoremname}
\theoremstyle{definition}

\theoremstyle{definition}
\newtheorem{example}[thm]{\protect\examplename}
\theoremstyle{plain}

\usepackage{amsfonts}
\usepackage{dsfont}\usepackage{mathrsfs}\usepackage{ushort}

\usepackage{titlesec}\usepackage{titling}
\usepackage{caption}
\usepackage{enumitem}\usepackage{booktabs}
\usepackage{tikz}
\usetikzlibrary{decorations.pathreplacing}
\usepackage{pstricks}\usepackage{pst-all}
\usepackage{pst-plot}
\usepackage{pst-node}\usepackage{pst-3dplot}\usepackage{sgamevar}
\usepackage{subcaption}
\usepackage{lscape}
\usepackage{pifont}
\usepackage{longtable}
\usepackage{multirow}
\usepackage{algorithmic}
\usepackage{pdflscape}
\usepackage{rotating}
\usepackage{float}

\setenumerate{label=\small(\roman*)}
\DeclareCaptionFont{fancy}{\bfseries\sffamily}
\gamemathtrue
\allowdisplaybreaks

\providecommand{\psreset}{\psset{%
		linewidth=0.3pt,linestyle=solid,linecolor=black,
		dotsize=2.5pt,dotsep=2.5pt,arrowsize=4pt,
		fillstyle=none,fillcolor=white,
		showpoints=false,arrows=-,linearc=0,framearc=0,
		hatchsep=2pt,hatchwidth=0.2pt,nodesep=4pt,opacity=1}
	\psset{gridcolor=black!60, subgridcolor=black!30}
}

\psreset

\usepackage{graphicx}
\graphicspath{ {figures/} }

\titleformat{\section}[block]{\centering\large\bfseries\sffamily}{\thesection.}{0.5em}{}
\titleformat{\subsection}[block]{\flushleft\bfseries}{\thesubsection.}{0.5em}{}
\titleformat{\subsection}[block]{\flushleft\bfseries\sffamily}{\thesubsection.}{0.5em}{}
\titleformat{\subsubsection}[runin]{\normalsize\bfseries\sffamily}{\bfseries\upshape\sffamily\thesubsubsection.}{0.5em}{}[.--\:]
\renewcommand{\thesubsubsection}{\arabic{section}.\arabic{subsection}.\arabic{subsubsection}}
\titlespacing{\section}{0ex}{10ex}{5ex}
\titlespacing{\subsection}{0in}{6ex}{3ex}
\titlespacing{\subsubsection}{0mm}{2ex}{0.5em}
\pretitle{\begin{center}\LARGE\bfseries\sffamily}
\posttitle{\par\end{center}\vskip 0.5em}
\preauthor{\begin{center} \large \lineskip 0.5em\begin{tabular}[t]{c}}
\postauthor{\end{tabular}\par\end{center}}
\predate{\begin{center}\small}
\postdate{\par\end{center}}
\providecommand{\abstitle}[1]{{\par\vspace*{2ex}\small\bfseries\sffamily #1}\hspace*{1ex}}
\renewenvironment{abstract}%
{\begin{center}\begin{minipage}{0.8\linewidth}%
			\abstitle{Abstract}\small}%
		{\end{minipage}\end{center}\vfill\clearpage}

\DeclareMathOperator*{\argmax}{arg\,max}

\providecommand{\Char}[1]{\mathds{1}\left(\,#1\,\right)}
\providecommand{\Real}{{\mathds{R}}}

\providecommand{\tr}{^{\prime}}
\providecommand{\as}{\ensuremath{\mathrm{a.s.}}}
\providecommand{\rand}[1]{\mathbf{#1}}
\providecommand{\rands}[1]{\boldsymbol{#1}}

\providecommand{\Exp}[1]{\mathds{E}\left[#1\right]}
\providecommand{\abs}[1]{\left\lvert#1\right\rvert}

\theoremstyle{remark}
  \theoremstyle{plain}
  \newtheorem{lemma}{\protect\lemmaname}\theoremstyle{definition}
    \newtheorem{proposition}{\protect\propositionname}\theoremstyle{definition}
  \newtheorem{definition}{\protect\definitionname}\theoremstyle{plain}
\newtheorem{theorem}{\protect\theoremname}\theoremstyle{plain}
  \newtheorem{cor}{\protect\corollaryname}\theoremstyle{definition}
  \providecommand{\assumptionname}{Assumption}

  \providecommand{\definitionname}{Definition}
  \providecommand{\lemmaname}{Lemma}
  \providecommand{\propositionname}{Proposition}
  \providecommand{\remarkname}{Remark}
\providecommand{\corollaryname}{Corollary}
\providecommand{\theoremname}{Theorem}
\providecommand{\examplename}{Example}

\makeatother

\providecommand{\definitionname}{Definition}
\providecommand{\examplename}{Example}
\providecommand{\lemmaname}{Lemma}
\providecommand{\theoremname}{Theorem}

\usepackage{tikz}
\usepackage{tikz-3dplot}
\usetikzlibrary{calc}

\begin{document}
\title{Dynamic and Stochastic Rational Behavior\thanks{{\tiny This paper subsumes ``Nonparametric Analysis of Dynamic Random Utility Models.'' The ``\textcircled{r}'' symbol indicates that the authors' names are in certified random order, as described by \citet{ray2018certified}. We thank Roy Allen, Chris Chambers, Pierre-Andr\'e Chiappori, Mark Dean, Adam Dominiak, Laura Doval, Mikhail Freer, David Freeman, Matt Kovach, Elliot Lipnowski, Paola Manzini, Krishna Pendakur, Matt Polisson, John Quah, J\"org Stoye, Jesse Shapiro, Tomasz Strzalecki, and Levent \"Ulk\"u for useful discussions and encouragement. Pl\'avala acknowledges support from the Deutsche Forschungsgemeinschaft (DFG, German Research Foundation, project numbers 447948357 and 440958198), the Sino-German Center for Research Promotion (Project M-0294), the ERC (Consolidator Grant 683107/TempoQ), the German Ministry of Education and Research (Project QuKuK, BMBF Grant No. 16KIS1618K), and the Alexander von Humboldt Foundation. Aguiar thanks USFQ School of Economics for kindly hosting him during the writing of this paper.}}}

\author{ 
	Nail Kashaev \textcircled{r}
	Victor H. Aguiar \textcircled{r}
	Martin Pl\'avala \textcircled{r}
    Charles Gauthier\thanks{Kashaev: Department of Economics, University of Western Ontario; \href{mailto:nkashaev@uwo.ca}{nkashaev@uwo.ca}. Aguiar: Department of Economics, University of Western Ontario; \href{mailto:vaguiar@uwo.ca}{vaguiar@uwo.ca}.
    Pl\'avala: Naturwissenschaftlich-Technische Fakult\"{a}t, Universit\"{a}t Siegen; \href{martin.plavala@uni-siegen.de}{martin.plavala@uni-siegen.de}.
	Gauthier: ECARES, Universit\'e Libre de Bruxelles; \href{mailto:charles.gauthier@ulb.be}{charles.gauthier@ulb.be}.}
	}
\date{This version August, 2023. First version: February, 2023}
\maketitle

\begin{abstract}
\footnotesize{The (static) utility maximization model of \citet{afriat1967construction}, which is the standard in analysing choice behavior, is under scrutiny. We propose the Dynamic Random Utility Model (DRUM) that is more flexible than the framework of \citet{afriat1967construction} and more informative than the static Random Utility Model (RUM) framework of \citet{mcfadden1990stochastic}. Under DRUM, each decision-maker randomly draws a utility function in each period and maximizes it subject to a menu. DRUM allows for unrestricted time correlation and cross-section heterogeneity in preferences. We characterize DRUM for situations when panel data on choices and menus are available. DRUM is linked to a finite mixture of deterministic behaviors that can be represented as a product of static rationalizable behaviors. This link allows us to convert the characterizations of the static RUM to its dynamic form. In an application, we find that although the static utility maximization model fails to explain population behavior, DRUM can explain it. 
}

\end{abstract}

\section{Introduction}
A fundamental question in economics is whether decision makers exhibit rational choice behavior. The traditional definition of rationality is effectively equivalent to maximizing a utility function that is fixed in time. However, the static utility maximization model is under empirical scrutiny. Here, we study a notion of rationality in choice behavior that is stochastic and dynamic\textemdash the Dynamic Random Utility Model (DRUM). Under DRUM, each consumer or decision maker (DM) at each period maximizes the utility realized from a stochastic utility process subject to a menu or budget.\footnote{From the point of view of the DMs, utilities and budgets are known and deterministic. Stochasticity appears due to unobserved heterogeneity from the point of view of the observer.} We provide a revealed preference characterization of DRUM for situations in which the longitudinal distribution of choices or demands is observed for a finite collection of menus or budgets in a finite time window. This characterization does not place any parametric restriction on (i) the form of utility functions, (ii) the correlation of utilities in time, and (iii) the heterogeneity of utility in the cross-section. In an application, we show DRUM can explain the behavior of a cross section of DMs when many of them fail to be consistent with the static utility maximization model. 

The two main frameworks available for analyzing consumer behavior are the framework of static utility maximization based on \citet{samuelson1938note} and \citet{afriat1967construction} and the framework of random utility maximization (RUM) based on  \citet{mcfadden1990stochastic}. DRUM addresses several  empirical limitations of these earlier models. In particular, the
Samuelson-Afriat framework is under scrutiny due to experimental and field evidence against it.\footnote{For examples in regard to household consumption see \citet{echenique2011money} and \citet{dean2016measuring}, and in regard to choices over portfolios with risk or uncertainty see \citet{choi2007revealing,choi2014more,ahn2014estimating}. The rationality violations were originally thought to be small \citep{echenique2011money,choi2007revealing}, but newer experimental datasets show these violations can be severe \citep{brocas2019consistency,Aguiar2021JMathE,halevy2022identifying}.} There is evidence that failures of the Samuelson-Afriat framework are driven by the stringent assumption of the stability of preferences over time. For example, utility functions may change over time because of variability in time of the neural computation of value \citep{kurtz2019neural}, of structural breaks \citep{cherchye2017household}, or of evolving risk aversion \citep{guiso2018time,akesaka2021temporal}. DRUM allows preferences to change freely in time. In contrast to the Samuelson-Afriat framework, RUM has found reasonable success explaining repeated cross-sections of household choices \citep{kawaguchi2017testing,kitamura2018nonparametric}. However, RUM cannot take advantage of the longitudinal variation in choices available in many datasets, and it may have limited empirical bite \citep{im2021non}. By considering a richer primitive, we simultaneously relax the assumption of stable preferences over time implicit in the Samuelson-Afriat framework while providing a more informative test of stochastic utility maximization than in the McFadden-Richter framework. 

Our first result is a mixture characterization of DRUM, which is analogous to the RUM characterization in \citet{mcfadden1990stochastic}. We exploit the fact that DRUM is associated with a finite mixture of preference profiles in time. We obtain results analogous to \citet{kitamura2018nonparametric} (henceforth, KS), \citet{mcfadden1990stochastic}, and \citet{kawaguchi2017testing} with a dynamic version of the Axiom of Stochastic Revealed Preferences. This finite mixture characterization lends itself to statistical testing using the results in KS. This characterization can also be used for nonparametric counterfactual analysis. In a Monte Carlo study, we show that the statistical test of KS applied to our characterization of DRUM performs well in finite samples. 

We show that the mixture representation of DRUM can be obtained using a Kronecker product of the mixture representation of RUM in each period.\footnote{Informally, the mixture representation of RUM can be represented as a matrix whose columns are deterministic rational demand types. The analogous matrix for DRUM is the Kronecker product of those RUM matrices.} This observation is vital to obtain (i) computational gains for testing because of the modularity of the mixture representation; and (ii) a novel characterization of DRUM using a recursive version of the \citet{block1960random} inequalities (henceforth, BM inequalities). In a static setting, \citet{kitamura2018nonparametric} were the first to observe that the empirical content of RUM can be expressed as cone restrictions on observed data. In particular, the Weyl-Minkowski theorem posits that a cone can be described equivalently by a convex combination of its vertices ($\mathcal{V}$-representation) or by its faces ($\mathcal{H}$-representation). \citet{kitamura2018nonparametric} note that the $\mathcal{H}$-representation corresponds to what decision theorists would call an axiomatic characterization of RUM.\footnote{Note that the BM inequalities are the $\mathcal{H}$-representation of RUM in the finite abstract setup.} We exploit recent mathematical advancements in the analysis of the Kronecker products of cones (due  to \citealp{aubrun2021entangleability} and \citealp{aubrun2022monogamy}) to provide an axiomatic characterization of DRUM using the axiomatic characterization of RUM. This paper is the first to bring this new mathematical tool to economics and to show how it can be used in the DRUM setup and in structurally similar models.\footnote{For example, models of bounded rationality, such as the model of random consideration of \citet{cattaneo2020random} and its extension with heterogeneous preferences in \citet{RAUM2022random}, can be extended to a dynamic setup in the spirit of DRUM.} 

The generalized Weyl-Minkowski theorem enables us to provide a full characterization of DRUM via dynamic BM inequalities, which covers as a special case the finite abstract setups of  \citet{li2021axiomatization} and of \citet{chambers2021correlated} with full menu variation. Our characterization works for cases of limited observability of menus and in the presence of a primitive order that is respected by the support of the utility process.  We also provide a novel behavioral condition necessary for the consistency of the longitudinal distribution of demand with DRUM, $\mathrm{D}$-\textit{monotonicity}. It is also sufficient in simple setups: (i) for any finite number of goods and two budgets per period, and (ii) for two goods and any finite number of budgets per period. $\mathrm{D}$-monotonicity is computationally simple to check and provides a deeper understanding of the empirical content of DRUM. It restricts the joint probability of choices in time beyond the RUM restrictions on marginal distributions in each period.  $\mathrm{D}$-monotonicity can be thought of as a dynamic version of the Weak Axiom of Stochastic Revealed Preference (see \citealp{bandyopadhyay1999stochastic,hoderlein2014revealed}) and a stochastic version of the Weak Axiom of Revealed Preference (in time series) by \citet{samuelson1938note}.  
 
We synthesize the two main paradigms of nonparametric demand analysis: Samuelson and Afriat and McFadden-Richter frameworks. The Samuelson-Afriat framework requires observing a \emph{time-series} of choices and budgets for a given consumer and assumes that the consumer maximizes the same utility function in each period. When preferences are allowed to vary in time, there are no empirical implications with only a time-series of choices. However, when a \emph{panel} of choices is used, DRUM bounds the share of consumers or DMs whose choices contain a revealed preference violation in the sense of \citet{afriat1967construction}. RUM instead requires observing a \emph{cross-section} of choices and budgets from a population of consumers. The panel structure is ignored as there is no time dimension. Hence, this approach misses the potential temporal correlation of utilities. As a result, panels of choices over menus exist that, when marginalized, are consistent with RUM but not with DRUM. In other words, ignoring the time dimension of choices may lead to false positives when testing DRUM. Importantly, our setup keeps the fundamental assumption in the McFadden-Richter framework\textemdash the distribution of utilities does not depend on the sequence of budgets or menus that the consumer faces in time.\footnote{This assumption can be relaxed in the same spirit as \citet{deb2017revealed}.} 

Our synthesis is advantageous because (i) it provides more informative bounds on counterfactual choice due to the richer variation in the panel of choices, (ii) it provides a theoretical justification for the validity of the RUM framework when marginalizing choices, and (iii) it clarifies the role of constant preferences across time in the Samuelson-Afriat framework. Fortunately, our primitive with a longitudinal level of variation is readily available in many consumption surveys, household scanner datasets, and experimental datasets, as documented in \citet{AK2021}.\footnote{In practice, panels of choices are often pooled in the time dimension to create a cross-section with enough budget variations to have empirical bite \citep{deb2017revealed,kitamura2018nonparametric}. In this case, we show that this approach could lead to false rejections of DRUM due to ignoring the time labels of budgets.} 

In our application, we find support for DRUM in a panel dataset collected by \citet{ABBK23}. In a large experimental dataset collected in Amazon Mechanical Turk, a large cross section of DMs ($2135$ DMs) faced a sequence of binary comparisons of lotteries. Although $8$ percent of DMs are inconsistent with the static utility maximization model, we cannot reject the null hypothesis that DRUM can explain the data.\footnote{A restriction of DRUM to Expected Utility fails to explain the behavior of this sample of DMs.}  Monte Carlo experiments mimicking the application setup provide evidence of high power of our DRUM test in finite samples.   

The DRUM framework is rich and extends well beyond the Samuelson-Afriat and McFadden-Richter worlds. We therefore can cover all of the following special cases: (i) consumption models of errors in the evaluation of utility \citep{kurtz2019neural}; (ii) dynamic random expected utility (defined in \citealp{frick2019dynamic}); (iii) static utility maximization in a population (without measurement error) \citep{AK2021}; (iv) dynamic utility maximization in a population \citep{browning1989anonparametric,gauthier2018,AK2021}; (v) changing utility or multiple-selves models \citep{cherchye2017household}; and (vi) changing-taste modelled with a constant utility in time with an additive shock \citep{adams2015prices}. 
\par

\noindent\textbf{Outline} The paper is organized as follows. Section~\ref{sec: setup} introduces the setup. Section~\ref{sec: characterization of DRUM} provides both a McFadden-Richter and a KS-type characterization of DRUM. Section~\ref{sec: behavioral characterization} provides a behavioral characterization of DRUM via linear inequality constraints.  Section~\ref{secc:AfriatMcFadden} synthesizes the setups of Samuelson-Afriat and McFadden-Richter. 
Section~\ref{sec: counterfactuals} provides results for our dynamic counterfactual analysis. Section~\ref{sec: application} provides an application to experimental data. Section~\ref{sec: litreview} contains the literature review. 
Section~\ref{sec: conclusion} concludes. 
All proofs can be found in Appendix~\ref{app: proofs}. Appendix~\ref{appendix: montecarlo} provides Monte Carlo experiments showcasing the finite sample properties of the statistical test of DRUM.

\section{Setup}\label{sec: setup}
We consider a time window $\mathcal{T}=\{1,\cdots,T\}$ with a finite terminal period $T\geq1$. Let $X^t$ be a nonempty finite choice set. We endow $2^{X^t}\setminus\{\emptyset\}$  with some acyclic partial order $>^t$. When $>^{t}$ is restricted to singletons, it induces an acyclic partial order on $X^t$. We will abuse notation and write $x>^{t}y$ instead of $\{x\}>^{t}\{y\}$ in this case.
In each $t\in \mathcal{T}$, there are $J^t<\infty$ distinct menus denoted by
\[
B^t_{j}\in 2^{X^t}\setminus\{\emptyset\}, \quad j\in\mathcal{J}^t=\{1,\dots, J^t\}.
\]
Since $X^t$ is a finite set, we denote the $i$-th element of menu $j\in\mathcal{J}^t$ as $x^t_{i|j}$. That is,  $B^t_{j}=\{x^t_{i|j}\}_{i\in\mathcal{I}^t_{j}}$, where $\mathcal{I}^t_{j}=\{1,2,\dots,I^t_j\}$ and $I^t_j$ is the number of elements in menu $j$.

Define a menu path as an ordered collection of indexes $\rand{j}=(j_t)_{t\in\mathcal{T}}$, $j_t\in\mathcal{J}^t$. Menu paths encode menus that were faced by agents in different time periods. Let $\rand{J}$ be the set of all \emph{observed} menu paths. Given $\rand{j}\in \rand{J}$, a \emph{choice path} is an array of alternatives $x_{\rand{i}|\rand{j}}=\left(x^t_{i_t|{j_t}}\right)_{t\in\mathcal{T}}$ for some collection of indexes $\rand{i}=\left(i_t\right)_{t\in\mathcal{T}}$ such that $i_t\in\mathcal{I}^t_{j_t}$ for all $t$. Similar to a menu path, a choice path encodes the choices of a DM in a given sequence of menus that she faced. The set of all possible choice path index sets $\rand{i}$, given a menu path $\rand{j}$, is denoted by $\rand{I}_\rand{j}$.

Note that every $\rand{j}\in\rand{J}$ encodes the Cartesian product of menus $\times_{t\in\mathcal{T}}B^t_{j_t}\subseteq \times_{t\in\mathcal{T}}X^t$. Then, for every $\rand{j}$, let $\rho_{\rand{j}}$ be a probability measure on $\times_{t\in\mathcal{T}}B^t_{j_t}$. That is, $\rho_{\rand{j}}\left(x_{\rand{i}|\rand{j}}\right)\geq 0$ for all $\rand{i}\in\rand{I}_\rand{j}$ and $\sum_{\rand{i}\in\rand{I}_\rand{j}}\rho_{\rand{j}}\left(x_{\rand{i}|\rand{j}}\right)=1$. The primitive in our framework is the collection of all observed $\rho_{\rand{j}}$, $\rho=(\rho_{\rand{j}})_{\rand{j}\in\rands{J}}$. We call this collection a \emph{dynamic stochastic choice function}.  

Given $\rho$, we can define a Dynamic Random Utility Model (DRUM). Let $U^t$ denote the set of all utility functions that (i)  map $X^t$ to $\Real$, (ii) are injective,  and (ii) monotone on $>^{t}$ (i.e., if $S,S'\in 2^{X^t}\setminus\{\emptyset\}$ and $S>^{t}S'$, then $\max_{s\in S}{u^t(s)}>\max_{s\in S'}u^t(s)$). Also let $\mathcal{U}=\times_{t\in\mathcal{T}}U^t$  and $u=(u^t)_{t\in\mathcal{T}}$ be an element of $\mathcal{U}$. 

\begin{definition}[DRUM] 
The dynamic stochastic choice function $\rho$ is consistent with DRUM if there exists a probability measure over $\mathcal{U}$, $\mu$, such that
\[
\rho_{\rand{j}}\left( x_{\rand{i}|\rand{j}} \right)=\int \prod_{t\in \mathcal{T}} \Char{\argmax_{y\in B_{j_t}^t}u^t(y)= x^t_{i_t|j_t}}d\mu(u)
\]
for all $\rand{i}\in \rand{I}_{\rand{j}}$ and $\rand{j}\in\rand{J}$.
\end{definition}

When $T=1$, DRUM coincides with RUM, such that every agent maximizes her utility function $u^1$ over a menu and the analyst observes the distribution of consumers' choices. DRUM extends RUM by introducing a time dimension with an unrestricted preference correlation across time. The stochastic utility process is captured by $\mu$. Similar to RUM, DRUM does not restrict preference heterogeneity in cross-sections (i.e., across agents) and requires $\mu$ not to depend on either the menu paths or the alternatives in the consumption space. In contrast to RUM and DRUM, the Samuelson-Afriat framework, , does not use variation in choices of agents in cross-sections (i.e., it is directed to the individual-level data or time series of choices). Thus, it does not restrict preferences of individuals in cross-sections. In contrast to DRUM, however, the Samuelson-Afriat framework imposes a strict restriction that preferences are perfectly correlated across time (i.e. $u^t=u^s\:\mu-\as$ for all $t,s\in\mathcal{T}$). We formalize these connections between RUM, Afriat's framework, and DRUM in Section~\ref{secc:AfriatMcFadden}. 
\par
Some examples of datasets in which a dynamic stochastic choice function is (partially) observed are: (i) household longitudinal survey datasets, (ii) scanner datasets, and (iii) experimental datasets with panels of choice. In survey datasets (e.g., \textit{Encuesta de Presupuestos Familiares} in Spain and Progresa Household Survey in Mexico, see \citealp{deb2017revealed,AK2021}), information about household purchases is usually collected several times a year. For a given time period, budget variation across households is driven by spatial or regional price variation \citep{AK2021}. Scanner datasets (e.g., Nielsen homescan data,  see \citealp{gauthier2018}) contain information about weekly purchases of consumers. Budget variation in this case is driven by price variation across stores in each time period \citep{gauthier2021}. In experimental settings, subjects often face  few budget paths drawn at random from a common set of budgets (e.g., experiments on preferences over giving, as in \citealp{porter2016love}). In our empirical application, we have a panel of choices over different menus of lotteries from an experimental data set collected by \citet{ABBK23}, (see also  \citet{mccausland2020testing} for a panel of choices in discrete choice.\footnote{In that paper they study static RUM in an individual setup using the time-series to estimate individual stochastic choice.})

\subsection{Preview of the Results: Binary Menus Example}

We illustrate the setting and our main result with an example. We start with the static RUM setup for a choice set $X^t=\{x,y,z\}$, where one  observes only binary menus. The order $>^{t}$ is assumed to be empty. Following \citet{mcfadden1990stochastic}, we describe RUM as a finite mixture of deterministic types captured by a matrix $A^t$, as displayed in Table~\ref{tab: A_t binary}. Each column of $A^t$ in Table~\ref{tab: A_t binary} corresponds to a deterministic rational type $r_i^t$ (e.g., $r_1^t$ represents a strict rational order over $X^t$ such that $xr_1^tyr_1^tz$). Note that there is some utility function $u_1^t$ such that $xr_1^tyr_1^tz$ if and only if $u_1^t(x)>u_1^t(y)>u_1^t(z)$. Each row in Table~\ref{tab: A_t binary} corresponds to a choice from a binary menu for each rational order. An element of $A^t$ that corresponds to type $r^t_i$ and pair $x',\{x',y'\}$ is equal to $1$ if $x' r^t_i y'$, and zero otherwise.
\begin{table}[h]
\begin{centering}
\scalebox{0.9}{
\begin{tabular}{c!{\vrule width 2pt}c|c|c|c|c|c|c|}
 & $r_1^t$& $r_2^t$ & $r_3^t$ & $r_4^t$ & $r_5^t$ & $r_6^t$ \\
\noalign{\hrule height 2pt}
$x,\{x,y\}$ & $1$ & $1$ & - & - & 1 & - \\
\hline 
$y,\{x,y\}$ & - & - & $1$ & $1$ & - & $1$\\
\hline 
$x,\{x,z\}$ & $1$ & $1$ & $1$ & - & - & - \\
\hline 
$z,\{x,z\}$ & - & - & - & $1$ & $1$ & $1$ \\
\hline
$y,\{y,z\}$ & $1$ & - & $1$ & $1$ & - & - \\
\hline 
$z,\{y,z\}$ & - & $1$ & - & - & $1$ & $1$ 
\end{tabular}
}
\par\end{centering}
\caption{The matrix $A^t$ for binary menus. $``-"$ corresponds to zero.}\label{tab: A_t binary}
\end{table}
Now consider two periods in which the choice set remains the same in time. DMs face (sequentially) two menus. Thus, menu paths are of the form $\{x',y'\},\{x'',y''\}$ for all $x',y',x'',y''\in X^t$. Given a menu path, DMs choose a choice path (e.g., the ordered tuple $(x,\{x,y\};y,\{y,z\})$ indicating the choices from menu path $\{x,y\},\{y,z\}$). Our primitive, or data set, is the collection of the joint probabilities of choice paths for each menu path, $\rho$. We show that $\rho$ is consistent with DRUM  (up to rearrangement) if and only if $\rho=(A^1\otimes A^2)\nu$ for some vector $\nu\geq 0$ such that $\sum_{i}{\nu_i}=1$, where $\otimes$ is the Kronecker product.\footnote{The result of the Kronecker product of two matrices $A$ and $B$ is a block matrix 
\[
A\otimes B=\left( \begin{array}{cccc}
     A_{1,1}B&A_{1,2}B&\dots&A_{1,n_A}B  \\
     A_{2,1}B&A_{2,2}B&\dots&A_{2,n_A}B \\
     \dots&\dots&\dots&\dots \\
     A_{k_A,1}B&A_{k_A,1}B&\dots&A_{k_A,n_A}B
\end{array} \right),
\] where $A$ is of the size $k_A$ by $n_A$.} The vector $\nu$ is a distribution over dynamic profiles of deterministic rational types or columns of $A^1\otimes A^2$ (e.g., a dynamic preference profile $(r_1^1,r_6^2)$ is such that preferences change from $xr_1^1yr_{1}^{1}z$ to $zr_{6}^{2}yr_6^2x$). This representation of DRUM as a mixture of deterministic dynamic rational types is called the $\mathcal{V}$-representation. The recursive structure of the $\mathcal{V}$-representation makes DRUM modular with the consequent computational gains. In addition, when $\mathcal{T}=\{1\}$, the consistency of $\rho$ with RUM is equivalent to the following triangle condition
\[
\rho_{\{x',y'\}}(x')+\rho_{\{y',z'\}}(y')- \rho_{\{x',z'\}}(x')\geq 0  
\] 
for all $x',y',z'\in X^t$.\footnote{The triangle conditions are equivalent to RUM with binary menus when $|X^t|\leq 5$ \citep{dridi1980sur}.} The triangle conditions can be summarized in a matrix $H^t$, as displayed in Table~\ref{table:H binary}. The triangle conditions can then be stated as $H^t\rho\geq0$ for the static case. This is called the $\mathcal{H}$-representation of RUM. 

\begin{table}[ht]
\centering
\begin{tabular}{cccccc}
\hline
$x,\{x,y\}$ & $y,\{x,y\}$ & $x,\{x,z\}$ & $z,\{x,z\}$ & $y,\{y,z\}$ & $z,\{y,z\}$ \\
\hline
1 & - & -1 & - & 1 & - \\
-1 & - & 1 & - & - & 1 \\
- & 1 & 1 & - & -1 & - \\
- & -1 & - & 1 & 1 & - \\
1 & - & - & 1 & - & -1 \\
    - & 1 & - & -1 & - & 1 \\
\hline
\end{tabular}
\caption{The matrix $H^t$ for binary menus.$``-"$ corresponds to zero.}
\label{table:H binary}
\end{table}
Our results demonstrate that one can use this characterization of RUM to obtain testable conditions for DRUM in the form of linear inequalities by computing $(H^1\otimes H^2)\rho\geq 0$. (We later show that,  with more work, we can also obtain the full characterization of DRUM in this environment.) We call these conditions the \emph{dynamic triangle conditions}. A dynamic triangle condition can be expressed recursively in the binary setup as follows:
\begin{align*}
&\mathrm{D}^{\triangle,1}_{z,x,y}((z,\{x,z\};x,\{x,z\}))=\mathrm{D}^{\triangle,2}_{x,z,y}((z,\{x,z\};x,\{x,z\}) \\
&+\mathrm{D}^{\triangle,2}_{x,z,y}((x,\{x,y\};x,\{x,z\})- \mathrm{D}^{\triangle,2}_{x,z,y}((z,\{y,z\};x,\{x,z\})\geq 0,
\end{align*}

where an instance of the dynamic triangle condition is 
\begin{align*}
\mathrm{D}^{\triangle,2}_{x,z,y}((z,\{x,z\};x,\{x,z\}))&= \rho_{\{x,z\},\{x,z\}}(z,x)
+\rho_{\{x,z\},\{y,z\}}(z,z) -\rho_{\{x,z\},\{x,y\}}(z,x)\geq 0.  
\end{align*}
This instance means that, when the choice in period $1$ is held constant, the triangle inequality holds for $\rho$ for a triple of choices in period $2$. Thus, the dynamic triangle condition recursively applies  the triangle inequality in the first period  to the triangle inequality in the second period.
\par 
Our results use the recursive version of the Weyl-Minkowski theorem to show how this insight \textemdash how to derive the $\mathcal{V}$- or $\mathcal{H}$-representation of DRUM from its one-time $\mathcal{V}$- or $\mathcal{H}$-representations\textemdash can be generalized to obtain the characterization of DRUM for any finite time window and any collection of menus (i.e, beyond binary menus), whether the choice set is discrete (e.g., \citealp{li2021axiomatization,chambers2021correlated}) or continuous (e.g., the demand setup of KS).

\subsection{Finite Abstract Setup}\label{subsec: abstract stup}

We consider a nonempty, finite, grand choice set $X^t$ in each $t\in\mathcal{T}$ with an empty (hence, acyclic) order $>^{t}$, and assume that the observed menus in each $t$ are all possible subsets of $X^t$ with cardinality at least $2$. This setup is a generalization of \citet{li2021axiomatization}, which assumes that $|X^t|\leq 3$, and \citet{chambers2021correlated}, which in effect assumes that $T=2$.  
 
\subsection{Demand Setup}\label{subsec: demand setup}
Let $X^*\subseteq \Real^K_{+}$ be the consumption space with finite $K\geq 2$ goods.\footnote{$\Real^K_{+}$ denotes the set of component-wise nonnegative elements of the $K$-dimensional Euclidean space $\Real^K$.} 
In each $t\in \mathcal{T}$, there are $J^t<\infty$ distinct budgets denoted by
\[
B^{*,t}_{j}=\left\{y\in X^*\::\:p_{j,t}\tr y=w_{j,t}\right\}, \quad j\in\mathcal{J}^t=\{1,\dots, J^t\},
\]
where $p_{j,t}\in \Real^K_{++}$ is the vector of prices and $w_{j,t}>0$ is the expenditure level. 

In a way that is similar to the general setup, $\rand{j}=(j_t)_{t\in\mathcal{T}}$, $j_t\in\mathcal{J}^t$ encodes budgets that were faced by DMs in different time periods, referred to as \emph{budget paths}.  Let $\rand{J}$ be a set of all observed budget paths.

For every $\rand{j}\in\rand{J}$, let $\mathrm{P}_{\rand{j}}$ be a probability measure on the set of all Borel measurable subsets of $\times_{t\in\mathcal{T}}X^*$. The primitive in the demand framework is the collection of all observed $\mathrm{P}_{\rand{j}}$, $\mathrm{P}=(\mathrm{P}_{\rand{j}})_{\rand{j}\in\rands{J}}$. We call this collection a \emph{dynamic stochastic demand system}.

Given $\mathrm{P}$, we can define a Dynamic Random Demand Model (DRDM). Let $U^*$ denote the set of all continuous, strictly concave, and monotone utility functions that map $X^*$ to $\Real$; let $\mathcal{U}^*=\times_{t\in\mathcal{T}}U^*$,  and let $u^*=(u^{*t})_{t\in\mathcal{T}}\in\mathcal{U}^*$. 

\begin{definition}[DRDM] 
The dynamic stochastic demand $\mathrm{P}$ is consistent with DRUM if there exists a probability measure over $\mathcal{U}^*$, $\mu^*$, such that
\[
\mathrm{P}_{\rand{j}}\left(\left(O^t\right)_{t\in\mathcal{T}}\right)=\int \prod_{t\in \mathcal{T}} \Char{\argmax_{y\in B_{j_t}^t}u^{*t}(y)\in O^t}d\mu^*(u^*)
\]
for all $\rand{j}\in\rand{J}$ and for all Borel measurable $O^t\subseteq X^*$, $t\in\mathcal{T}$.
\end{definition}

\par
We next show that DRDM is empirically equivalent to DRUM provided that we appropriately specify the choice set and primitive order. 
The monotonicity of the utility functions generates choices on the budget hyperplane. In the RUM demand setting, KS and \citet{kawaguchi2017testing} showed that to establish that $\mathrm{P}$ is consistent with DRUM not all possible Borel sets need to be checked. Stochastic rationalizability by RUM  depends only on the probability of certain regions of the budget hyperplanes called \textit{patches}.  

For any $t\in\mathcal{T}$ and $j\in\mathcal{J}^t$, let $\{x^t_{i|j}\}_{i\in \mathcal{I}^t_j}$, $\mathcal{I}^t_j=\{1,\dots, I^t_{j}\}$, denote a finite partition of $B^{*,t}_{j}$ (where each element of the partition is indexed by $i$).

\begin{definition} [Patches] For every $t\in\mathcal{T}$, let $\bigcup_{j\in \mathcal{J}^t} \{x^t_{i|j}\}$ be the coarsest partition of $\bigcup_{j\in\mathcal{J}^t}B^{*,t}_{j}$ such that
\[
x^t_{i|j}\bigcap B^{*,t}_{j'}\in\{x^t_{i|j},\emptyset\}
\]
for any $j,j'\in \mathcal{J}^t$ and $i\in \mathcal{I}^t_j$. A set $x^t_{i|j}$ is called a patch. If $x^t_{i|j}\subseteq B^{*,t}_{j'}$ for some $i$ and $j\neq j'$, then $x^t_{i|j}$ is called an \emph{intersection patch}.
\end{definition}

By definition, patches can only be strictly above, strictly below, or on budget hyperplanes. A typical patch belongs to one budget hyperplane. However, intersection patches always belong to several budget hyperplanes. The case for one time period, in which $K=2$ goods and $J^t=2$ budgets, is depicted in Figure~\ref{fig: budgets_general}. Note that by definition $\{x^t_{i|j}\}$ is a partition of $B^{*,t}_{j}$, and $I^t_{j}$ is the number of patches that form budget $B^{*,t}_{j}$.

The (discretized) choice set is
\[
X^{t}=\bigcup_{i_t\in I^t_{j},j\in\mathcal{J}^t}\{x_{i_t|j_t}\}.
\] 
The primitive order $>^{t}$ is given by $S'>^{t} S$ for $S,S'\in 2^{X^t}\setminus\{\emptyset\}$ whenever, for any $x_{i|j}\in S$ and $y\in x_{i|j}$, there exist $x_{i'|j'}\in S'$ and $y'\in x_{i'|j'}$ such that  $y'>y$, where $>$ is the strict vector order on $X^*$. 
We define a \emph{menu} as the collection of patches from the same budget hyperplane
\[
B_{j_t}^t=
\left\{x_{i_t|j_t}\right\}_{i_t\in I^t_{j}}.
\]

Henceforth, we refer to a menu or budget interchangeably. 
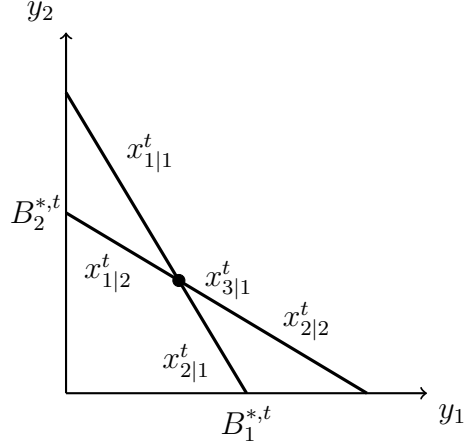
\begin{figure}
\begin{center}
\begin{tikzpicture}[scale=0.8]
\draw[thick,->] (0,0) -- (6,0) node[anchor=north west] {$y_1$};
\draw[thick,->] (0,0) -- (0,6) node[anchor=south east] {$y_2$};
\draw[very thick] (0,5) -- (3,0);
\draw[very thick] (5,0) -- (0,3);
\draw [fill=black] (1.875,1.875) circle[radius=.1];
\draw (2.7,1.875) node {$x^t_{3|1}$};
\draw (0.7,2) node {$x^t_{1|2}$};
\draw (2,0.5) node {$x^t_{2|1}$};
\draw (1.4,4) node {$x^t_{1|1}$};
\draw (4,1.2) node {$x^t_{2|2}$};
\draw (3,-0.5) node {$B^{*,t}_{1}$};
\draw (-0.5,3) node {$B^{*,t}_{2}$};
\end{tikzpicture}
\end{center}
\caption{Patches for the case with $K=2$ goods and $J^t=2$ budgets. The only intersection patch is $x^t_{3|1}$, which is the intersection of $B^t_{1}$ and $B^t_{2}$.}\label{fig: budgets_general}
\end{figure}

Let  
\[
\rho\left(x_{\rand{i}|\rand{j}}\right)=\mathrm{P}_{\rand{j}}\left(x_{\rand{i}|\rand{j}}\right)
\]
denote the fraction of agents who pick from a choice path $x_{\rand{i}|\rand{j}}$ given a budget path $\rand{j}$. 

The main building block of our demand framework is the dynamic stochastic choice function
\[
\rho=\left(\rho\left(x_{\rand{i}|\rand{j}}\right)\right)_{\rand{j}\in\rand{J},\rand{i}\in\rand{I}_\rand{j}}.
\]
The vector $\rho$ represents the distribution over finitely many patches and contains all the necessary information needed to determine whether $\mathrm{P}$ is consistent with DRDM and, in this discretized setup, consistent with DRUM. 
\begin{lemma}\label{lem: demandisdrum}
    The following are equivalent:
    \begin{enumerate}
        \item $P$ is consistent with DRDM.
        \item $\rho$ is consistent with DRUM.
    \end{enumerate}
\end{lemma}
The proof of Lemma~\ref{lem: demandisdrum} follows from \citet{kitamura2018nonparametric} and \citet{kawaguchi2017testing}. Next we provide some parametric examples of DRUM in this domain.
\begin{example}[Dynamic Random Cobb-Douglas Utility] Let $K=2$ and 
    $u^{t}(y_{1},y_{2})=y_{1}^{\alpha_{t}}y_{2}^{(1-\alpha_{t})}$. The utility parameter $\alpha_t$ is random and such that $\alpha_{t}=\max\{\min\{\alpha_{t-1}+\epsilon_{t},1\},0\}$, where $(\epsilon_t)_{t \in \mathcal{T}}$ are independent and identically distributed mean-zero random innovations with variance $\sigma^2$. The dynamic stochastic demand generated by this utility function is consistent with DRUM as long as $(\alpha_t)_{t\in\mathcal{T}}$ is independent of prices and income. 
\end{example}

\begin{example}[Based on \citealp{adams2015prices}] For a deterministic utility $v:X^*\to \Real$, the random utility at time $t\in\mathcal{T}$ is given by $u^{t}(x)=v(x)+\alpha_{t}\tr x$, where $\alpha_t$ is the random vector supported on $\Real^K$. The dynamic stochastic demand generated by this utility function is consistent with DRUM if $\alpha_t$ is independent of prices and income. 
\end{example}
In the two examples above, as well as in the examples in Section~\ref{sec: characterization of DRUM} below, we maintain the assumption that the distribution of preferences does not depend on the budget. This assumption is satisfied in experimental setups such as the ones in \citet{porter2016love}, \citet{mccausland2020testing}, and \citet{ABBK23}. That said, it may not be realistic in other setups, such as when saving is possible. This exogeneity assumption is relaxed in Section~\ref{sec:endogenous}. 

\subsection{Endogenous Expenditure in the Demand Setup}\label{sec:endogenous}
In the demand setup we assumed that budgets are exogenously given. Here, we relax the exogeneity assumption by extending the results of \citet{deb2017revealed} to our setup. Our new model will cover the classical consumption smoothing problem with income uncertainty \citep{browning1989anonparametric}. As we did with DRDM, we can define a Dynamic Random Augmented Demand Model (DRADM). Let $V$ denote the set of all continuous, strictly concave, and monotone augmented utility functions that map $X^*\times \Real_{-}$ to $\Real$, and let $\mathcal{V}=\times_{t\in\mathcal{T}}V$ be the Cartesian product of $T$ repetitions of $V$.

\begin{definition} [DRADM]
A dynamic stochastic demand $\mathrm{P}$ is consistent with DRADM if there exists a probability measure over $\mathcal{V}$, $\eta$, such that
\[
\mathrm{P}_{\rand{j}}\left(\left(O^t\right)_{t\in\mathcal{T}}\right)=\int \prod_{t\in \mathcal{T}} \Char{\argmax_{y\in X^*}v^t(y,-p\tr_{j,t}y)\in O^t}d\eta(v)
\]
for all $\rand{j}\in\rand{J}$ and for all Borel measurable $O^t\subseteq X^*$, $t\in\mathcal{T}$, where $v=(v^t)_{t\in\mathcal{T}}$.
\end{definition}

While DRDM is an extension of RUM to a dynamic setting (i.e., DRDM and RUM coincide when $T=1$), DRADM is a dynamic extension of the Random Augmented Utility Model of \citet{deb2017revealed}.  
\begin{example}[Consumption Smoothing with Income Uncertainty]\label{Example:Consumptionsmoothing}
Consider a consumer with random income stream $y=(y_t)_{t\in\mathcal{T}}$ who maximizes the expected flow of instantaneous, concave, locally nonsatiated, and continuous utilities, $u$, and does so subject to the budget constraints, discount factor $\delta$, history of incomes captured by the information set $I_t$, and the initial level of savings $s_0$. That is, at every time period $\tau$ the consumer solves
\begin{align*}
    \max_{\left\{c_{\tau}(\cdot),s_{\tau}(\cdot)\right\}_{\tau=t,\dots,T}}\Exp{\sum_{\tau=t}^{T}\delta^{\tau-t}u(c_{\tau}(y))\Big|I_\tau}
\end{align*}
subject to
\[
p_{\tau}\tr c_{\tau}(y)+s_t(y)=y_{\tau}+(1+r_{\tau})s_{\tau-1}(y).
\]
The sequences of consumption policy functions $\left(c_{t}(\cdot)\right)_{t\in\mathcal{T}}$ and saving policy functions  $\left(s_{t}(\cdot)\right)_{t\in\mathcal{T}}$ fully describe the consumption and saving decisions, respectively, of the consumer. In addition, we restrict these functions to depend only on the income history. That is, for all $t$ it is the case that $c_t(y')=c_{t}(y)$ and $s_t(y')=s_{t}(y)$ for all $y$ and $y'$ such that $y'_{\tau}=y_{\tau}$ for all $\tau\leq t$. The Bellman equation for this problem is
\[
W_{t-1}(s_{t-1})=\max_{c}\left[u(c)+\delta \Exp{W_t(y_t+(1+r_t)s_{t-1}-p_{t}\tr c)\Big|I_{t}}\right],
\]
where $W_t$ is the value function at time period $t$.
Thus, one can define the state-dependent utility function as
\[
\hat{v}^t(x,s_{t-1})=u(x)+\delta\Exp{W_t(y_t+(1+r_t)s_{t-1}(y)-p_{t}\tr c)\Big|I_{t}}.
\]
Any correlation in income across time would generate a correlation between $\{\hat{v}^t\}_{t\in\mathcal{T}}$. One can define the augmented utility function as
\[
v^t(x,-p\tr x )=u(x)+\delta\Exp{W_t(y_t+(1+r_t)s_{t-1}(y)-p_{t}\tr c)\Big|I_{t}}.
\]
Notice that the utility $\hat{v}^t$ depends on $s_{t-1}$ only through the contemporaneous expenditure $p'x$. If one assumes that different individuals have different $u$, $\delta$, and $y$ such that their joint distribution does not depend on prices, then this setup is a particular case of DRADM.
\end{example}

In Example~\ref{Example:Consumptionsmoothing},  the random augmented utility stochastic process is independent of prices because prices are determined exogenously by supply and demand forces. Next, we characterize DRADM by using the fact that consistency with DRADM is equivalent to consistency with DRDM for a normalized budget path. A normalized budget path has the same price path $(p_{j,t})_{t\in\mathcal{T}}$ and income equal to $1$. Using these normalized budgets, we can define patches as before to obtain
\[
\rho\left(x_{\rand{i}|\rand{j}}\right)=\mathrm{P}_{\rand{j}}\left(\left\{y^t\in X^*\::\:y^t/p\tr_{j,t}y^t\in x_{i|j}^t\right\}_{t\in\mathcal{T}}\right),
\]
for all $\rand{i}\in \rand{I}_{\rand{j}}$, $\rand{j}\in \rand{J}$. 
As before, the stochastic choice function that corresponds to $\mathrm{P}$ is
\[
\rho=\left(\rho(x_{\rand{i},\rand{j}})\right)_{\rand{i}\in \rand{I}_{\rand{j}},\rand{j}\in \rand{J}}.
\]
Similarly to DRDM, we rule out intersection patches.  The choice set $X^t$ and the partial order $>^{t}$ are defined analogously to the definition of DRDM.

\begin{lemma}\label{lem: demand is draum} The following are equivalent:
\begin{enumerate}
    \item $\mathrm{P}$ is consistent with DRADM.
    \item $\rho$ is consistent with DRUM.
\end{enumerate}
\end{lemma}
The proof  of Lemma~\ref{lem: demand is draum} is omitted because it is analogous to the results in \citet{deb2017revealed}.

\section{Characterization of DRUM}\label{sec: characterization of DRUM}
Here we provide a characterization of rationalizability by DRUM when $\rho$ is observed (i.e., estimable). The main result in this section is an analogue of the results in \citet{mcfadden1990stochastic} and KS for RUM. Given the finite choice set, let a preference profile be $\rand{r}=(r^t)_{\in\mathcal{T}}$, where $r^t$ is a linear order defined on the finite set of alternatives available at time $t$, $X^t$. We restrict these linear orders to be extensions of $>^{t}$ (i.e., for $S,S'\in 2^{X^t}\setminus\{\emptyset\}$ and $S>^{t}S'$, there is some $x\in S$ such that $x$ is preferred to $y$,  $x r^t y$, for all $y\in S'$). Recall that $\rand{i}$ encodes choices in each time period. 
Given $\rand{r}$, we can encode choices in different time periods and menus in a vector $a_{\rand{r}}$ as 
\[
a_{\rand{r}}=\left(a_{\rand{r},\rand{i},\rand{j}}\right)_{\rand{j}\in\rand{J},\rand{i}\in\rand{I}_\rand{j}},
\]
with $a_{\rand{r},\rand{i},\rand{j}}=1$ if the alternative $x^t_{i_t|j_t}$ is the best item available in $B^t_{j_t}$ according to $r_t$ for all $t\in\mathcal{T}$, and $a_{\rand{r},\rand{i},\rand{j}}=0$ otherwise. 
Denote $\mathcal{R}^t$ as the set of (strict) rational preferences in a given time period $t\in\mathcal{T}$. The set of dynamic rational preference profiles $\mathcal{R}$ is the set of all preference profiles $\rand{r}$ for which there exists $u_r=(u^t_{r})_{t\in\mathcal{T}}\in\mathcal{U}$ such that
\[
a_{\rand{r},\rand{i},\rand{j}}=1\quad\iff \quad\forall t\in\mathcal{T},\: \argmax_{x\in B^t_{j_t}}u^t_{r}(x)= x_{i_t|j_t}.
\]

We form the matrix $A_T$ by stacking the column vectors $a_{\rand{r}}$ for all preference profiles $\rand{r}\in\mathcal{R}$. The dimension of this matrix is $d_{\rho}\times \abs{\mathcal{R}}$, where $d_{\rho}$ is the length of vector $\rho$. This matrix will be used to provide a characterization of DRUM that is amenable to statistical testing. 

The next axiom is the analogue of the axiom in \citet{mcfadden1990stochastic} for (static) stochastic revealed preferences \citep{border2007introductory} and will provide a different characterization of DRUM.
\begin{definition} [Axiom of Dynamic Stochastic Revealed Preference, ADSRP]  
A stochastic choice function $\rho$ satisfies ADSRP if for every finite sequence of pairs of menu and choice paths (including repetitions), $k$, $\{(\rand{i}_k,\rand{j}_k)\}$ such that $\rand{j}_k\in\rand{J}$ and $\rand{i}_k\in\rand{I}_{\rand{j}_k}$, it follows that 
\[
\sum_{k}\rho\left(x_{\rand{i}_k|\rand{j}_k}\right)\leq \max_{\rand{r}\in\mathcal{R}}\sum_{k}a_{\rand{r},\rand{i}_k,\rand{j}_k}.
\]
\end{definition}

The next theorem provides a full characterization of DRUM. Let 
\[
\Delta^{L}=\left\{y\in\Real_{+}^{L+1}\::\:\sum_{l=1}^{L+1} y_l=1\right\}
\]
denote the $L$-dimensional simplex.

\begin{theorem}\label{thm:main} 
The following are equivalent:
\begin{enumerate}
    \item $\rho$ is consistent with DRUM.
    \item There exists $\nu\in \Delta^{|\mathcal{R}|-1}$ such that $\rho=A_T\nu$.
    \item There exists $\nu\in \Real^{|\mathcal{R}|}_{+}$ such that $\rho=A_T\nu$.
    \item $\rho$ satisfies ADSRP. 
\end{enumerate}
\end{theorem}
The proof of Theorem~\ref{thm:main}  is analogous to the proofs for RUM in \citet{mcfadden1990stochastic,mcfadden2005revealed}, KS, and \citet{kawaguchi2017testing}.  Theorem~\ref{thm:main} (iii) is amenable to statistical testing using the test developed in KS. However, the number of columns in $A_T$ grows exponentially with $T$. Thus, naively, testing DRUM may seem impossible for relatively small $T$ even if one uses the tools of \citet{smeulders2021nonparametric}. The next lemma shows that the computational complexity of computing $A_T$, $T\geq 1$ does not grow that much relative to the computation complexity of computing $A_{1}$.

\begin{lemma}\label{lemma: Kronecker A}
    Let $A^t$ be a matrix constructed under the assumption that $\mathcal{T}=\{t\}$. That is,  $A^t$ is the matrix that encodes static rational types at time $t$. Then $A_T=\otimes_{t\in\mathcal{T}}A^t$ up to a permutation of its rows.
\end{lemma}
\begin{proof}
Note that the $k$-th and the $l$-th columns of $A_1$ and $A_2$, $a^1_{k}$ and $a^2_{l}$, encode the choices of particular types of consumers at time $t=1$ and $t=2$ (i.e., their choices in each menu at $t=1$ and $t=2$). Since there are no restrictions across $t$ on these deterministic types, we can generate the $(k,l)$-type, $a^1_k\otimes a^2_{l}$, that encodes what is picked in pairs of menus such that each menu is taken from two different time periods. Next, if we take some column from $A^3$, we can repeat the above step and obtain a composite type for three time periods. Repeating this exercise $T$ times for all possible combinations of columns will lead to a matrix that is equal to $A_T$ up to a permutation of rows.   
\end{proof}
Lemma~\ref{lemma: Kronecker A} substantially simplifies the computation of $A_T$ given that one can use the methods in KS and \citet{smeulders2021nonparametric} to construct $A^t$. In instances in which the menu structure is such that $A^t=A^{s}$ for $t\neq s$, significant computational savings are achieved. Note that $A^t=A^{s}$ can occur without the observed menus in $t$ and the observed menus in $s$ being the same. In fact, in the demand setting, $A^t$ depends only on the intersection structure induced by the budgets and not on the specific prices (see examples in the next section). Lemma~\ref{lemma: Kronecker A} also allows exploiting sparsity because the Kronecker product propagates any zero entry in $A_t$. The Kronecker product structure of the mixture representation of DRUM shows its structure is modular. Indeed, the structure of $A_T$ is built from its static components. This property allows one to parallelize the computation of $A_T$.\footnote{Since each factor of the Kronecker product can be computed independently, we can parallelize along the time dimension.} This modularity is exploited to obtain a recursive characterization of DRUM.

Unfortunately, the DRUM characterization in Theorem~\ref{thm:main} does not provide an intuitive understanding of the behavioral implications of DRUM. In the next sections, we provide just such an intuitive characterization of DRUM. This characterization demonstrates that DRUM provides additional implications relative to RUM in longitudinal data. That is, we show that requiring consistency with (static) RUM for all conditional and marginal probabilities is not enough to guarantee consistency with DRUM. In fact, the new conditions will affect the joint distribution $\rho$. 

\section{Axiomatic Characterization of DRUM via Linear Inequality Restrictions}\label{sec: behavioral characterization}
In order to understand the axiomatic structure of DRUM, we provide a characterization of many of its special cases via linear inequalities. We also provide a way to obtain a general axiomatic characterization of DRUM via linear inequalities when its static counterpart is known.   First, we need some preliminary mathematical results. 

\subsection*{$\mathcal{H}$- and $\mathcal{V}$-representations}
Theorem~\ref{thm:main} (iii) states that to test whether $\rho$ is consistent with DRUM it is enough to check whether it belongs to the convex cone
\[
\left\{A_Tv\::\:v\geq 0\right\}.
\]
This is called the $\mathcal{V}$-representation of the cone. The Weyl-Minkowski theorem states that there exists an equivalent representation of the cone (the $\mathcal{H}$-representation) via some matrix $B_T$:
\[
\left\{z\::\:H_Tz\geq 0\right\}.
\]
The $\mathcal{V}$-representation of the cone associated with DRUM provides an interpretation of the former as the observed distribution over choices is a finite mixture of deterministic types (KS, \citealp{smeulders2021nonparametric}).\footnote{\citet{kitamura2018nonparametric} were the first to notice, that in the static case, checking whether a stochastic demand is consistent with RUM amounts to checking whether its vector representation belongs to a convex cone. They also introduced the Weyl-Minkowski theorem to the study of RUM in economics.} Unfortunately, the $\mathcal{V}$-representation does not give any direct restrictions on the observed $\rho$. As a result, the analyst can hardly use the $\mathcal{V}$-representation to arrive at any helpful intuition about the empirical content of DRUM. In contrast, the $\mathcal{H}$-representation can deliver direct and sometimes intuitive restrictions on the data (see Section~\ref{sec: 2x2 case} or the dynamic BM inequalities in Section~\ref{sec: BM charachterization}). 
\par
In theory, if one possesses $A_T$, one can obtain $H_T$. Unfortunately, as noted in KS, the construction of $H_T$ from $A_T$ is a nontrivial task that becomes computationally burdensome even for moderate $T$ since the number of columns of $A_T$ grows exponentially with $T$. In Lemma~\ref{lemma: Kronecker A}, we showed that one could use the recursive structure of $A_T$ to simplify its construction substantially. In this section, we show that the same intuition carries over to the construction of $H_T$: one can move from the $\mathcal{H}$-representation of RUM to the $\mathcal{H}$-representation of DRUM with only a small computational costs. Our next result generalizes the Weyl-Minkowski theorem in a direction that is useful for our recursive setup. Henceforth, we assume that all cones are finite-dimensional and are subsets of a Euclidean vector space. 

\begin{proposition}\label{thm:weylmiknowskyrecursive} If
\[
\left\{K^tv\::\:v\geq0\right\}=\left\{z\::\:L^tz\geq0\right\}
\]
for all $t\in\mathcal{T}$, then
\[
\left\{\left(\otimes_{t\in\mathcal{T}}K^t\right)v\::\:v\geq0\right\}\subseteq\left\{z\::\:\left(\otimes_{t\in\mathcal{T}}L^t\right)z\geq0\right\}.
\]
\end{proposition}
Proposition~\ref{thm:weylmiknowskyrecursive} is a direct extension of Theorem~A in \citet{aubrun2021entangleability} and Theorem~$7.15$ in \citet{jossede2020tensor} to more than two time periods.\footnote{We thank Chris Chambers for pointing out an error in the proof of a previous version of this result.} 
Note that if $K^t$ represents a model that can be expressed as mixtures of deterministic behavior (i.e., columns of $K^t$), then Proposition~\ref{thm:weylmiknowskyrecursive} allows one to easily construct testable conditions of the dynamic extensions of this model using its one-time $\mathcal{H}$-representations.
\par 
Next, we provide not only testable conditions but a full equivalence of the recursive $\mathcal{H}$-representation via Kronecker products and the corresponding $\mathcal{V}$-representation. To do this, first define the Kronecker power for a matrix $C$, $C^{\otimes_k}=\otimes_{j=1}^k C$ for any integer $k\geq 1$.\footnote{For $k=0$, the Kronecker power is equal to scalar $1$.} We say that the cone $\left\{Cv\::\:v\geq0\right\}$ is proper if $C$ is full row rank, the cone is closed, and any line in the vector space that contains the cone is not in the cone. For any $\phi^t$ in the interior of $\{L^{t\prime}v\::\: v\geq0\}$ (e.g., the case in which $\phi^t$ is a strict convex combination of columns of $L^{t\prime}$) and any $k\geq1$, define the projection map as follows:
\[
\gamma_{k}^{\phi^t}=\frac{1}{k}\sum_{j=1}^k\phi^{t,\otimes(j-1)}\otimes I^t \otimes \phi^{t,\otimes(k-j)},
\]
where $I^t$ is the identity matrix in the vector space containing the cone $\left\{K^tv\::\:v\geq0\right\}$. Define 
\[
\Gamma^{\rands{\phi}}_{\rand{k}}=I^1\otimes \left(\otimes_{t\in\mathcal{T}\setminus\{1\}}\gamma_{k_t}^{\phi^t}\right)
\]
for a given collection of static operators $\gamma^{\phi^t}_{k_t}$, $t\in\mathcal{T}\setminus\{1\}$. 

\begin{theorem}\label{thm:sufficiencyminimaltensorequalmaximaltensor} 
Suppose $K^t$ is proper for all $t\in\mathcal{T}$. Then
\[
\left\{\left(\otimes_{t\in\mathcal{T}}K^t\right)v\::\:v\geq0\right\}=\bigcap_{k_1=1,k_2\cdots,k_T\geq1}\left\{\Gamma^{\rands{\phi}\prime}_{\rand{k}}z\::\:\left(\otimes_{t\in\mathcal{T}}L^{t,\otimes_{k_t}}\right)z\geq0\right\}.
\]
Moreover,  $K^t$ does not have full column rank for at most one $t\in\mathcal{T}$ if and only if
\[
\left\{\left(\otimes_{t\in\mathcal{T}}K^t\right)v\::\:v\geq0\right\}=\left\{z\::\:\left(\otimes_{t\in\mathcal{T}}L^t\right)z\geq0\right\}.
\]
\end{theorem}
For $T=2$, Theorem~\ref{thm:sufficiencyminimaltensorequalmaximaltensor} is proved in \citet{aubrun2022monogamy}. It easily extends to more than two time periods due to the associativity of the Kronecker product.\footnote{See also Remark 1 on page 9 of \citet{aubrun2022monogamy}.} The ``moreover'' part of our Theorem~\ref{thm:sufficiencyminimaltensorequalmaximaltensor} above is established in Theorem~A  and Corollary~$4$ in \citet{aubrun2021entangleability}. The first part of Theorem~\ref{thm:sufficiencyminimaltensorequalmaximaltensor}  essentially provides an approximation result that allows us to obtain the $\mathcal{H}$-representation of a dynamic model via its one-time counterparts by using extensions of the model to $k_t$ time periods. For some cones the number of extensions can be infinite but in practice we can use a finite number of extensions with the knowledge that in the limit this produces an exact characterization. The simplest case of the previous result happens under additional full column rank requirements. 

\subsection*{$\mathcal{H}$-representation of DRUM}
Proposition~\ref{thm:weylmiknowskyrecursive} gives the necessary conditions for building the $\mathcal{H}$-representation of DRUM from its static components (i.e., where $K^t=A^t$ and $L^t=H^t$, where $H^t$ is the matrix from the $\mathcal{H}$-representation of a cone generated by $A^t$). We show that despite the fact that $A^t$ does not generate a proper cone (because it is never of full row rank), we can use  Theorem~\ref{thm:sufficiencyminimaltensorequalmaximaltensor} to obtain sufficient conditions. The row rank is not full because of the ``adding-up'' constraint\textemdash one alternative has to be picked from \emph{every} menu. Hence, the sum of all rows belonging to the same menu will give the row of ones. In the running example with binary menus, with matrix $A^t$ given by Table~\ref{tab: A_t binary}, the sum of the first two rows is equal to the sum of the third and the fourth rows and is equal to the sum of the last two rows. However, Theorem~\ref{thm:sufficiencyminimaltensorequalmaximaltensor} can still be used to obtain the characterization of DRUM as Theorem~\ref{thm: WM stable rho} demonstrates.   To formalize this, consider the following submatrix of $A^t$, $t\in\mathcal{T}$: from every menu except the first one, pick the last alternative and remove the corresponding row from $A^t$. Let $A^{t*}$ denote the resulting matrix. Observe that when $\mathcal{T}=\{t\}$, $\rho$ is consistent with DRUM if and only if $\rho^*$, defined analogously to $A^{t*}$, is such that $A^{t*}\nu^*=\rho^*$ for some $\nu^*\geq 0$. In particular, given the simplex constraints on a given $\rho$, we can safely drop the adding-up constraints in matrix $H^t$ from now on when computing the $\mathcal{H}$-representation of $A^t$. The reason is that the adding-up constraints are guaranteed to hold. Before stating the next theorem, we need to introduce a key behavioral condition implied by DRUM. 

\begin{definition}[Stability]
We say that $\rho$ is \emph{stable} if 
$
\sum_{i\in\mathcal{I}^t_j}\rho\left(x_{\rand{i}|\rand{j}}\right)
$
is the same for all $j\in\mathcal{J}^t$, for any $t\in\mathcal{T}$ and $x_{\rand{i}|\rand{j}}$.
\end{definition}
Stability means that the marginal distribution of choices at any $t$ does not depend on the menu in any other $t'\neq t$. Under stability, the marginal distribution of choices will not change due to either the menu the consumers faced in the past or the menus the consumers will face in the future.  Recall that we have assumed that the stochastic utility process does not depend on the budgets. This condition is an implication of that assumption.\footnote{Stability and the simplex constraints on $\rho$ are formally defining a vector subspace within which $A^t$ is associated with a proper cone.} 
Stability was first defined in \citet{straleckinotes}. \citet{chambers2021correlated} call this condition  ``marginality'' in their domain. 
\par 
We also need a notion of uniqueness of RUM and DRUM that is directly connected to the requirement of full column rank in Theorem~\ref{thm:sufficiencyminimaltensorequalmaximaltensor}. 
\begin{definition}[Uniqueness] 
We say that $A^t$ generates a unique RUM when the system $\rho=A^t\nu$ has a unique solution for all stochastic choice functions $\rho$. Also, we say that DRUM (associated with matrix $\otimes_{t\in\mathcal{T}}A^t$) satisfies uniqueness when for all  $t\in\mathcal{T}\setminus{\{t'\}}$, $A^t$ generates a unique RUM for a given  period $t'$. 
\end{definition}
\par
Set $\phi^{*,t}$ to be the average of all columns of $H^{t\prime}$. In other words, $\phi^{*,t}$ represents a testable linear inequality of static RUM such that for $\mathcal{T}=\{ t \}$, $\phi^{*,t\prime}\rho\geq 0$. Define  implicitly the associated operator $\Gamma^{\rands{\phi}^*\prime}_{\rand{k}}$ as well.

\begin{theorem}\label{thm: WM stable rho}
Assume that $A^{t*}$ is full row rank for all $t\in \mathcal{T}$. Then $\rho$ is consistent with DRUM if and only if 
 $\rho$ is stable and
\begin{align}\label{eq: entangled}
\rho\in \bigcap_{k_1=1,k_2,\cdots,k_T\geq1}\left\{\Gamma^{\rands{\phi}^*\prime}_{\rand{k}}z\::\:\left(\otimes_{t\in\mathcal{T}}H^{t,\otimes_{k_t}}\right)z\geq0\right\}.    
\end{align} 
Moreover, $\rho$ is consistent with a unique DRUM if and only if $\rho$ is stable and $(\otimes_{t\in\mathcal{T}}H^t)\rho\geq 0$.
\end{theorem}
The condition that $A^{t*}$ is full row rank for all $t\in \mathcal{T}$ is satisfied in all the examples we are aware of. Importantly, the condition holds for the abstract setups of \citet{li2021axiomatization} and of \citet{chambers2021correlated}, as proved by \citet{dogan2022every}. It holds as well in the finite abstract setup with limited menu variation, as proven in \citet{saito2017axiomatizations}. We conjecture that this condition  is true in many other settings and have verified it to be true in the demand setup with up to six budgets per period with a maximal intersection pattern, and with $K \in \{2,3,4,5\}$.\footnote{Note that the fact that the full row rank condition is the only condition that needs to be verified in each case as the cone associated with RUM is closed, and any line in the vector space that contains the cone is not in the cone.} 
\par
Note that stability is a set of equality restrictions on $\rho$. Since any equality restriction can be represented as two inequality restrictions, Theorem~\ref{thm: WM stable rho} allows us to recursively obtain the $\mathcal{H}$-representation of DRUM from the $\mathcal{H}$-representation of its static components for any time window, for the unique DRUM case.  In other words, one just needs to derive the $\mathcal{H}$-representation of a (unique) RUM and then convert it to the dynamic setting and add the constraints implied by stability. This delivers a substantial gain over the direct computation of the $\mathcal{H}$-representation since the existing numerical algorithms transforming $\mathcal{V}$-representations to $\mathcal{H}$-representations are known to work only for small and moderate-size problems. That is, the computational complexity of the dynamic problem is only bounded by the computational complexity of the static one. 

To better interpret our axiomatization of DRUM, note that for $T=2$,
\[
\otimes_{t=1}^T H^t\rho=\text{vec}\left(H^2 R H^{1\prime}\right),
\]
where $\text{vec}(C)$ is the vector obtained by stacking columns of matrix $C$ and $R$ is the matrix such that $\text{vec}\left(R\right)=\rho$. Since every column of $R$ corresponds to a subvector of $\rho$ with a fixed choice in period $t=1$, in the case of unique RUM, $V_2=H^2R$  captures the distributions over static preferences types in $t=2$ conditional on choices in $t=1$. The requirement that 
\[
\otimes_{t=1}^T H^t\rho=\text{vec}(H^2RH^{1\prime})=\text{vec}(V_2H^{1\prime})\geq0
\]
is effectively equivalent to requiring that $V_2\tr$ can be written as a mixture over linear orders at $t=1$. In other words, in the symmetric case where $H^1=H^2$, we just recursively apply the same RUM restrictions to linear combinations of $\rho$. The necessity of these conditions for DRUM is not surprising because its separable structure in time (i.e., DRUM describes DMs maximizing a sequence of random utilities in time). In turns out, these conditions (together with stability) are sufficient when DRUM is unique. 

When DRUM is not unique, then additional restrictions implied by DRUM can emerge. In particular, the separability of DRUM does not translate to its axiomatic structure. The reason is that for unique DRUM, each period RUM (except one) is associated with a convex cone geometrically equivalent to a simplex. The cone associated with the $\mathcal{H}-$representation is, in this case, the Kronecker product of the faces of the different simplices. The vertex of such a cone coincides precisely with the collection of deterministic preference profiles that characterize the unique DRUM.
In contrast, when uniqueness fails for more than one time period, then the cones associated with RUM in each period are no longer equivalent to simplices, in that case, the Kronecker product of the facets of this cone may have vertices that are not consistent with a deterministic preference profile, in addition to the vertex associated with the $\mathcal{V}$-representation of DRUM. We will call these additional vertices \emph{nonrational preference profiles}. DRUM further restricts behavior beyond the recursive application of the static conditions such that we rule out these nonrational preference profiles. The way we can eliminate these additional restrictions is by creating a \emph{thought experiment} where we replicate one time period to obtain additional \emph{experiments} where we hope extensions of the nonrational preference profiles are not longer in the cone produced by the Kronecker product of the facets of the RUM cones. Then we project this extension back to the original time window, and the projection is chosen in such a way as to keep the projection to the interior of the cone. In other words, each of these replications weakens the effect of nonrational preference profiles allowing us to isolate the DRUM types. 

Remarkably, since Theorem~\ref{thm:sufficiencyminimaltensorequalmaximaltensor} provides a necessary and sufficient condition, if DRUM is not unique, then there is no hope that the $\mathcal{H}$-representation of DRUM is just a Kronecker product of the static matrices associated with the $\mathcal{H}$-representation of static RUM. That is why the problem of providing a full characterization of DRUM is a hard problem. 

For the case of the nonunique DRUM, we still obtain the recursive linear inequalities, implied by the Kronecker product of the static RUM inequalities, as necessary conditions for consistency with DRUM. But to obtain a sufficient condition, we have to do more work. Nevertheless, our result provides an explicit way to compute the additional  restrictions on $\rho$ consistent with DRUM. Indeed, applying the results in \citet{aubrun2022monogamy}, we obtain equation/condition~\eqref{eq: entangled} that describes a decreasing sequence of outer approximations to the convex cone associated with DRUM. \citet{doherty2004complete} shows that using these outer approximations can do a good job approximating some cones of interest for finite $k$. 

The general characterization of (static) RUM for our demand setup via $\mathcal{H}$-representation is yet to be discovered \citep{stoye2019revealed}. Only special cases are fully solved: the case of two budgets \citep{hoderlein2014revealed}, and the case of three goods and three budgets (KS). This stands in contrast with the abstract setup  solved in \citet{block1960random} and \citet{falmagne1978representation}. Fortunately, the BM inequalities can be modified in our discretized setup, as we will see below, to deal with the demand setup. Nevertheless, our result implies that once the generic $\mathcal{H}$-representation of RUM becomes available, the analogous DRUM characterization will also become available. 

We want to highlight that our result bears a conceptual resemblance to \citet{debreu1963limit} result on the core convergence to the competitive economy allocation. In that result, as in ours, there is a gap between the core and the competitive economy allocation. Using replicas of the economy, \citet{debreu1963limit} can eliminate that gap. We replicate one of the periods until we eliminate the gap between the axiomatic structure of DRUM and the recursive application of the static axioms of RUM. Similarly to this classical result of \citet{debreu1963limit}, Theorem~\ref{thm: WM stable rho} shows how one can shrink the gap by diminishing the relative importance of the nonrational preference profiles in the cone associated with the Kronecker product of the facets (or axioms) of static RUM.

\subsection*{Back to the Binary Menus Example.}

In our running example, $A^t$ is not full column rank. That means that the dynamic triangle conditions (and stability) are necessary but not sufficient conditions for DRUM.  In our running example, for $T=2$, we will now obtain explicitly all components of equation~\eqref{eq: entangled}, and explain how they provide testable implications for DRUM that become sufficient as we take large enough $k$. For $k=1$, we conclude that if $\rho$ is consistent with DRUM, then it satisfies $\otimes_{t\in \mathcal{T}}H^t\rho\geq 0$ or the dynamic triangle conditions. We then focus on the second extension of the theory for $k=2$. This means that we consider three virtual periods, $T^{v,3}=3$. Let $\rand{J}^v$ be the set of all menu paths in the virtual time window. Define $\rho^v$ as a dynamic stochastic choice function on $\rand{J}^v$.
\par 
To simplify the exposition and to obtain a reduction of dimensionality, we will use the fact that $A^t$, in Table~\ref{tab: A_t binary}, can be simplified without loss of generality by computing $A^{t*}$ as the submatrix with all rows of $A^t$ except for row~$4$ and row~$6$. Due to the simplex constraints of DRUM these rows are redundant.   We can then obtain the $\mathcal{H}$-representation matrix $H^{t*}$ from Table~\ref{table:H binary} by deleting columns~$4$ and $6$, and rows~$2$, $4$, $5$, and $6$. We have to retain the nonnegativity constraints that are not associated with the deleted rows in $A^t$. Since we are imposing the simplex and stability constraints on $\rho$, we can simplify equation~\eqref{eq: entangled} to work with the reduced $H^{t*}$ and the reduced $\rho^*$, of which  the latter is equivalent to the subarray of $\rho$ obtained after we delete all entries with a choice path containing one of the deleted rows of $A^t$. (We define $\rho^{*,v}$ in an analogous way to $\rho^*$.)  
\par
We set $\phi^{*,2}$ as the average of all the triangle conditions in Table~\ref{table:H binary} and the nonnegativity constraints:
\[
\phi^{*,2}=\frac{1}{6}\left(\begin{array}{cccccccccccc}
     2&2&1&1 
\end{array}\right)\tr.
\]
Recall that $I^t$ is a diagonal matrix of dimension $4$ for all $t\in \mathcal{T}$.
The projection mapping $\gamma_{2}^{\phi^{*,2}}$ is given by a matrix of size $4\times 16$ whose rows with nonnegative entries add up to $1$:
\[
\left(
\begin{array}{cccccccccccccccc}
 \frac{1}{3} & \frac{1}{6} & \frac{1}{12} & \frac{1}{12} & \frac{1}{6} & 0 & 0 & 0 &
   \frac{1}{12} & 0 & 0 & 0 & \frac{1}{12} & 0 & 0 & 0 \\
 0 & \frac{1}{6} & 0 & 0 & \frac{1}{6} & \frac{1}{3} & \frac{1}{12} & \frac{1}{12} & 0 &
   \frac{1}{12} & 0 & 0 & 0 & \frac{1}{12} & 0 & 0 \\
 0 & 0 & \frac{1}{6} & 0 & 0 & 0 & \frac{1}{6} & 0 & \frac{1}{6} & \frac{1}{6} &
   \frac{1}{6} & \frac{1}{12} & 0 & 0 & \frac{1}{12} & 0 \\
 0 & 0 & 0 & \frac{1}{6} & 0 & 0 & 0 & \frac{1}{6} & 0 & 0 & 0 & \frac{1}{12} &
   \frac{1}{6} & \frac{1}{6} & \frac{1}{12} & \frac{1}{6} \\
\end{array}
\right).
\] 
The matrix $\gamma_2^{\phi^{*,2}}$  can be interpreted as a linear operator taking weighted averages. In fact, the associated matrix $\Gamma^{\phi^{*,2}\prime}_{2}$
reduces any $\rho^v$ of size $64\times 1$ to a vector of size $16\times 1$, that is the size of $\rho^*$. 
Recall that we defined the matrix $H^t$ by Table~\ref{table:H binary} and nonnegativity constraints. If $\rho$ is consistent with DRUM, then there is some virtual $\rho^v$ such that the following will be satisfied: 
\begin{align*}
    &\rho^*=\Gamma^{\phi^{*,2}\prime}_{2}\rho^{*,v} \\
    &\otimes_{t=1}^3H^{*,t}\rho^{*,v} \geq 0.
\end{align*}
It should be evident that the dynamic triangle inequality conditions are implied by the previous conditions, but new emergent conditions also appear. In particular, we write down explicitly the first entries of the $\rho^*$ vector (with each entry representing the probability of a choice path) in terms of the $\rho^{*,v}$ vector (the remaining entries can be computed easily by the reader):
\begin{align*}
    \rho^*_1&=\frac{\rho_1^{*,v}}{3}+\frac{\rho_2^{*,v}}{6}+\frac{\rho_3^{*,v}}{12}+\frac{\rho_4^{*,v}}{12}+\frac{\rho_5^{*,v}}{6}+\frac{\rho_9^{*,v}}{12}+\frac{\rho_{13}^{*,v}}{12} \\
    \rho^*_2&=\frac{\rho_2^{*,v}}{6}+\frac{\rho_5^{*,v}}{6}+\frac{\rho_6^{*,v}}{3}+\frac{\rho_7^{*,v}}{12}+\frac{\rho_8^{*,v}}{12}+\frac{\rho_{10}^{*,v}}{12}+\frac{\rho_{14}^{*,v}}{12}.
\end{align*}

These two equations illustrate the fact that $\rho_1$ and $\rho_2$ become connected or dependent through $\rho^{*,v}$ entries $\rho_2^{*,v}$ and $\rho_5^{*,v}$. In addition, we can obtain new inequalities explicitly from these relations. Behaviorally, these additional conditions state that when $\rho$ is consistent with DRUM, then $\rho$ must be the result of projecting, back to $T=2$, a $\rho^v$ that is consistent with the dynamic triangle conditions, in a larger time window. The projection mapping $\Gamma^{\phi^{*,2}\prime}_{2}$ can be interpreted as a \textit{projection device} based on the average of testable implications of the static case $\phi^{*,2}$. This device asks what the choices of the virtual DMs associated with $\rho^v$ in the actual time window are. We can make analogous statements for any $k\geq 3$.

\subsection{The Simple Setup: two budgets per time period}\label{sec: 2x2 case}
Here we illustrate our main results in the demand environment with two budgets in each time period $B^{*,t}_{1}$ and $B^{*,t}_{2}$ such that $B^{*,t}_{1}\cap B^{*,t}_{2}\neq \emptyset$ and $w_{1,t}/p_{1,t,K}>w_{2,t}/p_{2,t,K}$ for all $t\in \mathcal{T}$. To simplify the analysis, we assume that the intersection patches are picked with probability zero. Thus, in each time period there are four patches $x^t_{1|1},x^t_{2|1},x^t_{1|2}$, and $x^t_{2|2}$ (see Figure~\ref{fig:simple-setup} for a graphical representation of the case with $K=2$ goods).\footnote{Formally, $x^t_{1|1}=\{y\in B^{*,t}_1\::\:p\tr_{2,t}y>w_{2,t}\}$, $x^t_{2|1}=\{y\in B^{*,t}_1\::\:p\tr_{2,t}y<w_{2,t}\}$, $x^t_{1|2}=\{y\in B^t_2\::\:p\tr_{1,t}y<w_{1,t}\}$, and $x^t_{2|2}=\{y\in B^{*,t}_2\::\:p\tr_{1,t}y>w_{1,t}\}$.} We call choice path configurations implied by these four patches the \emph{simple setup} choice paths. An example of a budget path for $T=2$ is $(2,1)$ (i.e. $B^1_2$ and $B^2_1$), and an example of a choice path in this budget path is $\left(x_{1|2}^{1},x_{1|1}^{2}\right)$. 
\par 
In this setup, there are three rational demand types per time period that are described in Table~\ref{tab:demandtypes1}.\footnote{We use the convenient notation developed in \citet{im2021non}.} Each demand type $\theta^t_{i,j}$ picks the $i$-th patch in menu $B_1^t$ and the $j$-th patch in menu $B^t_2$ at time $t$.
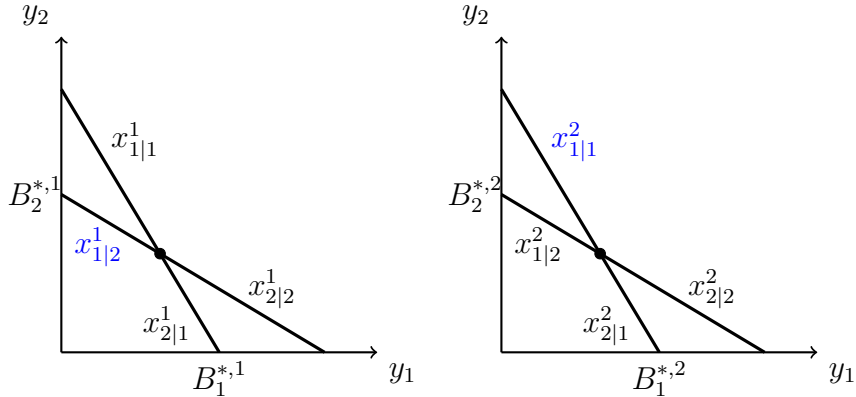
\begin{figure}[h]
\begin{centering}
\begin{tikzpicture}[scale=0.7] 
\draw[thick,->] (0,0) -- (6,0) node[anchor=north west] {$y_1$}; 
\draw[thick,->] (0,0) -- (0,6) node[anchor=south east] {$y_2$}; 
\draw[very thick] (0,5) -- (3,0); 
\draw[very thick] (5,0) -- (0,3); 
\draw [fill=black] (1.875,1.875) circle[radius=.1]; 
\draw (0.7,2) node [color=blue] {$x^1_{1|2}$}; \draw (2,0.5) node {$x^1_{2|1}$}; \draw (1.4,4) node {$x^1_{1|1}$}; \draw (4,1.2) node {$x^1_{2|2}$}; \draw (3,-0.5) node {$B^{*,1}_{1}$}; \draw (-0.5,3) node {$B^{*,1}_{2}$};
\end{tikzpicture}
\begin{tikzpicture}[scale=0.7] 
\draw[thick,->] (0,0) -- (6,0) node[anchor=north west] {$y_1$}; 
\draw[thick,->] (0,0) -- (0,6) node[anchor=south east] {$y_2$}; 
\draw[very thick] (0,5) -- (3,0); 
\draw[very thick] (5,0) -- (0,3); 
\draw [fill=black] (1.875,1.875) circle[radius=.1]; 
\draw (0.7,2) node {$x^2_{1|2}$}; \draw (2,0.5) node {$x^2_{2|1}$}; \draw (1.4,4) node [color=blue] {$x^2_{1|1}$}; \draw (4,1.2) node {$x^2_{2|2}$}; \draw (3,-0.5) node {$B^{*,2}_{1}$}; \draw (-0.5,3) node {$B^{*,2}_{2}$};
\end{tikzpicture}
\par\end{centering}
\caption{Simple setup for $K=2$ goods and no intersection patches. \label{fig:simple-setup}}
\end{figure}
\begin{table}[h]
\begin{centering}
\begin{tabular}{|c|c|c|}
\hline 
 Type/Budget & $B^t_{1}$ & $B^t_{2}$\\
\hline 
$\theta^t_{1,1}$ & $x^t_{1|1}$ & $x^t_{1|2}$\\
\hline 
$\theta^t_{1,2}$ & $x^t_{1|1}$ & $x^t_{2|2}$\\
\hline 
$\theta^t_{2,2}$ & $x^t_{2|1}$ & $x^t_{2|2}$\\
\hline 
\end{tabular}
\par\end{centering}
\caption{Choices of 3 rational types in menus $B^t_1$ and $B^t_2$ at time $t$.}\label{tab:demandtypes1}
\end{table}
\par 
We can now write down the associated $A_T$ matrix. Since there are two intersecting budgets in every time period, $A^t=A^{t'}$ for all $t,t'\in\mathcal{T}$. Thus, by Lemma~\ref{lemma: Kronecker A}, we can compute the matrix $A^t$ for one period. We display the matrix $A^t$ in Table~\ref{tab: A_t}. Note that it is easy to verify that $A^t$ has full column rank; the DRUM associated with $\otimes_{t\in\mathcal{T}}A^t$ is thus unique according to our definition. This allows us to obtain necessary and sufficient conditions explicitly. Using $A^t$, we can write down the matrix $A_T$ for any $T$ (e.g., see Table~\ref{tab: A_T for simplesetup}  for $T=2$ or $A_T=A^1\otimes A^2$).
\begin{table}
\begin{centering}
\scalebox{0.9}{
\begin{tabular}{c!{\vrule width 2pt}c|c|c|c|}
 & $\theta^t_{1,1}$& $\theta^t_{1,2}$ & $\theta^t_{2,2}$\\
\noalign{\hrule height 2pt}
$x^t_{1|1}$ & $1$ & $1$ & - \\
\hline 
$x^t_{2|1}$ & - & - & $1$\\
\hline 
$x^t_{1|2}$ & $1$ & - & - \\
\hline 
$x^t_{2|2}$ & - & $1$ & $1$
\end{tabular}
}
\par\end{centering}
\caption{The matrix $A^t$ for 2 budgets. $``-''$ corresponds to zero.}\label{tab: A_t}
\end{table}
\par 
     
\begin{table}
\begin{centering}
\scalebox{0.6}{
\begin{tabular}{c!{\vrule width 2pt}c|c|c|c|c|c|c|c|c}
 & $(\theta^1_{1,1},\theta^2_{1,1})$ & $(\theta^1_{1,1},\theta^2_{1,2})$ & $(\theta^1_{1,1},\theta^2_{2,2})$ & $(\theta^1_{1,2},\theta^2_{1,1})$ & $(\theta^1_{1,2},\theta^2_{1,2})$ & $(\theta^1_{1,2},\theta^2_{2,2})$ & $(\theta^1_{2,2},\theta^2_{1,1})$ & $(\theta^1_{2,2},\theta^2_{1,2})$ & $(\theta^1_{2,2},\theta^2_{2,2})$\\
\noalign{\hrule height 2pt}
$\left(x^1_{1|1},x^2_{1|1}\right)$ & $1$ & $1$ & - & 1 & 1 & - & - & - & -\\
\hline 
$\left(x^1_{1|1},x^2_{2|1}\right)$ & - & - & 1 & - & - & 1 & - & - & -\\
\hline 
$\left(x^1_{1|1},x^2_{1|2}\right)$ & $1$ & - & - & 1 & - & - & - & - & -\\
\hline 
$\left(x^1_{1|1},x^2_{2|2}\right)$ & - & 1 & 1 & - & 1 & 1 & - & - & -\\
\hline 
$\left(x^1_{2|1},x^2_{1|1}\right)$ & - & - & - & - & - & - & 1 & 1 & -\\
\hline 
$\left(x^1_{2|1},x^2_{2|1}\right)$ & - & - & - & - & - & - & - & - & 1\\
\hline 
$\left(x^1_{2|1},x^2_{1|2}\right)$ & - & - & - & - & - & - & 1 & - & -\\
\hline 
$\left(x^1_{2|1},x^2_{2|2}\right)$ & - & - & - & - & - & - & - & 1 & 1\\
\hline 
$\left(x^1_{1|2},x^2_{1|1}\right)$ & $1$ & 1 & - & - & - & - & - & - & -\\
\hline 
$\left(x^1_{1|2},x^2_{2|1}\right)$ & - & - & 1 & - & - & - & - & - & -\\
\hline 
$\left(x^1_{1|2},x^2_{1|2}\right)$ & 1 & - & - & - & - & - & - & - & -\\
\hline 
$\left(x^1_{1|2},x^2_{2|2}\right)$ & - & 1 & 1 & - & - & - & - & - & -\\
\hline 
$\left(x^1_{2|2},x^2_{1|1}\right)$ & - & - & - & 1 & 1 & - & 1 & 1 & -\\
\hline 
$\left(x^1_{2|2},x^2_{2|1}\right)$ & - & - & - & - & - & 1 & - & - & 1\\
\hline 
$\left(x^1_{2|2},x^2_{1|2}\right)$ & - & - & - & 1 & - & - & 1 & - & -\\
\hline 
$\left(x^1_{2|2},x^2_{2|2}\right)$ & - & - & - & - & 1 & 1 & - & 1 & 1
\end{tabular}
}
\par\end{centering}
\caption{The matrix $A_T$ for 2 time periods with 2 budgets per period. $``-''$ corresponds to zero.}\label{tab: A_T for simplesetup}
\end{table}

\subsubsection*{$\mathrm{D}$-monotonicity}
We introduce a new behavioral restriction on $\rho$ that, together with stability, characterizes DRUM in the simple setup. We first introduce a static notion of \emph{dominance} among patches. 

\begin{definition}[Patch-Revealed Dominance] We say that patch $x^t_{i|j}$ is revealed dominant to $x^t_{i'|j'}$ if $x^t_{i_t|j_t}>^{t} x^t_{i'_t|j'_t}$.
\end{definition}
Patch-revealed dominance requires that all elements in the dominant patch are directly revealed preferred (in Afriat's sense) to the dominated patch, and that all the elements of the dominated patch are not directly revealed preferred to the elements of the dominant patch. We can visualize this ordering in Figure~\ref{fig:simple-setup}, where $x^1_{1|1}>^{t} x^1_{1|2}$ and $x^2_{2|2}>^{t} x^2_{2|1}$. Let $x_{\rand{i}|\rand{j}}\downarrow_{t} x^t_{i'_t|j'_t}$ denote a choice path where the $t$-th component of $x_{\rand{i}|\rand{j}}$ was replaced by $x^t_{i'_t|j'_t}$. We show that if $\rho$ is consistent with DRUM and if $x^t_{i'_t|j'_t}>^{t} x^t_{i_t|j_t}$, then
\[
\rho\left(x_{\rand{i}|\rand{j}}\downarrow_{t} x^t_{i'_t|j'_t}\right)\geq \rho\left(x_{\rand{i}|\rand{j}}\right).
\]
We illustrate the necessity for the simple case where  $T=1$; in that case $A_T\nu=\rho$ can be rewritten as
\begin{align*}
    \nu_1+\nu_2&=\rho(x^1_{1|1}),\quad\quad \nu_3=\rho(x^1_{2|1}),\\
    \nu_1&=\rho(x^1_{1|2}),\quad\quad
    \nu_2+\nu_3=\rho(x^1_{2|2}).
\end{align*}
Then the following two inequalities, consistent with monotonicity, have to be satisfied:
\begin{align*}
    0\leq\nu_2=\rho(x^1_{1|1})-\rho(x^1_{1|2}),\\
    0\leq\nu_2=\rho(x^1_{2|2})-\rho(x^1_{2|1}).
\end{align*}
In fact, we can write down the $\mathcal{H}$-representation of static RUM for the simple setup using matrix $H^{2,t}$ as defined in Table~\ref{table:B2x2}, capturing this monotonicity condition such that $H^{2,t}\rho \geq 0$.
\begin{table}[ht]
\centering
\begin{tabular}{cccc}
\hline
$x^t_{1|1}$ & $x^t_{2|1}$ &  $x^t_{1|2}$ & $x^t_{2|2}$ \\
\hline
1 & - & -1 & -  \\
1 & - & - & - \\
- & 1 & - & - \\
- & - & 1 & - \\
\hline
\end{tabular}
\caption{The $\mathcal{H}$-representation of RUM for $2$ goods and $2$ budgets.}
\label{table:B2x2}
\end{table} 
\par 
For $T\geq2$, DRUM implies that $\rho$ satisfies  dynamic monotonicity. For illustrative purposes, set $T=2$. Then we get the new condition by exploiting the recursive structure of the matrix $A_T$, as shown here:
\[
\rho=A_T\nu=A^1\otimes A^{2}\nu=\left( \begin{array}{ccc}
    A^{1} & A^{1} & 0 \\
    0 & 0 & A^{1}  \\
    A^{1} & 0 & 0 \\  
    0 & A^{1} & A^{1}  
\end{array}
\right) \left( \begin{array}{c}
    \nu^1_{1}\\
    \nu^1_{2}\\
    \nu^1_{3}
\end{array}
\right)=\left( \begin{array}{c}
    A^{1}(\nu^1_{1}+\nu^1_{2})\\
    A^{1}\nu^1_{3}\\
    A^{1}\nu^1_{1}\\
    A^{1}(\nu^1_{2}+\nu^1_{3})
\end{array}
\right).
\]
We can next derive the following system of equations:
\begin{align*}
A^1\nu_1= \left[\rho^1_{1|1}-\rho^1_{1|2}\right],\\
A^1\nu_2= \left[\rho^1_{2|2}-\rho^1_{2|1}\right],
\end{align*}
where $\rho^1_{i|j}$ is a vector of all probabilities that correspond to all choice paths that contain patch $x^1_{i|j}$. For example,
\[
\rho^1_{1|1}=\left(
\begin{array}{c}
     \rho\left(\left(x^1_{1|1}, x^2_{1|1}\right)\right)\\
     \rho\left(\left(x^1_{1|1}, x^2_{2|1}\right)\right)\\
     \rho\left(\left(x^1_{1|1}, x^2_{1|2}\right)\right)\\
     \rho\left(\left(x^1_{1|1}, x^2_{2|2}\right)\right).
\end{array}
\right)
\]
Repeating the argument for  $T=1$, from $A^1\nu_1= \left[\rho^1_{1|1}-\rho^1_{1|2}\right]$, we derive the following:
\begin{align}\label{ineq: 2order mon}
    0\leq\left[\rho\left(\left(x^1_{1|1}, x^2_{1|1}\right)\right)-\rho\left(\left(x^1_{1|2}, x^2_{1|1}\right)\right)\right]-\left[\rho\left(\left(x^1_{1|1}, x^2_{1|2}\right)\right)-\rho\left(\left(x^1_{1|2}, x^2_{1|2}\right)\right)\right].
\end{align}
The following inequalities have to hold under DRUM:
\begin{align*}
\rho\left(\left(x^1_{1|1}, x^2_{1|1}\right)\right)-\rho\left(\left(x^1_{1|2}, x^2_{1|1}\right)\right)\geq 0,\\
\rho\left(\left(x^1_{1|1}, x^2_{1|2}\right)\right)-\rho\left(\left(x^1_{1|2}, x^2_{1|2}\right)\right)\geq 0.
\end{align*}
Thus, Inequality~\eqref{ineq: 2order mon} shown just above imposes a restriction on how the distribution over patches can grow. In particular, it implies that the increase in probability caused by switching from patch $x^1_{1|2}$ to the dominant patch $x^1_{1|1}$ is bigger if the patch in the second time period, $x^2_{1|1}$, dominates $x^2_{1|2}$. In other words, there is some form of complementarity between dominant patches in different time periods.
In order to generalize the above arguments for an arbitrary but finite time window, we would need to work with higher-order differences  (that are in fact the Kronecker products of the static monotonicity conditions). To do this, we next introduce the difference operator.
\begin{definition}[Difference operator]
For any $t$, $x^{t}_{i'_{t}|i'_{t}}$, and $x_{\rand{i}|\rand{j}}$, let $\mathrm{D}\left(x^t_{i'_t|j'_t}\right)[\cdot]$ be a linear operator such that
\[
\mathrm{D}\left(x^t_{i'_t|j'_t}\right)\left[f(x_{\rand{i}|\rand{j}})\right]=f\left( x_{\rand{i}|\rand{j}}\downarrow_{t}x^t_{i'_t|j'_t}\right)-f\left(x_{\rand{i}|\rand{j}}\right)
\]
for any $f$ that maps choice paths to reals.
\end{definition}
The $\mathrm{D}$-operator applied to $\rho$ calculates the difference in $\rho$ when only one patch in a choice path was replaced. When the operator is applied twice to two different time periods, it computes the difference in differences. That is, for $t_1\neq t_2$, we have the following equalities:
\begin{align*}
&\mathrm{D}\left(x^{t_2}_{i'_{t_2}|j'_{t_2}}\right)\mathrm{D}\left(x^{t_1}_{i'_{t_1}|j'_{t_1}}\right)\left[f\left(x_{\rand{i}|\rand{j}}\right)\right]=\mathrm{D}\left(x^{t_2}_{i'_{t_2}|j'_{t_2}}\right)\left[f\left(x_{\rand{i}|\rand{j}}\downarrow_{t_1}x^{t_1}_{i'_{t_1}|j'_{t_1}}\right)-f\left(x_{\rand{i}|\rand{j}}\right)\right]=\\
&\mathrm{D}\left(x^{t_2}_{i'_{t_2}|j'_{t_2}}\right)\left[f\left(x_{\rand{i}|\rand{j}}\downarrow_{t_1}x^{t_1}_{i'_{t_1}|j'_{t_1}}\right)\right]-\mathrm{D}\left(x^{t_2}_{i'_{t_2}|j'_{t_2}}\right)\left[f\left(x_{\rand{i}|\rand{j}}\right)\right]=\\
&\left[f\left(x_{\rand{i}|\rand{j}}\downarrow_{t_1}x^{t_1}_{i'_{t_1}|j'_{t_1}}\downarrow_{t_2}x^{t_2}_{i'_{t_2}|j'_{t_2}}\right)-f\left(x_{\rand{i}|\rand{j}}\downarrow_{t_1}x^{t_1}_{i'_{t_1}|j'_{t_1}}\right)\right]-\left[f\left(x_{\rand{i}|\rand{j}}\downarrow_{t_2}x^{t_2}_{i'_{t_2}|j'_{t_2}}\right)-f\left(x_{\rand{i}|\rand{j}}\right)\right],
\end{align*}
where the second equality uses the linearity of $\mathrm{D}\left(x^{t_2}_{i'_{t_2}|j'_{t_2}}\right)$.

Similarly, we can apply $\mathrm{D}$ any $K$ number of times. Let 
\[
\rands{\mathcal{T}}=\left\{\left(t_k\right)_{k=1}^K\::\:t_k\in\mathcal{T},\:t_k< t_{k+1}, K\leq T\right\}
\]
be a collection of all possible increasing sequences of length at most $T$. For any $\rand{t}\in\rands{\mathcal{T}}$ and any $x^{\rand{t}}_{\rand{i}'|\rand{j}'}=\left(x^t_{i'_t|j'_t}\right)_{t\in \rand{t}}$, we define
\[
\mathrm{D}\left(x^{\rand{t}}_{\rand{i}'|\rand{j}'}\right)\left[f\left(x_{\rand{i}|\rand{j}}\right)\right]=\mathrm{D}\left(x^{t_K}_{i'_{t_K}|j'_{t_K}}\right)\dots\mathrm{D}\left(x^{t_2}_{i'_{t_2}|j'_{t_2}}\right)\mathrm{D}\left(x^{t_1}_{i'_{t_1}|j'_{t_1}}\right)\left[f\left(x_{\rand{i}|\rand{j}}\right)\right],
\]
where $K$ is the length of $\rand{t}$.
\begin{definition}[$\mathrm{D}$-monotonicity]
We say that $\rho$ is $\mathrm{D}$-monotone if for any $\rand{t}\in \rands{\mathcal{T}}$, $x^{\rand{t}}_{\rand{i}'|\rand{j}'}$, and any $x_{\rand{i}|\rand{j}}$ such that $x^{t}_{i'_{t}|j'_{t}}>^{t} x^{t}_{i_{t}|j_{t}}$ for all $t\in\rand{t}$, it is the case that
\[
\mathrm{D}\left(x^{\rand{t}}_{\rand{i}'|\rand{j}'}\right)\left[\rho\left(x_{\rand{i}|\rand{j}}\right)\right]\geq0.
\]
\end{definition}

$\mathrm{D}$-monotonicity is the generalization to our dynamic setup of the Weak Axiom of Stochastic Revealed Preference (WASRP) introduced by \cite{bandyopadhyay1999stochastic}. Since the set of utility functions we are considering is monotone, our condition coincides with the \textit{stochastic substitutability} condition in \citet{bandyopadhyay2004general} when $\mathcal{T} = \{t\}$. Whenever the domain of choices is incomplete, \citet{dasgupta2007regular} shows that WASRP is sufficient but not necessary for regularity. This observation translates to the dynamic case as well.\footnote{In that regard,  $\mathrm{D}$-monotonicity is not implied by any of the conditions derived in \citet{li2021axiomatization} or \citet{chambers2021correlated} that require complete menu variation and use generalizations of the static regularity conditions for the dynamic or correlated case.} We emphasize that $\mathrm{D}$-monotonicity imposes stronger restrictions on the data than WASRP does. We also note that the strength of these restrictions increases as the time window expands (see our Monte Carlo experiments in Appendix~\ref{appendix: montecarlo}). As such, $\mathrm{D}$-monotonicity can be used to derive informative counterfactual bounds on the longitudinal distribution of out-of-sample demand with sufficiently long panels. Given these observations, we are ready to present the main result of this section.
\begin{theorem}\label{thm:DRUM2x2} 
For the simple setup, the following are equivalent:
\begin{enumerate}
    \item $\rho$ is consistent with DRUM. 
    \item $\rho$ is $\mathrm{D}$-monotone and stable. 
    \item $\rho$ is stable and $\otimes_{t\in\mathcal{T}}H^{2,t}\rho\geq0$.
\end{enumerate}
\end{theorem}

The fact that $(i)$ implies $(ii)$ is easy to verify. The converse is proved constructively.  The equivalence of $(iii)$ and $(i)$ is a corollary of Theorem~\ref{thm: WM stable rho}. Importantly, in that case, we can set $k_t=1$ for all $t\in \mathcal{T}$. This is because $A^t$ in the simple setup is full column rank and we can use Theorem~\ref{thm:sufficiencyminimaltensorequalmaximaltensor}. 
\begin{cor}\label{cor:unique}
    For the simple setup if  $\rho=A_T\nu=A_T\nu'$ for some $\nu,\nu'\in\Delta^{|\mathcal{R}-1|}$, then $\nu=\nu'$.
\end{cor}

Stability and $\mathrm{D}$-monotonicity are logically independent as we demonstrate in the next counterexample to DRUM.\footnote{Another example of a $\rho$ that fails both conditions of the simple setup is discussed in Section~\ref{secc:AfriatMcFadden}.}   

\begin{example}[Violation of $\mathrm{D}$-monotonicity] 
Consider the stochastic demand presented in Table~\ref{table:mon+intensity}. It satisfies stability. However, it fails to satisfy $\mathrm{D}$-monotonicity because $\rho\left(\left(x^1_{1|2},x^2_{1|1}\right)\right)-\rho\left(\left(x^1_{1|2},x^2_{1|2}\right)\right)=-\frac{1}{4}$ and $x^2_{1|1}>^{t} x^2_{1|2}$.  
\end{example}

\begin{table}
\begin{centering}
\begin{tabular}{c!{\vrule width 2pt}c|c|c|c}
& $x^2_{1|1}$ & $x^2_{2|1}$ & $x^2_{1|2}$ & $x^2_{2|2}$\\
\noalign{\hrule height 2pt}
$x^1_{1|1}$ & $3/4$ & - & $3/4$ & -\\
\hline 
$x^1_{2|1}$ & - & $1/4$ & $1/4$ & -\\
\hline 
$x^1_{1|2}$ & - & $1/4$ & $1/4$ & \\
\hline 
$x^1_{2|2}$ & $3/4$ & - & $3/4$ & -
\end{tabular}
\par\end{centering}
\caption{Matrix representation of $\rho$ for $T=2$ that violates  $\mathrm{D}$-monotonicity, but satisfies simple stability. ``-'' corresponds to zero}\label{table:mon+intensity}
\end{table}

\subsubsection*{Generalization of $\mathrm{D}$-monotonicity for the demand setup with three goods and three budgets per period}
Our results can be used to construct a set of necessary and sufficient conditions for DRUM in the demand setup on the basis of the $\mathcal{H}$-representation of (static, $\mathcal{T}=\{t\}$) DRDM for the case of three goods and three budgets.\footnote{In each period there are $3$ budgets with maximal intersections as in Example~$3.2$ in KS.} The $\mathcal{V}$-representation in this case is given by matrix $A^t$ in Table~\ref{tab:At3x3}. The $\mathcal{H}$-representation, $H^{3,t}$, is displayed in Table~\ref{table:B3x3} (without the nonnegativity constraints). We note that $H^{3,t}=H^{3,s}$ for any $s,t\in\mathcal{T}$. We can then establish the following direct implication of Theorem~\ref{thm: WM stable rho} since $A^t$ in this case is such that $A^{t*}$ is full row rank.  

\begin{table}[]
    \centering
\scalebox{.9}{
\begin{tabular}{ccccccccccccccccccccccccc|c}
\hline
     1&1&1&1&1&1&1&1&1&1&1&1&-&-&-&-&-&-&-&-&-&-&-&-&-&$x^t_{1|1}$\\
     -&-&-&-&-&-&-&-&-&-&-&-&1&1&1&1&1&-&-&-&-&-&-&-&-&$x^t_{2|1}$\\
     -&-&-&-&-&-&-&-&-&-&-&-&-&-&-&-&-&1&1&1&1&1&-&-&-&$x^t_{3|1}$\\
     -&-&-&-&-&-&-&-&-&-&-&-&-&-&-&-&-&-&-&-&-&-&1&1&1&$x^t_{4|1}$\\
     1&1&1&1&-&-&-&-&-&-&-&-&1&1&1&1&-&1&1&-&-&-&1&1&-&$x^t_{1|2}$\\
     -&-&-&-&1&1&1&1&-&-&-&-&-&-&-&-&-&-&-&1&-&-&-&-&-&$x^t_{2|2}$\\
     -&-&-&-&-&-&-&-&1&1&-&-&-&-&-&-&1&-&-&-&1&-&-&-&1&$x^t_{3|2}$\\
     -&-&-&-&-&-&-&-&-&-&1&1&-&-&-&-&-&-&-&-&-&1&-&-&-&$x^t_{4|2}$\\
     1&-&-&-&1&-&-&-&1&-&1&-&1&-&-&-&1&1&-&1&1&1&1&-&1&$x^t_{1|3}$\\
     -&1&-&-&-&1&-&-&-&-&-&-&-&1&-&-&-&-&1&-&-&-&-&1&-&$x^t_{2|3}$\\
     -&-&1&-&-&-&1&-&-&1&-&1&-&-&1&-&-&-&-&-&-&-&-&-&-&$x^t_{3|3}$\\
     -&-&-&1&-&-&-&1&-&-&-&-&-&-&-&1&-&-&-&-&-&-&-&-&-&$x^t_{4|3}$\\
\hline
\end{tabular}
}
    \caption{$A^t$ for $3$ goods and $3$ budgets. ``-'' corresponds to zero.}
    \label{tab:At3x3}
\end{table}

\begin{table}[ht]
\centering
\begin{tabular}{cccccccccccc}
\hline
$x^t_{1|1}$ & $x^t_{2|1}$ & $x^t_{3|1}$ & $x^t_{4|1}$ & $x^t_{1|2}$ & $x^t_{2|2}$ & $x^t_{3|2}$ & $x^t_{4|2}$ & $x^t_{1|3}$ & $x^t_{2|3}$ & $x^t_{3|3}$ & $x^t_{4|3}$\\
\hline
- & - & - & -1 & - & - & - & -1 & 1 & 1 & 1& - \\
- & - & - & -1 & 1 & - & - & - & 1 & - & - & - \\
1 & - & - & - & 1 & - & - & - & - & - & - & -1 \\
1 & - & - & - & - & - & - & -1 & 1 & - & - & - \\
- & -1 & - & -1 & 1 & - & 1 & - & - & - & - & - \\
- & - & - & - & 1 & 1 & - & - & - & -1 & - & -1\\
- & - & -1 & -1 & - & - & - & - & 1 & 1 & - & -\\
\hline
\end{tabular}
\caption{The $\mathcal{H}$-representation of RUM for $3$ goods and $3$ budgets excluding nonnegativity.}
\label{table:B3x3}
\end{table}

\begin{cor} For the demand setup  ($K=3$,  $J^t=3$ for all $t\in\mathcal{T}$), the following are equivalent: 
\begin{enumerate}
    \item $\rho$ is consistent with DRUM. 
    \item $\rho$ is stable and 
    \[
    \rho\in \bigcap_{k_1,\cdots,k_T\geq1}\left\{\Gamma^{\rands{\phi}^*\prime}_{\rand{k}}z\::\:\left(\otimes_{t\in\mathcal{T}}H^{3,t,\otimes_{k_t}}\right)z\geq0\right\}.
    \]   
\end{enumerate}
\end{cor}
For $k=1$ we obtain necessary conditions in the form of $(\otimes_{t\in \mathcal{T}}H^{3,t})\rho\geq 0$. The last three rows of the matrix displayed in Table~\ref{table:B3x3} and their Kronecker product in time are consistent with $\mathrm{D}$-monotonicity.\footnote{Note that if $\mathrm{P}=(\mathrm{P}_{\rand{j}})_{\rand{j}\in\rand{J}}$ is consistent with DRDM, then the stochastic demand system consisting of 2 different budget paths $(\mathrm{P}_{\rand{j}},\mathrm{P}_{\rand{j}'})$ would also be consistent with DRDM. Moreover, note that $(\mathrm{P}_{\rand{j}},\mathrm{P}_{\rand{j}'})$ form the simple setup since at every time period there are exactly two budgets. Thus, $\mathds{D}$-monotonicity is a simple necessary condition for DRUM.} The rest of the conditions produce new testable implications. For instance, for $T=2$, the following conditions are implied by monotonicity in row~$4$:
\begin{align*}
    D^{*1}(x_{i|j}^2)=\left[\rho\left(\left(x^1_{1|2}, x_{i|j}^2\right)\right)+ \rho\left(\left(x^1_{3|2}, x_{i|j}^2\right)\right)
    -\rho\left(\left(x^1_{2|1}, x_{i|j}^2\right)\right) -\rho\left(\left(x^1_{4|1}, x_{i|j}^2\right)\right) \right]\geq 0.
\end{align*}
The interaction of row~$3$, a triangle condition, and row~$4$, a monotonicity condition, gives the implication:
\begin{align*}
    D^{*1}(x_{1|1}^2)+D^{*1}(x_{1|2}^2)-D^{*1}(x_{4|3}^2)\geq0.
\end{align*}

\par
Obtaining sufficient conditions in terms of an $H^{3,t}$ matrix would require more work because $A^t$ is no longer full column rank in this case. Nevertheless, we can obtain sufficient conditions explicitly from the knowledge of $H^{3,t}$ alone. We set $\phi^{*,t}$ as the average of all conditions implied by $H^{3,t}$. The rest of the computations for any $k\geq2$ can be done in the same way as in the running example for binary menus.

\subsection{Characterization of DRUM via Recursive Block Marschak Inequalities}\label{sec: BM charachterization}
Here, we provide a characterization, via linear inequalities, based on the BM inequalities for the case of limited menu variability and with monotonicity of utilities.  

Recall that $\mathcal{J}^t$ denotes the set of observed menus at time $t$.  Let $\bar{\mathcal{J}}^t=\{1,2,\dots,2^{\abs{X^t}}-1\}$ be the ``virtual'' set of menus such that there is a one-to-one mapping between $j_t\in \bar{\mathcal{J}}^t$ and a nonempty subset of $X^{t}$ (i.e., the virtual data set has full menu variation). We also assume that this mapping is such that the first $\mathcal{J}^t$ indexes correspond to observed menus $B^t_j$. That is $\mathcal{J}^t$ is the set of all observed menus and $\bar{\mathcal{J}}^t\setminus \mathcal{J}^t$ is the set of all ``virtual menus'' that were not observed in the data. Using this extended definition of menus, we can define a set of all menu paths $\bar{\rand{J}}$ (including the observed ones). Note that the set of observed menu paths $\rand{J}$ is a subset of $\bar{\rand{J}}$. Finally, as before, given a menu path $\rand{j}\in\bar{\rand{J}}$, let $\bar{\rho}$ be a distribution over choice paths $\rand{i}$ in $\rand{j}$. That is, $\bar{\rho}\left(x_{\rand{i}|\rand{j}}\right)\geq0$ and $\sum_{\rand{i}\in\rand{I}_{\rand{j}}}\bar{\rho}\left(x_{\rand{i}|\rand{j}}\right)=1$ for all $\rand{j}\in\bar{\rand{J}}$. Recall that the observed $\rho$ is defined as
\[
\rho=\left(\rho\left(x_{\rand{i}|\rand{j}}\right)\right)_{\rand{j}\in\rand{J},\rand{i}\in\rand{I}_{\rand{j}}}.
\]
Similarly, the extended stochastic choice function is defined by
\[
\bar{\rho}=\left(\bar{\rho}\left(x_{\rand{i}|\rand{j}}\right)\right)_{\rand{j}\in\bar{\rand{J}},\rand{i}\in\rand{I}_{\rand{j}}}.
\]

Next we define some properties of $\bar{\rho}$. 

\begin{definition}
We say that $\bar{\rho}$ \emph{agrees} with $\rho$ if they coincide on observed menu paths. That is, $\bar{\rho}\left(x_{\rand{i}|\rand{j}}\right)=\rho\left(x_{\rand{i}|\rand{j}}\right)$ for all $\rand{j}\in\rand{J}$ and $\rand{i}\in\rand{I}_{\rand{j}}$.
\end{definition}
If  $\bar{\rho}$ agrees with $\rho$, then it is an extension of the latter to virtual menu paths.

\begin{definition}[Increasing Utility (IU) Consistency]
We say that $\bar{\rho}$ is \emph{IU-consistent} if $\bar{\rho}\left(x_{\rand{i}|\rand{j}}\right)=0$ whenever there exists $t\in\mathcal{T}$ and $I'_t\subseteq\mathcal{I}^t_{j_t}$ such that $\cup_{i'_t\in I'_t}\{x_{i'_t|j_t}\}>^{t} \{x_{i_t|j_t}\}$.
\end{definition}
IU-consistency captures the empirical content of the monotonicity of the utility functions with respect to $>^{t}$.

\begin{definition}[BM inequalities] We say that  $\bar{\rho}$ satisfies the BM inequalities if for all $t\in\mathcal{T}$, $\rand{j}\in\bar{\rand{J}}$, and $\rand{i}\in\rand{I}_{\rand{j}}$
\[
\mathds{B}^t(\rand{i},\rand{j})=\sum_{j_t':B^t_{j_t}\subseteq B^t_{j_t'}}(-1)^{\abs{B^t_{j_t'}\setminus B^t_{j_t}}}\bar{\rho}\left(x_{\rand{i}|\rand{j'}}\right)\geq 0.
\]
\end{definition}
Note that the BM inequalities are linear in $\bar{\rho}$. Hence, we can construct matrix $\bar{H}^t$ with elements in $\{-1,0,1\}$ such that each row of $\bar{H}^t$ corresponds to one BM inequality.
\par 
We are now ready to state the two main results of this section: first, a BM characterization of RUM for our setup, and second, an analogous characterization for DRUM.

\begin{theorem}\label{thm: static RUM-BM}
Let $\mathcal{T}=\{t\}$. For a given $\rho$, the following are equivalent:
\begin{enumerate}
    \item $\rho$ is consistent with RUM.
    \item There exists $\bar{\rho}$ that agrees with $\rho$, is IU-consistent, and satisfies the BM inequalities. 
    \item There exists $\bar{\rho}$ that agrees with $\rho$, is IU-consistent, and is such that $\bar{H}^t\bar{\rho}\geq0$.
\end{enumerate}
\end{theorem}

The BM inequalities provide the $\mathcal{H}$-representation of RUM for the static case. For this extended setup with $\bar{\rho}$, we say that $\bar{A}^t$ generates a unique RUM when the system $\bar{\rho}=\bar{A}^t\nu$ has a unique solution for all choice functions $\bar{\rho}$. We also say that DRUM (associated to matrix $\otimes_{t\in\mathcal{T}}\bar{A}^t$) satisfies uniqueness when, for all  $t\in\mathcal{T}\setminus{\{t'\}}$, $\bar{A}^t$ generates a unique RUM for an arbitrary period $t'$. We can obtain the modified $\mathcal{H}$-representation, $\bar{H}^{t}$, for the unique RUM for $\mathcal{T}=\{t\}$ from \citet{turansick2022identification}. The rows of $\bar{H}^{t}$correspond to: (i) BM inequalities, (ii) nonnegativity constraints, and (iii) the uniqueness restriction in Theorem~$1$ in \citet{turansick2022identification}.\footnote{The conditions in \citet{turansick2022identification} are equivalent to having some BM polynomials to be equal to zero. This can be expressed in our setup by adding to the $\mathcal{H}$-representation the row corresponding to the relevant BM polynomial and with the same row multiplied by $-1$. This guarantees that the relevant BM polynomial will be zero.}  

\begin{theorem}\label{thm: unique DRUM-BM}
    The following are equivalent:
   \begin{enumerate}
        \item $\rho$ is consistent with DRUM satisfying uniqueness. 
        \item  $\rho$ is IU-consistent, stable, and satisfies  $(\otimes_{t\in \mathcal{T}}\bar{H}^{t})\bar{\rho}\geq 0$. 
    \end{enumerate}
\end{theorem}
The result is a direct consequence of Theorem~\ref{thm: WM stable rho} and Theorem~\ref{thm: static RUM-BM}. Theorem~\ref{thm: unique DRUM-BM} is not exactly a $\mathcal{H}$-representation of DRUM; however, it becomes one when all menus in $\bar{\rand{J}}$ are observed, as is assumed in \citet{chambers2021correlated} and in \citet{li2021axiomatization}.  Note, moreover, that our proof can be used in the environments of \citet{chambers2021correlated} and \citet{li2021axiomatization} in the finite abstract setting because the primitive order $>^{t}$ can be empty.  Notice that the results in \citet{li2021axiomatization} are a special case of Theorem~\ref{thm: unique DRUM-BM} because when the choice set is of cardinality $3$ or less, RUM is unique \citep{fishburn1998stochastic,turansick2022identification}. Our result also generalizes Theorem~$3$ in \citet{chambers2021correlated}, which is equivalent to the special case $T=2$.  
\par 
We generalize the BM inequalities for the case of unobserved menus. Even if, for our primitive, this recursive characterization of DRUM is not an $\mathcal{H}$-representation of DRUM, this characterization has several advantages: (i) it avoids the computation of matrix $A_T$ associated with the $\mathcal{V}$-representation, which can be computationally burdensome (\citealp{kitamura2019nonparametric} note that computing $A^t$ is NP hard); (ii) it provides further intuition about the additional empirical bite of DRUM in comparison to RUM; and (iii) it means DRUM can be tested with a linear program.\footnote{For the statistical problem, we can use tools in KS and \citet{fang2023inference}.}
\par
Note that one cannot weaken the uniqueness assumption at all. In fact, the result fails to be true when $A^t$ is associated with nonunique static RUM for more than two periods. In that case, we have to go back to Theorem~\ref{thm: DRUM-BM}. 
\par 
We set $\phi^{*,t}$ for any $t\in\mathcal{T}$ as the average of rows of $\bar{H}^t$. The typical entry of $\bar{H}^t$ corresponding to the BM inequalities is: 
\[
\bar{H}^t_{(i_t,j_t),(i'_t,j'_t)}=(-1)^{|B^t_{j'_t}\setminus B^t_{j_t}|}\Char{x_{i_t|j_t}=x_{i'_t|j'_t},B^t_{j_t}\subseteq B^t_{j'_t}, x_{i_t|j'_t}\in B^t_{j'_t}}.
\]
Hence, $\phi^{*,t}_{i_t\prime,j_t\prime}=\sum_{i_t,j_t}\bar{H}^t_{(i_t,j_t),(i'_t,j'_t)}$.  As a direct application of Theorem~\ref{thm: WM stable rho}, we can now establish the following result.

\begin{theorem}\label{thm: DRUM-BM}
    The following are equivalent:
   \begin{enumerate}
        \item $\rho$ is consistent with DRUM. 
        \item There exists $\bar{\rho}$ that agrees with $\rho$, is IU-consistent, stable, and satisfies  
        \[
        \bar{\rho}\in \bigcap_{k_1,\cdots,k_T\geq1}\left\{\Gamma^{\rands{\phi}^*\prime}_{\rand{k}}z\::\:\left(\otimes_{t\in\mathcal{T}}\bar{H}^{t,\otimes_{k_t}}\right)z\geq0\right\}.
        \]
    \end{enumerate}
\end{theorem} 

To understand the intuition behind the $\mathcal{H}$-representation of DRUM, we focus on a necessary condition that it implies. We define a new set of inequalities we call DRUM-BM. 
\begin{definition}[DRUM-BM inequalities] We say that $\bar{\rho}$ satisfies the DRUM-BM inequalities if, for all $t\in\mathcal{T}$, $\rand{j}\in\bar{\rand{J}}$, and $\rand{i}\in\rand{I}_{\rand{j}}$, we have that $\mathds{B}_{t}(\rand{i},\rand{j})\geq0$, where $\mathds{B}_{T}(\rand{i},\rand{j})=\mathds{B}^T(\rand{i},\rand{j})$ and
\[
\mathds{B}_{t}(\rand{i},\rand{j})=
\sum_{j_{t}':B^t_{j_t}\subseteq B^t_{j_t'}}(-1)^{\abs{B^t_{j'_t}\setminus B^t_{j_t}}}\mathds{B}_{t+1}(\rand{i},\rand{j'})
\]
for all $t\in\mathcal{T}\setminus\{T\}$.
\end{definition}
Now we establish the following result. 
\begin{cor}
    If $\rho$ is consistent with DRUM, then $\bar{\rho}$ satisfies the DRUM-BM inequalities. 
\end{cor}
The DRUM-BM inequalities are as intuitive as the BM inequalities for the static case but they are not enough for DRUM. This is because these conditions interact with the hierarchical theory extensions for the nonunique DRUM case, giving rise to emergent conditions.

\section{Relationship with the Samuelson-Afriat and the McFadden-Richter frameworks} \label{secc:AfriatMcFadden}
In this section, we study the implications of DRUM for time series and cross-sections. First, we look at a time series with the assumption of constant utility across time periods as in the Samuelson-Afriat framework. In this case, DRUM implies that the (deterministic) Strong Axiom of Revealed Preference (SARP) has to hold in time series. Second, we study cross-sections, like the ones described in \citet{mcfadden1990stochastic} and \citet{mcfadden2005revealed} that are obtained by \emph{marginalizing} or \emph{pooling} panels. Marginalization and pooling correspond to the empirical practices of using cross-sections that correspond to one or many time periods, respectively. We show that if $\rho$ is consistent with DRUM, then any marginal distribution derived from it is rationalizable by RUM. At the same time, not every DRUM-consistent panel is RUM-rationalizable when pooled. Importantly, marginal consistency with RUM is not sufficient for consistency with DRUM.  

\subsection{The Samuelson-Afriat framework}
DRUM has no testable implications for a time series without further restrictions. That is, if we observe $\rho_{\rand{j}}$ for a single budget path $\rand{j}$, then there are no testable restrictions of DRUM. (We need at least two observed budget paths to test DRUM.) However, in the Samuelson-Afriat framework, one needs only a time-series of choices from budgets in order to test utility maximization. The reason for this is that the Samuelson-Afriat framework makes an additional assumption on the stochastic process, namely, that $\mu$ is such that $u^t=u^s$ $\mu-\as$ for all $t,s\in\mathcal{T}$. We call this restriction \emph{constancy} of the stochastic utility process. Under this restriction, the testable implications of DRUM in a time series are re-established. We need some preliminaries to formalize this intuition. To simplify the exposition, all the results in this section are for the demand setup.

\begin{definition}[Strong Axiom of Revealed Path Dominance, SARPD] For a given $\rand{j}\in\rand{J}$, we say that $\rho_{\rand{j}}=(\rho(x_{\rand{i}|\rand{j}}))_{\rand{i}\in\rand{I}_{\rand{j}}}$ satisfies SARPD if 
\[
\rho\left(x_{\rand{i}|\rand{j}}\right)=0
\]
whenever there is a finite set of patches from $x_{\rand{i}|\rand{j}}$, $\left\{x^{t_n}_{i_{t_n}|j_{t_n}}\right\}_{n=1}^N$, such that
$x^{t_1}_{i_{t_1}|j_{t_1}}\succeq^* x^{t_2}_{i_{t_2}|j_{t_2}}\succeq^*\dots\succeq^* x^{t_N}_{i_{t_N}|j_{t_N}}$ and $x^{t_N}_{i_{t_N}|j_{t_N}}\succeq^* x^{t_1}_{i_{t_1}|j_{t_1}}$(where  $x^t_{i_t|j_t}\succeq^* x^s_{i_s|j_s}$ whenever there are $x\in x^t_{i_t|j_t}$ and $y\in x^s_{i_s|j_s}$ such that $p_{j_t}'(x-y)\geq 0$).
\end{definition}
SARPD requires that the probability of observing a choice path that contains consumption bundles that form a revealed preference cycle is zero.  It is  analogous to the Strong Axiom of Revealed Preferences (SARP) in the Samuelson-Afriat framework. Using SARPD, we can establish the following result:

\begin{proposition}\label{prop:constantDRUMimpliesSARP} 
If $\rho$ is rationalized by DRUM with $\mu$ that satisfies constancy, then $\rho_{\rand{j}}$ satisfies SARPD for all $\rand{j}\in\rand{J}$.
\end{proposition}
Note that DRUM bounds above the probability of choice paths that contain a revealed preference cycle. To see this, we consider again the simple setup with $T=2$. There are two choice paths that contain a revealed preference cycle: $\left(x_{1|2}^1,x_{2|1}^2\right)$ and $\left(x_{2|1}^1,x_{1|2}^2\right)$. We focus on the first choice path without loss of generality. Using $\mathrm{D}$-monotonicity, we know that 
\[
\rho\left(\left(x_{1|2}^1,x_{2|1}^2\right)\right)\leq \rho\left(\left(x_{1|1}^1,x_{2|2}^2\right)\right).
\]
This means that the probability of a choice path that contains a  violation of SARP (i.e., it contains a revealed preference cycle), is bounded above by the probability of a choice path that contains no such cycles. That is, DRUM meaningfully restricts the probability of choice paths with revealed preference cycles. This \emph{endogenous} bound on the probability of a choice path that contains a revealed preference cycle has an important advantage with respect to measures of deviations from rationality like the Critical Cost Efficiency Index \citep{afriat1973efficiency}. Indeed, in the literature that uses this index, it is an open question how to set a threshold below which the level of deviations from static utility maximization is deemed reasonable. In our setup, we convert this problem into a population one and then bound endogenously the fraction of consumers  that have choices that involve revealed preference cycles.  Importantly, notice that if $\rho$ is degenerate taking values on $\{0,1\}$ for a given budget path then $\mathrm{D}$-monotonicity is equivalent to the Weak Axiom of Revealed Preference by \citet{samuelson1938note} in the Samuelson-Afriat framework. To see this, note that the probability of choice paths with revealed preference cycles of size $2$ (i.e., violations of the weak axiom) under the degeneracy of $\rho$ must be zero. 

\subsection{Marginal and Conditional Distributions}
Given a budget path $\rand{j}$, let $\rho^\textrm{c}_{t,\rand{j}}$ and $\rho^\textrm{m}_{t,\rand{j}}$ be the conditional and the marginal distributions, respectively, over patches implied by $\rho_{\rand{j}}$. That is,
\begin{align*}
    \rho^\textrm{c}_{t,\rand{j}}\left(x_{\rand{i}|\rand{j}}\right)&=\dfrac{\rho\left(x_{\rand{i}|\rand{j}}\right)}{\sum_{i\in\mathcal{I}^t_{j_t}}\rho\left(x_{\rand{i}|\rand{j}}\right)},\\
    \rho^\textrm{m}_{t,\rand{j}}\left(x_{i_t|j_t}\right)&=\sum_{\tau\in\mathcal{T}\setminus\{t\}}\sum_{i\in\mathcal{I}^{\tau}_{j_\tau}}\rho\left(x_{\rand{i}|\rand{j}}\right),
\end{align*}
where the conditional distribution is defined only when $\sum_{i\in\mathcal{I}^t_{j_t}}\rho\left(x_{i_t|j_t}\right)\neq 0$. Given the marginal distribution of a budget path, we can also define the slicing distribution as
\[
\rho^\textrm{s}_{t}\left(x_{i_t|j_t}\right)=\sum_{\rand{j}\in\rand{J}}\rho^\textrm{m}_{t,\rand{j}}\left(x_{i_t|j_t}\right)F(\rand{j}|j_t),
\]
where $F(\rand{j}|j_t)$ is the conditional probability of observing the budget path $\rand{j}$ conditional on the $t$-th budget being $j_t$ in the data. The slicing distribution is a mixture of marginal distributions. It corresponds to the situation in which the researcher only focuses on one cross-section. 

\begin{proposition}\label{prop:DRUMimpliesmarginalRUM}
If $\rho$ is rationalized by DRUM, then $\rho^\textrm{c}_{t,\rand{j}}$, $\rho^\textrm{m}_{t,\rand{j}}$, and $\rho^\textrm{s}_{t}$ are rationalized by RUM for any $t\in\mathcal{T}$ and $\rand{j}\in\rand{J}$.
\end{proposition}

Proposition~\ref{prop:DRUMimpliesmarginalRUM} means that if $\rho$ is consistent with DRUM, then the data are consistent with RUM in any given cross-section (i.e., slice). In this sense, the empirical implications of DRUM when an analyst has access to only a slice of choices is the same as the empirical implications of RUM. However, consistency of the marginal or slicing distributions does not exhaust the empirical content of DRUM. This is illustrated in Example~\ref{ex:RUMmarginalisnotDRUM}.  

\begin{example}\label{ex:RUMmarginalisnotDRUM}[Marginals are consistent with WASRP but not rationalized by DRUM]
Consider $\rho$ as presented in Table~\ref{tab:marginalRumNotDrum}. This $\rho$ violates stability and $\mathrm{D}$-monotonicity. So DRUM cannot possibly explain it. At the same time its marginal probabilities at $t=1$ are consistent with the WASRP: $\rho^\textrm{m}_{1,(2,1)}\left(x^1_{1|2}\right)=\frac{1}{2}$, $\rho^\textrm{m}_{1,(1,1)}\left(x^1_{2|1}\right)=\frac{1}{2}$; and $\rho^\textrm{m}_{1,(2,2)}\left(x^1_{1|2}\right)=\frac{2}{3}$, $\rho^\textrm{m}_{1,(1,2)}\left(x^1_{2|1}\right)=\frac{1}{3}$. Thus, each of these marginal distributions is consistent with RUM.\footnote{Recall that WASRP is the necessary and sufficient condition for marginal probabilities to be rationalized by RUM in the sense of Proposition~\ref{prop:DRUMimpliesmarginalRUM}.} Moreover, the slicing distribution would satisfy $\rho^\textrm{s}_{1}\left(x^1_{2|1}\right)=F((1,1)|1)\frac{1}{2}+F((1,2)|1)\frac{1}{3}$ and $\rho^\textrm{s}_{1}\left(x^1_{1|2}\right)=F((2,1)|2)\frac{1}{2}+F((2,2)|2)\frac{2}{3}$. As a result, depending on $F$,
\[
\rho^\textrm{s}_{1}\left(x^1_{1|2}\right)+\rho^\textrm{s}_{1}\left(x^1_{2|1}\right)\in\left[5/6,7/6\right].
\]
Thus, if, for example, all budget paths are observed with equal conditional probabilities, then $\rho^\textrm{s}_{1}\left(x^1_{1|2}\right)+\rho^\textrm{s}_{1}\left(x^1_{2|1}\right)=1$. The slicing distribution, therefore, is also consistent with RUM.
\end{example}

\begin{table}
\begin{centering}
\begin{tabular}{c!{\vrule width 2pt}c|c|c|c}
& $x^2_{1|1}$ & $x^2_{2|1}$ & $x^2_{1|2}$ & $x^2_{2|2}$\\
\noalign{\hrule height 2pt}
$x^1_{1|1}$ & 1/6 & 1/3 & 2/3 & -\\
\hline 
$x^1_{2|1}$ & 1/3 & 1/6 & 1/6 & 1/6\\
\hline 
$x^1_{1|2}$ & 1/6 & 1/3 & 2/3 & -\\
\hline 
$x^1_{2|2}$ & 1/3 & 1/6 & 1/6 & 1/6
\end{tabular}
\par\end{centering}
\caption{Matrix representation of $\rho$ that is consistent with RUM after slicing, but is not consistent with DRUM}\label{tab:marginalRumNotDrum}
\end{table}

\subsection{Pooling \label{sec:pooling}}
In practice, and in the absence of panel variation,  several years or time periods of choices from budgets are pooled before testing for consistency with RUM (KS,\citealp{deb2017revealed}). Here we explore a potential pitfall of this practice. We show that when a panel dataset that is consistent with DRUM is pooled, it may not be consistent with RUM. The spurious rejection of rationality may be driven by the fact that pooling requires us to ignore time labels and imposes the restriction that the distribution of preferences is independent across time. 
\par
First, we formally define \emph{pooling}. To simplify the exposition, assume that $B_j^{*,t}\neq B_{j'}^{*,t'}$ for all $t,t'\in\mathcal{T}$, $j\in\mathcal{J}^t$, and $j'\in\mathcal{J}^{t'}$. That is, there are no repeated budgets across time and agents. Let $\mathcal{J}=\{1,2,\dots, J\}$, where $J=\sum_{t\in\mathcal{T}}J^t$ is the total number of budgets.
\begin{definition}[Pooled Patches] Let 
$
\mathcal{X}=\bigcup_{t\in\mathcal{T}}\bigcup_{j\in \mathcal{J}^t} \{\xi^t_{k|j}\}
$
be the coarsest partition of $\bigcup_{t\in\mathcal{T}}\bigcup_{j\in\mathcal{J}^t}B^{*,t}_{j}$ such that
$
\xi^t_{k|j}\bigcap B^{*,t}_{j'}\in\{\xi^t_{k|j},\emptyset\}
$
for any $j,j'$, and $k$.
\end{definition}
The pooled patches $\{\xi^t_{k|j}\}$ partition every $x^t_{i|j}$ since $B^{*,t}_{j}$ may now intersect with budgets from different periods (see Figure~\ref{fig:pooledpatches}). 
\begin{figure}[h]
\begin{centering}
\begin{tikzpicture}[scale=0.6] 
\draw[thick,->] (0,0) -- (6,0) node[anchor=north west] {$y_1$}; 
\draw[thick,->] (0,0) -- (0,6) node[anchor=south east] {$y_2$}; 
\draw[very thick] (0,5) -- (3,0); 
\draw (2.1,3) node {$x^1_{1|1}$}; 
\draw (3,-0.5) node {$B^{*,1}_{1}$}; 
\end{tikzpicture}
\begin{tikzpicture}[scale=0.6] 
\draw[thick,->] (0,0) -- (6,0) node[anchor=north west] {$y_1$}; 
\draw[thick,->] (0,0) -- (0,6) node[anchor=south east] {$y_2$}; 
\draw[very thick] (5,0) -- (0,3); 
\draw (3,2) node {$x^2_{1|1}$}; 
\draw (5,-0.5) node {$B^{*,2}_{1}$}; 
\end{tikzpicture}
\begin{tikzpicture}[scale=0.6] 
\draw[thick,->] (0,0) -- (6,0) node[anchor=north west] {$y_1$}; 
\draw[thick,->] (0,0) -- (0,6) node[anchor=south east] {$y_2$}; 
\draw[very thick] (0,5) -- (3,0); 
\draw[very thick] (5,0) -- (0,3); 
\draw [fill=black] (1.875,1.875) circle[radius=.1]; 
\draw (0.7,2) node {$\xi^2_{1|1}$}; 
\draw (2,0.5) node {$\xi^1_{2|1}$}; 
\draw (1.4,4) node {$\xi^1_{1|1}$}; 
\draw (4,1.2) node {$\xi^2_{2|1}$}; 
\draw (3,-0.5) node {$B^{*,1}_{1}$}; 
\draw (-0.5,3) node {$B^{*,2}_{1}$};
\end{tikzpicture}
\par\end{centering}
\caption{$K=2$ goods, $T=2$ time periods, one budget per time period. The first and the second picture depict patches in 2 time periods. The third picture depicts new patches that arise after pooling the data. \label{fig:pooledpatches}}
\end{figure}
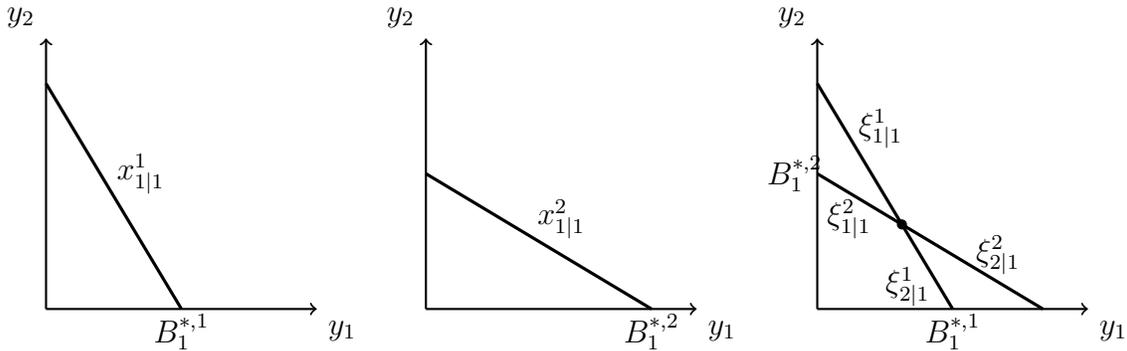
Given these new patches, we can define the pooled stochastic function $\rho^{\text{pool}}\left(\xi^t_{k|j}\right)$ as the probability of observing someone picking from patch $\xi^t_{k|j}$. Next, we construct an example in which $\rho$ is rationalizable by DRUM but the corresponding $\rho^{\text{pool}}$ is not consistent with RUM (in the sense of Proposition~\ref{prop:DRUMimpliesmarginalRUM}). Consider the setting with $K=2$ goods and $T=2$ time periods. In each time period $t$, there is only one budget $B_1^{*,t}$. Assume that $B_1^{*,1}\neq B_1^{*,2}$ and $B_1^{*,1}\cup B_1^{*,2}\neq\emptyset$ (see Figure~\ref{fig:pooledpatches}). Given that there is no budget variation for any given time period, there is only one choice path $\left(x^1_{1|1},x^2_{1|1}\right)$. Thus, the trivial $\rho\left(\left(x^1_{1|1},x^2_{1|1}\right)\right)=1$ is rationalizable by DRUM. After pooling, since the budgets overlap, there are four patches (we assume that there is no intersection patch). Since there is only one choice path, DRUM does not impose any restrictions on the choice of individuals on these two budgets. As a result, we can take  $\nu^1$ and $\nu^2$ from the DRUM definition such that $\rho^{\text{pool}}\left(\xi^2_{1|1}\right)+\rho^{\text{pool}}\left(\xi^1_{2|1}\right)>1$. This $\rho^{\text{pool}}$ cannot be consistent with RUM.

\section{Counterfactuals in the Demand Setup}\label{sec: counterfactuals}
This section shows how to conduct sharp counterfactual analyses within our framework in the demand setup.\footnote{\emph{Sharpness} in this setting means that we can compute the tightest set of parameters that are consistent with the observed data and the model.} The sharpness of our results follows from the fact that we have a full characterization of DRUM.\footnote{For early connections between nonparametric counterfactuals and specification testing, see \citet{varian1982nonparametric,varian1984nonparametric}. See \citet{blundell2008best}, \citet{norets2014semiparametric}, \citet{blundell2014bounding}, \citet{allen2019identification}, \citet{AK2021}, and \citet{AGUIAR2022PPPP} for recent examples in the analysis of demand, dynamic binary choice, and production.} To simplify the exposition, we focus on the simple setup (i.e., two intersecting budgets per period) for which we possess the $\mathcal{H}$-representation.

Given $\rho$ in the time window $\mathcal{T}$, we want to bound some known function of counterfactual stochastic demands at the counterfactual time $T+1$. We assume that consumers face a pair of prices $p_{1,T+1}$ and $p_{2,T+1}$ that are known to the analyst at $T+1$. Let income in each period be $1$. Denote the extended time window by $\mathcal{T}^{\mathrm{c}}=\mathcal{T}\cup\{T+1\}$. Similarly, the extended set of budget paths is denoted by $\rand{J}^{\mathrm{c}}$, and the extended vector representation of stochastic demand is denoted by $\rho^{\mathrm{c}}$.

Let $y^{\mathrm{c}}_{j_{T+1}}$ denote the counterfactual random demand of a consumer facing budget $j_{T+1}$ at time $T+1$. That is,
\[
y_{j_{T+1}}^{\mathrm{c}}=\argmax_{y\in B_{j_{T+1}}^{T+1}}u^{T+1}(y),
\]
where $u^{T+1}$ is a random utility function at time $T+1$. 

\begin{definition}[Counterfactual marginal and conditional demands]
Given $\rho$, $x_{\rand{i}|\rand{j}}$, and budget $j_{T+1}\in\mathcal{J}^{T+1}$, the counterfactual conditional and marginal demands $\rho^{*}\left(\cdot|j_{T+1},x_{\rand{i}|\rand{j}}\right)$ and $\rho^{**}\left(\cdot|j_{T+1}\right)$ are distributions over patches of $j_{T+1}$ such that
\begin{align*}
\rho_{j_{T+1}}^{*}\left(x^{T+1}_{i_{T+1}|j_{T+1}}|x_{\rand{i}|\rand{j}}\right)&=\rho^{\mathrm{c}}\left(\left(x^{t}_{i_t|j_t}\right)_{t\in\mathcal{T}^\mathrm{c}}\right)\Big/\rho\left(x_{\rand{i}|\rand{j}}\right),\\
\rho_{j_{T+1}}^{**}\left(x^{T+1}_{i_{T+1}|j_{T+1}}\right)&=\sum_{x_{\rand{i}|\rand{j}}}\rho^{\mathrm{c}}\left(\left(x^{t}_{i_t|j_t}\right)_{t\in\mathcal{T}^\mathrm{c}}\right)
\end{align*}
for any $\rho^{\mathrm{c}}$ that satisfies $\mathrm{D}$-monotonicity and stability, and is such that
\[
\rho\left(x_{\rand{i}|\rand{j}}\right)=\sum_{i_{T+1}\in\mathcal{I}^{T+1}_{j_{T+1}}}\rho^{\mathrm{c}}\left(\left(x^{t}_{i_t|j_t}\right)_{t\in\mathcal{T}^\mathrm{c}}\right).
\]
\end{definition}
The counterfactual conditional and marginal distributions fully characterize the choices of consumers in counterfactual situations, thus allowing us to compute sharp bounds for the expectation of any function of $y^{\mathrm{c}}$. For a given measurable function $g:X\to\Real$, let
\begin{align*}
    \underline{g}\left(x^t_{i_t|j_t}\right)=\inf_{y\in x^t_{i_t|j_t}}g(y),\quad\quad
    \overline{g}\left(x^t_{i_t|j_t}\right)=\sup_{y\in x^t_{i_t|j_t}}g(y).
\end{align*}
be the smallest and the largest values, respectively, $g$ can take over the patch $x^t_{i_t|j_t}$.
\begin{proposition}
     Given $\rho$, $x_{\rand{i}|\rand{j}}$, and budget $j_{T+1}\in\mathcal{J}^{T+1}$, 
     \begin{align*}
         {\scriptstyle \inf_{\rho_{j_{T+1}}^{*}}\sum_{i\in\mathcal{I}^{T+1}_{j_{T+1}}}\rho_{j_{T+1}}^{*}\left(x^{T+1}_{i|j_{T+1}}|x_{\rand{i}|\rand{j}}\right)\underline{g}\left(x^{T+1}_{i|j_{T+1}}\right)}&{\scriptstyle\leq \Exp{g\left(y_{j_{T+1}}^{\mathrm{c}}\right)|x_{\rand{i}|\rand{j}}}\leq\sup_{\rho_{j_{T+1}}^{*}}\sum_{i\in\mathcal{I}^{T+1}_{j_{T+1}}}\rho_{j_{T+1}}^{*}\left(x^{T+1}_{i|j_{T+1}}|x_{\rand{i}|\rand{j}}\right)\overline{g}\left(x^{T+1}_{i|j_{T+1}}\right)},\\
         {\scriptstyle\inf_{\rho_{j_{T+1}}^{**}}\sum_{i\in\mathcal{I}^{T+1}_{j_{T+1}}}\rho_{j_{T+1}}^{**}\left(x^{T+1}_{i|j_{T+1}}\right)\underline{g}\left(x^{T+1}_{i|j_{T+1}}\right)}&{\scriptstyle\leq \Exp{g\left(y_{j_{T+1}}^{\mathrm{c}}\right)}\leq\sup_{\rho_{j_{T+1}}^{**}}\sum_{i\in\mathcal{I}^{T+1}_{j_{T+1}}}\rho_{j_{T+1}}^{**}\left(x^{T+1}_{i|j_{T+1}}\right)\overline{g}\left(x^{T+1}_{i|j_{T+1}}\right),}
     \end{align*}
     where the infimum and supremum are taken over all possible counterfactual marginal and conditional distributions.
\end{proposition}
Note that our results are complementary to those of \citet{kitamura2019nonparametric} that predict counterfactual stochastic demand for a new budget in a given cross-section using static RUM. We can use their techniques here as well. The main difference from their work is that we focus on the counterfactual prediction in the time-dimension allowing dynamic preference change. 

\section{Empirical Application: Binary Menus of Lotteries}\label{sec: application}
We study a sample of experimental subjects studied in \citet{ABBK23} that was surveyed in the MTurk platform between August $25$, $2018$ and September $17$, $2018$. We find evidence that deterministic rationality fails to explain the totality of the sample behavior, yet consistency with DRUM cannot be rejected. Moreover, restricting the linear orders in DRUM to those consistent with Expected Utility \citep{frick2019dynamic} cannot explain the population behavior. 

The sample contains $2135$ DMs. The grand choice set is the same across time $X^{t}=\{l_1,l_2,l_3\}$. These lotteries are defined over the set of prizes $Z=\{0,10,30,50\}$ in tokens. The lotteries are $l_1=(1/2,0,0,0,1/2)$, $l_2=(0,1/2,1/2,0)$ and $l_3=\frac{1}{2}l_1+\frac{1}{2}l_2$. The binary menus are $\{l_1,l_2\},\{l_1,l_3\}$, and $\{l_2,l_3\}$ at every $t\in\mathcal{T}$. In the experiment, consumers face one of these three menus uniformly at random and in uniformly random order. No DM faces the same binary menu twice. That means that there are six menu paths. Note that the design ensures that the probability of facing any of the six menu paths is uniform across DMs. Payments are made at random for one of the choices made by the DMs. (For details of the payment, recruitment, and sample demographics see \citealp{ABBK23}.)  

Concerns about limited consideration are not first-order here because these binary menus were shown to DMs under a low-cost-of-attention treatment (see \citealp{ABBK23} for details). In addition, DMs already faced a task focused on varying the attention cost, as described in \citet{ABBK23}, before facing the binary-menu task. This means DMs are familiar with the lotteries when they face  the binary menus. In addition, we do not have concerns about measurement error due to the experimental design and discrete choice nature of the data. This means we can focus on testing deterministic rationality at the individual level versus DRUM without confounding due to limited consideration and measurement error. 
 
\subsection*{Testing Deterministic Rationality}  
We test deterministic rationality at the individual level. In particular, we check the Strong Axiom of Revealed Preference (SARP) which implies that it cannot be the case that $l_i\succ^* l_i'$, $l_i'\succ^* l_i''$,  and $l_i''\succ^* l_i$ with $l_i\succ^* l_i'$ defined as $l_i$ is picked out of $\{l_i,l_i'\}$). We observe that $92$ percent of DMs are consistent with SARP in our sample. Given that the power of this experiment to detect violations of SARP is low, we consider a $8$ percent rejection rate to represent a significant fraction of DMs. Formally, the whole sample of DMs is not consistent with (static) utility maximization.
 
\subsection*{Testing DRUM}  
 
We compute the sample analogue of $\rho$, $\hat{\rho}$, following the methodology described in KS and \citet{ABBK23}. We implement the statistical test described in KS and \citet{ABBK23} to test the null hypothesis that $\rho=A\nu$ for $\nu\geq 0$. $A$ is computed in two steps. First, we compute the matrix $A^t$ whose columns correspond to all possible linear orders on $X^t$ and whose rows correspond to all possible lotteries from each menu. Observe that $A^t$ is a square matrix of dimension $6$ for all $t\in\mathcal{T}$ such that $A_T=\otimes_{t=1}^3A^t$. In this application, we do not observe all possible choice paths since, by design, DMs never see repeated binary menus. Thus, $A$ is obtained by extracting the submatrix of $A_T$ that corresponds to observed menu paths ($6$ menu paths out of $27$). We cannot reject the null hypothesis of consistency with DRUM  ($p$-value is $0.72$).\footnote{The number of bootstrap samples used in every test is 999. We choose the tuning parameter $\tau_{N}$ following KS.} Our finding confirms that a sample can be consistent with DRUM even when a significant fraction of DMs is inconsistent with static utility maximization. We believe our results are compatible with the findings of \citet{kurtz2019neural} which document that DMs fail deterministic rationality (in a different choice domain) because they make mistakes in evaluating the utility of lotteries. These mistakes can be interpreted as a form of dynamic random taste shocks. Indeed, DMs may get fatigued or learn, so their  mistakes are correlated in time but with draws from the same distribution across choice paths, which are hence consistent with DRUM. 
To alleviate the concerns about the finite sample power of our test, we perform Monte Carlo experiments mimicking the setup of this application and find evidence that our test has high power even in small samples.

\subsection*{Testing the Expected Utility version of DRUM}  
  
\citet{kashaev2022randomrank} propose a methodology to test for the null hypothesis of consistency with the Expected Utility version of RUM. Their approach restricts the set of preference orders in the static test of RUM to those that are further consistent with expected utility maximization. We can apply the same idea to our dynamic setting to test the null hypothesis of consistency with the Expected Utility version of DRUM. That is, we test DRUM with the additional restriction that the set of utilities is consistent with Expected Utility. We reject the null hypothesis ($p$-value is less than $0.001$).

\section{Related Literature\label{sec: litreview}} 
DRUM was first introduced in \citet{straleckinotes} for abstract discrete choice domains. \citet{frick2019dynamic} offered an axiomatic characterization with decision trees and expected utility restrictions on stochastic utility processes.\footnote{Applying techniques from \citet{kashaev2022randomrank}, we provide a KS-type characterization for special DRUM cases with expected utility restrictions.} However, as far as we know, our work is the first to develop a BM-like characterization for DRUM. In finite abstract discrete choice spaces, two partial characterizations exist when the primitive is the joint distribution of choices across time and total menu variation.

\citet{li2021axiomatization} provides an BM-like axiomatic DRUM characterization  for any finite number of periods, and full menu variation, but with no more than three alternatives. \citet{chambers2021correlated} considers correlated choice: a joint distribution of choice on a pair of menus faced by two different DMs or a group. This model is mathematically equivalent to DRUM in the abstract domain, characterizing DRUM for a special case of two periods, an abstract and finite choice set, but with a uniqueness property for one DM's choice (i.e., one period has uniquely identified RUM). Our work subsumes and generalizes both \citet{li2021axiomatization} and \citet{chambers2021correlated}. Moreover, our results go beyond both \citet{li2021axiomatization} and \citet{chambers2021correlated}, and our general setup  also includes classical consumer choice domains with primitive orders, binary menus, and general limited menu variation. Our DRUM respects this primitive order and restricts utilities to be monotonic, while addressing limited observability of menus and menu paths.

Our work also contributes to the random exponential discounting literature, as in \citet{browning1989anonparametric}, by generalizing \citet{deb2017revealed}'s demand setup to a dynamic context. \citet{Apesteguia2022random} introduces a heterogeneous model with exponential discounting and time separability, differing in choice domains and being semiparametric, while our setup is nonparametric. \citet{lu2018random} examines exponential discounting with random discount factors and stochastic choices over consumption streams in the first period.

\citet{AK2021} studies panel setups with first-order-conditions approaches for some dynamic preferences, allowing for measurement error. However, their setup doesn't accommodate changing utility beyond discount factors or marginal utility of income. \citet{im2021non} investigates the McFadden-Richter framework and panel structure limitations, suggesting individual static rationality checks, like the Samuelson-Afriat framework. We generalize the Samuelson-Afriat framework, allowing utility changes over time while exploiting panel structure for more empirical implications.

DRUM allows for consumption to be correlated in time, similar to rational addiction in \citet{becker1988theory} (henceforth, Habits as Durables--HAD). However, the consumption correlation in DRUM is due solely to preference correlation over time. \citet{demuynck2013habits} provides a characterization of the short-memory habits model of \citet{becker1988theory} which is similar to \citet{afriat1967construction}'s theorem,  and applies it to a panel of Spanish households' choices. However, the HAD pass-rate is slightly above $50$ percent, subject to the same critique made for the traditional utility maximization problem. This revealed-preference test of HAD imposes no parametric restrictions on utilities besides monotonicity and concavity, making it more general than parametric structural work on HAD. Theoretically, the relationship between DRUM and HAD remains an open question since a RUM-like characterization of a stochastic generalization of HAD is unavailable.

Dynamic Discrete Choice (DDC) models, surveyed in \citet{aguirregabiria2010DDC}, are the most popular approach in studying discrete choice with attribute variation. \citet{frick2019dynamic} thoroughly explores the relationship between DRUM and DDC in their domain. Much of this analysis carries over to our domain. Specifically, in our setup with exogenously given menu paths, DDC is nested by DRUM. However, our focus differs from most work on DDC in that we emphasize nonparametric utilities, unrestricted heterogeneity, comparative statics, and counterfactual predictions, rather than identification and estimation in parametric settings.
No axiomatization of DDC has yet been arrived at, to our knowledge,  except for the special case of independent and identically distributed logit shocks in \citet{fudenberg2015dynamic}, so we cannot fully compare DDC with DRUM.

\section{Conclusion}\label{sec: conclusion}
In this paper, we have fully characterized the Dynamic Random Utility Model (DRUM), a new model of consumer behavior in which we observe a panel of choices from budget paths. In contrast to the static utility maximization framework, DRUM does not require the assumption that decision makers or consumers keep their preferences stable over time. This generality is essential because the static utility maximization framework often fails to explain the behavior of individuals. 
\par 
Our characterization works for any finite collection of choice paths in any finite time window. The characterization can be applied directly to existing panel consumption datasets using the statistical tools in \citet{kitamura2018nonparametric}. Moreover, our simple setup characterization showcases the fact that DRUM implies a richer set of behavioral restrictions on the panel of choices than does the Random Utility Model, alleviating some concerns about the empirical bite of the latter in a richer domain. These features position DRUM as in-between the Samuelson-Afriat framework and the McFadden-Richter framework as a result of combining their strengths and reducing their weaknesses.  
\par 
We have also introduced to economics a generalization of the Weyl-Minkowski theorem for cones. This result is the basis of a recursive characterization of DRUM in the demand and abstract choice setups. This new mathematical result will be helpful beyond DRUM to obtain analogous generalizations of bounded rational models of stochastic choice.

\bibliographystyle{ecca}
\bibliography{main.bib}

\begin{thebibliography}{71}
\providecommand{\natexlab}[1]{#1}

\bibitem[{Adams \textit{et~al.}(2015)Adams, Blundell, Browning and
  Crawford}]{adams2015prices}
\textsc{Adams, A.}, \textsc{Blundell, R.}, \textsc{Browning, M.} and
  \textsc{Crawford, I.} (2015). \textit{Prices versus preferences: taste change
  and revealed preference}. Tech. rep., IFS Working Papers.

\bibitem[{Afriat(1967)}]{afriat1967construction}
\textsc{Afriat, S.~N.} (1967). The construction of utility functions from
  expenditure data. \textit{International economic review}, \textbf{8}~(1),
  67--77.

\bibitem[{Afriat(1973)}]{afriat1973efficiency}
\textsc{---} (1973). On a system of inequalities in demand analysis: an
  extension of the classical method. \textit{International economic review},
  pp. 460--472.

\bibitem[{Aguiar \textit{et~al.}(2023)Aguiar, Boccardi, Kashaev and
  Kim}]{ABBK23}
\textsc{Aguiar, V.~H.}, \textsc{Boccardi, M.~J.}, \textsc{Kashaev, N.} and
  \textsc{Kim, J.} (2023). Random utility and limited consideration.
  \textit{Quantitative Economics}, \textbf{14}~(1), 71--116.

\bibitem[{Aguiar and Kashaev(2021)}]{AK2021}
\textsc{---} and \textsc{Kashaev, N.} (2021). Stochastic revealed preferences
  with measurement error. \textit{The Review of Economic Studies},
  \textbf{88}~(4), 2042--2093.

\bibitem[{Aguiar \textit{et~al.}(2022)Aguiar, Kashaev and
  Allen}]{AGUIAR2022PPPP}
\textsc{---}, \textsc{---} and \textsc{Allen, R.} (2022). Prices, profits,
  proxies, and production. \textit{Journal of Econometrics}.

\bibitem[{Aguiar and Serrano(2021)}]{Aguiar2021JMathE}
\textsc{---} and \textsc{Serrano, R.} (2021). Cardinal revealed preference:
  Disentangling transitivity and consistent binary choice. \textit{Journal of
  Mathematical Economics}, \textbf{94}, 102462.

\bibitem[{Aguirregabiria and Mira(2010)}]{aguirregabiria2010DDC}
\textsc{Aguirregabiria, V.} and \textsc{Mira, P.} (2010). Dynamic discrete
  choice structural models: A survey. \textit{Journal of Econometrics},
  \textbf{156}~(1), 38--67.

\bibitem[{Ahn \textit{et~al.}(2014)Ahn, Choi, Gale and
  Kariv}]{ahn2014estimating}
\textsc{Ahn, D.}, \textsc{Choi, S.}, \textsc{Gale, D.} and \textsc{Kariv, S.}
  (2014). Estimating ambiguity aversion in a portfolio choice experiment.
  \textit{Quantitative Economics}, \textbf{5}~(2), 195--223.

\bibitem[{Akesaka \textit{et~al.}(2021)Akesaka, Eibich, Hanaoka and
  Shigeoka}]{akesaka2021temporal}
\textsc{Akesaka, M.}, \textsc{Eibich, P.}, \textsc{Hanaoka, C.} and
  \textsc{Shigeoka, H.} (2021). Temporal instability of risk preference among
  the poor: Evidence from payday cycles. \textit{American Economic Journal:
  Applied Economics}, forthcoming.

\bibitem[{Allen and Rehbeck(2019)}]{allen2019identification}
\textsc{Allen, R.} and \textsc{Rehbeck, J.} (2019). Identification with
  additively separable heterogeneity. \textit{Econometrica}, \textbf{87}~(3),
  1021--1054.

\bibitem[{Apesteguia \textit{et~al.}(2022)Apesteguia, Ballester and
  Gutierrez-Daza}]{Apesteguia2022random}
\textsc{Apesteguia, J.}, \textsc{Ballester, M.~A.} and \textsc{Gutierrez-Daza,
  A.} (2022). Random discounted expected utility.

\bibitem[{Aubrun \textit{et~al.}(2021)Aubrun, Lami, Palazuelos and
  Pl{\'a}vala}]{aubrun2021entangleability}
\textsc{Aubrun, G.}, \textsc{Lami, L.}, \textsc{Palazuelos, C.} and
  \textsc{Pl{\'a}vala, M.} (2021). Entangleability of cones. \textit{Geometric
  and Functional Analysis}, \textbf{31}~(2), 181--205.

\bibitem[{Aubrun \textit{et~al.}(2022)Aubrun, M{\"u}ller-Hermes and
  Pl{\'a}vala}]{aubrun2022monogamy}
\textsc{---}, \textsc{M{\"u}ller-Hermes, A.} and \textsc{Pl{\'a}vala, M.}
  (2022). Monogamy of entanglement between cones. \textit{arXiv preprint
  arXiv:2206.11805}.

\bibitem[{Bandyopadhyay \textit{et~al.}(1999)Bandyopadhyay, Dasgupta and
  Pattanaik}]{bandyopadhyay1999stochastic}
\textsc{Bandyopadhyay, T.}, \textsc{Dasgupta, I.} and \textsc{Pattanaik, P.~K.}
  (1999). Stochastic revealed preference and the theory of demand.
  \textit{Journal of Economic Theory}, \textbf{84}~(1), 95--110.

\bibitem[{Bandyopadhyay \textit{et~al.}(2004)Bandyopadhyay, Dasgupta and
  Pattanaik}]{bandyopadhyay2004general}
\textsc{---}, \textsc{---} and \textsc{---} (2004). A general revealed
  preference theorem for stochastic demand behavior. \textit{Economic Theory},
  \textbf{23}~(3), 589--599.

\bibitem[{Becker and Murphy(1988)}]{becker1988theory}
\textsc{Becker, G.~S.} and \textsc{Murphy, K.~M.} (1988). A theory of rational
  addiction. \textit{Journal of political Economy}, \textbf{96}~(4), 675--700.

\bibitem[{Block and Marschak(1960)}]{block1960random}
\textsc{Block, H.} and \textsc{Marschak, J.} (1960). Random orderings and
  stochastic theories of responses". in i. olkin, s. ghurye, w. hoeffding, w.
  madow, and h. man (eds) contributions to probability and statistics, stanford
  university press.

\bibitem[{Blundell \textit{et~al.}(2008)Blundell, Browning and
  Crawford}]{blundell2008best}
\textsc{Blundell, R.}, \textsc{Browning, M.} and \textsc{Crawford, I.} (2008).
  Best nonparametric bounds on demand responses. \textit{Econometrica},
  \textbf{76}~(6), 1227--1262.

\bibitem[{Blundell \textit{et~al.}(2014)Blundell, Kristensen and
  Matzkin}]{blundell2014bounding}
\textsc{---}, \textsc{Kristensen, D.} and \textsc{Matzkin, R.} (2014). Bounding
  quantile demand functions using revealed preference inequalities.
  \textit{Journal of Econometrics}, \textbf{179}~(2), 112--127.

\bibitem[{Border(2007)}]{border2007introductory}
\textsc{Border, K.} (2007). Introductory notes on stochastic rationality.
  \textit{California Institute of Technology}.

\bibitem[{Brocas \textit{et~al.}(2019)Brocas, Carrillo, Combs and
  Kodaverdian}]{brocas2019consistency}
\textsc{Brocas, I.}, \textsc{Carrillo, J.~D.}, \textsc{Combs, T.~D.} and
  \textsc{Kodaverdian, N.} (2019). Consistency in simple vs. complex choices by
  younger and older adults. \textit{Journal of Economic Behavior \&
  Organization}, \textbf{157}, 580--601.

\bibitem[{Browning(1989)}]{browning1989anonparametric}
\textsc{Browning, M.} (1989). A nonparametric test of the life-cycle rational
  expections hypothesis. \textit{International Economic Review}, pp. 979--992.

\bibitem[{Cattaneo \textit{et~al.}(2020)Cattaneo, Ma, Masatlioglu and
  Suleymanov}]{cattaneo2020random}
\textsc{Cattaneo, M.~D.}, \textsc{Ma, X.}, \textsc{Masatlioglu, Y.} and
  \textsc{Suleymanov, E.} (2020). A random attention model. \textit{Journal of
  Political Economy}, \textbf{128}~(7), 2796--2836.

\bibitem[{Chambers \textit{et~al.}(2021)Chambers, Masatlioglu and
  Turansick}]{chambers2021correlated}
\textsc{Chambers, C.~P.}, \textsc{Masatlioglu, Y.} and \textsc{Turansick, C.}
  (2021). Correlated choice. \textit{arXiv preprint arXiv:2103.05084}.

\bibitem[{Cherchye \textit{et~al.}(2017)Cherchye, Demuynck, De~Rock and
  Vermeulen}]{cherchye2017household}
\textsc{Cherchye, L.}, \textsc{Demuynck, T.}, \textsc{De~Rock, B.} and
  \textsc{Vermeulen, F.} (2017). Household consumption when the marriage is
  stable. \textit{American Economic Review}, \textbf{107}~(6), 1507--1534.

\bibitem[{Choi \textit{et~al.}(2007)Choi, Fisman, Gale and
  Kariv}]{choi2007revealing}
\textsc{Choi, S.}, \textsc{Fisman, R.}, \textsc{Gale, D.~M.} and \textsc{Kariv,
  S.} (2007). Revealing preferences graphically: an old method gets a new tool
  kit. \textit{American Economic Review}, \textbf{97}~(2), 153--158.

\bibitem[{Choi \textit{et~al.}(2014)Choi, Kariv, M{\"u}ller and
  Silverman}]{choi2014more}
\textsc{---}, \textsc{Kariv, S.}, \textsc{M{\"u}ller, W.} and
  \textsc{Silverman, D.} (2014). Who is (more) rational? \textit{The American
  Economic Review}, \textbf{104}~(6), 1518--1550.

\bibitem[{Dasgupta and Pattanaik(2007)}]{dasgupta2007regular}
\textsc{Dasgupta, I.} and \textsc{Pattanaik, P.~K.} (2007). `regular'choice and
  the weak axiom of stochastic revealed preference. \textit{Economic Theory},
  \textbf{31}~(1), 35--50.

\bibitem[{de~Bruyn(2020)}]{jossede2020tensor}
\textsc{de~Bruyn, J. v.~D.} (2020). Tensor products of convex cones, part ii:
  Closed cones in finite-dimensional spaces. \textit{arXiv preprint
  arXiv:2009.11843}.

\bibitem[{Dean and Martin(2016)}]{dean2016measuring}
\textsc{Dean, M.} and \textsc{Martin, D.} (2016). Measuring rationality with
  the minimum cost of revealed preference violations. \textit{Review of
  Economics and Statistics}, \textbf{98}~(3), 524--534.

\bibitem[{Deb \textit{et~al.}(2021)Deb, Kitamura, Quah and
  Stoye}]{deb2017revealed}
\textsc{Deb, R.}, \textsc{Kitamura, Y.}, \textsc{Quah, J. K.-H.} and
  \textsc{Stoye, J.} (2021). Revealed price preference: Theory and statistical
  analysis. \textit{Review of Economic Studies}.

\bibitem[{Debreu and Scarf(1963)}]{debreu1963limit}
\textsc{Debreu, G.} and \textsc{Scarf, H.} (1963). A limit theorem on the core
  of an economy. \textit{International Economic Review}, \textbf{4}~(3),
  235--246.

\bibitem[{Demuynck and Verriest(2013)}]{demuynck2013habits}
\textsc{Demuynck, T.} and \textsc{Verriest, E.} (2013). I'll never forget my
  first cigarette: a revealed preference analysis of the ``habits as durables''
  model. \textit{International Economic Review}, \textbf{54}~(2), 717--738.

\bibitem[{Dogan and Yildiz(2022)}]{dogan2022every}
\textsc{Dogan, S.} and \textsc{Yildiz, K.} (2022). Every choice function is
  pro-con rationalizable. \textit{Operations Research}.

\bibitem[{Doherty \textit{et~al.}(2004)Doherty, Parrilo and
  Spedalieri}]{doherty2004complete}
\textsc{Doherty, A.~C.}, \textsc{Parrilo, P.~A.} and \textsc{Spedalieri, F.~M.}
  (2004). Complete family of separability criteria. \textit{Physical Review A},
  \textbf{69}~(2), 022308.

\bibitem[{Dridi(1980)}]{dridi1980sur}
\textsc{Dridi, T.} (1980). Sur les distributions binaires associes des
  distributions ordinales. \textit{Mathmatiques et Sciences Humaines},
  \textbf{69}, 15--31.

\bibitem[{Echenique \textit{et~al.}(2011)Echenique, Lee and
  Shum}]{echenique2011money}
\textsc{Echenique, F.}, \textsc{Lee, S.} and \textsc{Shum, M.} (2011). The
  money pump as a measure of revealed preference violations. \textit{Journal of
  Political Economy}, \textbf{119}~(6), 1201--1223.

\bibitem[{Falmagne(1978)}]{falmagne1978representation}
\textsc{Falmagne, J.-C.} (1978). A representation theorem for finite random
  scale systems. \textit{Journal of Mathematical Psychology}, \textbf{18}~(1),
  52--72.

\bibitem[{Fang \textit{et~al.}(2023)Fang, Santos, Shaikh and
  Torgovitsky}]{fang2023inference}
\textsc{Fang, Z.}, \textsc{Santos, A.}, \textsc{Shaikh, A.~M.} and
  \textsc{Torgovitsky, A.} (2023). Inference for large-scale linear systems
  with known coefficients. \textit{Econometrica}, \textbf{91}~(1), 299--327.

\bibitem[{Fishburn(1998)}]{fishburn1998stochastic}
\textsc{Fishburn, P.~C.} (1998). Stochastic utility. In S.~Barber\'a, P.~J.
  Hammond and C.~Seidl (eds.), \textit{Handbook of Utility Theory}, pp.
  272--311.

\bibitem[{Frick \textit{et~al.}(2019)Frick, Iijima and
  Strzalecki}]{frick2019dynamic}
\textsc{Frick, M.}, \textsc{Iijima, R.} and \textsc{Strzalecki, T.} (2019).
  Dynamic random utility. \textit{Econometrica}, \textbf{87}~(6), 1941--2002.

\bibitem[{Fudenberg and Strzalecki(2015)}]{fudenberg2015dynamic}
\textsc{Fudenberg, D.} and \textsc{Strzalecki, T.} (2015). Dynamic logit with
  choice aversion. \textit{Econometrica}, \textbf{83}~(2), 651--691.

\bibitem[{Gauthier(2018)}]{gauthier2018}
\textsc{Gauthier, C.} (2018). Nonparametric identification of discount factors
  under partial efficiency. \textit{Working paper}.

\bibitem[{Gauthier(2021)}]{gauthier2021}
\textsc{---} (2021). Price search and consumption inequality: Robust, credible,
  and valid inference. \textit{Working paper}.

\bibitem[{Guiso \textit{et~al.}(2018)Guiso, Sapienza and
  Zingales}]{guiso2018time}
\textsc{Guiso, L.}, \textsc{Sapienza, P.} and \textsc{Zingales, L.} (2018).
  Time varying risk aversion. \textit{Journal of Financial Economics},
  \textbf{128}~(3), 403--421.

\bibitem[{Halevy and Mayraz(2022)}]{halevy2022identifying}
\textsc{Halevy, Y.} and \textsc{Mayraz, G.} (2022). Identifying rule-based
  rationality. \textit{Review of Economics and Statistics}, pp. 1--44.

\bibitem[{Hoderlein and Stoye(2014)}]{hoderlein2014revealed}
\textsc{Hoderlein, S.} and \textsc{Stoye, J.} (2014). Revealed preferences in a
  heterogeneous population. \textit{Review of Economics and Statistics},
  \textbf{96}~(2), 197--213.

\bibitem[{Im and Rehbeck(2021)}]{im2021non}
\textsc{Im, C.} and \textsc{Rehbeck, J.} (2021). Non-rationalizable
  individuals, stochastic rationalizability, and sampling. \textit{Available at
  SSRN 3767994}.

\bibitem[{Kashaev and Aguiar(2022{\natexlab{a}})}]{RAUM2022random}
\textsc{Kashaev, N.} and \textsc{Aguiar, V.~H.} (2022{\natexlab{a}}). A random
  attention and utility model. \textit{Journal of Economic Theory},
  \textbf{204}, 105487.

\bibitem[{Kashaev and Aguiar(2022{\natexlab{b}})}]{kashaev2022randomrank}
\textsc{---} and \textsc{---} (2022{\natexlab{b}}). Random rank-dependent
  expected utility. \textit{Games}, \textbf{13}~(1), 13.

\bibitem[{Kawaguchi(2017)}]{kawaguchi2017testing}
\textsc{Kawaguchi, K.} (2017). Testing rationality without restricting
  heterogeneity. \textit{Journal of Econometrics}, \textbf{197}~(1), 153--171.

\bibitem[{Kitamura and Stoye(2018)}]{kitamura2018nonparametric}
\textsc{Kitamura, Y.} and \textsc{Stoye, J.} (2018). Nonparametric analysis of
  random utility models. \textit{Econometrica}, \textbf{86}~(6), 1883--1909.

\bibitem[{Kitamura and Stoye(2019)}]{kitamura2019nonparametric}
\textsc{---} and \textsc{---} (2019). Nonparametric counterfactuals in random
  utility models. \textit{arXiv preprint arXiv:1902.08350}.

\bibitem[{Kurtz-David \textit{et~al.}(2019)Kurtz-David, Persitz, Webb and
  Levy}]{kurtz2019neural}
\textsc{Kurtz-David, V.}, \textsc{Persitz, D.}, \textsc{Webb, R.} and
  \textsc{Levy, D.~J.} (2019). The neural computation of inconsistent choice
  behavior. \textit{Nature communications}, \textbf{10}~(1), 1583.

\bibitem[{Li(2021)}]{li2021axiomatization}
\textsc{Li, R.} (2021). An axiomatization of stochastic utility. \textit{arXiv
  preprint arXiv:2102.00143}.

\bibitem[{Lu and Saito(2018)}]{lu2018random}
\textsc{Lu, J.} and \textsc{Saito, K.} (2018). Random intertemporal choice.
  \textit{Journal of Economic Theory}, \textbf{177}, 780--815.

\bibitem[{McCausland \textit{et~al.}(2020)McCausland, Davis-Stober, Marley,
  Park and Brown}]{mccausland2020testing}
\textsc{McCausland, W.~J.}, \textsc{Davis-Stober, C.}, \textsc{Marley, A.~A.},
  \textsc{Park, S.} and \textsc{Brown, N.} (2020). Testing the random utility
  hypothesis directly. \textit{The Economic Journal}, \textbf{130}~(625),
  183--207.

\bibitem[{McFadden and Richter(1990)}]{mcfadden1990stochastic}
\textsc{McFadden, D.} and \textsc{Richter, M.~K.} (1990). Stochastic
  rationality and revealed stochastic preference. \textit{Preferences,
  Uncertainty, and Optimality, Essays in Honor of Leo Hurwicz, Westview Press:
  Boulder, CO}, pp. 161--186.

\bibitem[{McFadden(2005)}]{mcfadden2005revealed}
\textsc{McFadden, D.~L.} (2005). Revealed stochastic preference: a synthesis.
  \textit{Economic Theory}, \textbf{26}~(2), 245--264.

\bibitem[{Norets and Tang(2014)}]{norets2014semiparametric}
\textsc{Norets, A.} and \textsc{Tang, X.} (2014). Semiparametric inference in
  dynamic binary choice models. \textit{Review of Economic Studies},
  \textbf{81}~(3), 1229--1262.

\bibitem[{Porter and Adams(2016)}]{porter2016love}
\textsc{Porter, M.} and \textsc{Adams, A.} (2016). For love or reward?
  characterising preferences for giving to parents in an experimental setting.
  \textit{The Economic Journal}, \textbf{126}~(598), 2424--2445.

\bibitem[{Ray and Robson(2018)}]{ray2018certified}
\textsc{Ray, D.} and \textsc{Robson, A.} (2018). Certified random: A new order
  for coauthorship. \textit{American Economic Review}, \textbf{108}~(2),
  489--520.

\bibitem[{Saito(2017)}]{saito2017axiomatizations}
\textsc{Saito, K.} (2017). Axiomatizations of the mixed logit model.

\bibitem[{Samuelson(1938)}]{samuelson1938note}
\textsc{Samuelson, P.~A.} (1938). A note on the pure theory of consumer's
  behaviour. \textit{Economica}, \textbf{5}~(17), 61--71.

\bibitem[{Smeulders \textit{et~al.}(2021)Smeulders, Cherchye and
  De~Rock}]{smeulders2021nonparametric}
\textsc{Smeulders, B.}, \textsc{Cherchye, L.} and \textsc{De~Rock, B.} (2021).
  Nonparametric analysis of random utility models: computational tools for
  statistical testing. \textit{Econometrica}, \textbf{89}~(1), 437--455.

\bibitem[{Stoye(2019)}]{stoye2019revealed}
\textsc{Stoye, J.} (2019). Revealed stochastic preference: A one-paragraph
  proof and generalization. \textit{Economics Letters}, \textbf{177}, 66--68.

\bibitem[{Strzalecki(2021)}]{straleckinotes}
\textsc{Strzalecki, T.} (2021). \textit{Stochastic Choice}. Mimeo.

\bibitem[{Turansick(2022)}]{turansick2022identification}
\textsc{Turansick, C.} (2022). Identification in the random utility model.
  \textit{Journal of Economic Theory}, \textbf{203}, 105489.

\bibitem[{Varian(1982)}]{varian1982nonparametric}
\textsc{Varian, H.~R.} (1982). The nonparametric approach to demand analysis.
  \textit{Econometrica: Journal of the Econometric Society}, pp. 945--973.

\bibitem[{Varian(1984)}]{varian1984nonparametric}
\textsc{---} (1984). The nonparametric approach to production analysis.
  \textit{Econometrica: Journal of the Econometric Society}, pp. 579--597.

\end{thebibliography}
\appendix
\section{Proofs}\label{app: proofs}
\subsection{Proof of Lemma~\ref{lem: demandisdrum}}
We adapt the proof of Theorem~$3.1$ in KS for RUM for the dynamic case.  Our proof uses profiles of nonstochastic demand. 
For each time period $t\in \mathcal{T}$, we define nonstochastic demand types as in KS: $(\theta^t_{1},\cdots,\theta^t_{J^t})\in B^t_{1}\times\cdots\times B^t_{J^t}$. This system of types is rationalizable if $\theta^t_{j}\in \argmax_{y\in B^t_{j}}u^t(y)$ for $j=1,\cdots,J^t$ for some utility function $u^t$. 
\par
Then, we form any given nonstochastic demand profile by stacking up the demand types in a budget path $\rand{j}$ as $\theta_{\rand{j}}=(\theta^t_{j_t})_{j_t\in \rand{j}}$.  
\par 
Fix $\rho$. For a fixed $t\in\mathcal{T}$, let the set $\mathcal{Y}^{**}_t$ collect the geometric center point of each patch. Let $\rho^{**}$ be the unique dynamic stochastic demand system concentrated on $\mathcal{Y}^{**}_t$ for all $t\in \mathcal{T}$. KS established that demand systems can be arbitrarily perturbed  within patches in a given time period $t$ such that $\rho$ is rationalizable by DRUM if and only if $\rho^{**}$ is. It follows that the rationalizability of $\rho$ can be decided by checking whether there exists a mixture of nonstochastic demand profiles supported on $\mathcal{Y}^{**}_t$ for all $t\in\mathcal{T}$.
\par 
Since we have assumed a finite number of budgets and time periods, there will be a finite number of budget paths. That is, using our notation, we have $\abs{\rand{J}}$ budget paths. Also, because $\mathcal{Y}^{**}_t$ is finite for all $t\in\mathcal{T}$, there are finitely many nonstochastic demand profiles. Noting that these demand profiles are characterized by binary vector representations corresponding to columns of $A_T$, the statement of the theorem follows immediately. 

\subsection{Proof of Theorem~\ref{thm:main}}
To prove $(i)\iff(ii)\iff(iii)$, we adapt the proof of Theorem~$3.1$ in KS for RUM for the dynamic case.  Our proof uses nonstochastic linear order profiles. Then, up to this redefinition $(i)\iff(ii)\iff(iii)$ follows from their results. 

The proof of $(i)\implies (iv)$ follows from \citet{border2007introductory}. The proof of $(iv) \implies (i)$ is analogous to the proof for the case of RUM in \citet{border2007introductory} and \citet{kawaguchi2017testing}. We just need to replace the system of equations in those proofs with the one we describe in Theorem~\ref{thm:main} (ii). The rest of the proof follows from Farkas' lemma.

\subsection{Proof of Proposition~\ref{thm:weylmiknowskyrecursive}}
For completeness we provide here the proof of Proposition~\ref{thm:weylmiknowskyrecursive}.
Let $L_T=\otimes_{t=1}^TL^t$ and $K_T=\otimes_{t=1}^TK^t$. Note that for any $v$, $z$ and $\otimes_{t=1}^TK^t$ such that $\left(\otimes_{t=1}^TK_t\right)v=z$ is well-defined, we can construct $V$ and $Z$ such that columns of $V$ and $Z$ are subvectors\footnote{A vector $x$ is a subvector of $y=(y_j)_{j\in J}$, if $x=(y_j)_{j\in J'}$ for some $J'\subseteq J$.} of $v$ and $z$ and
\[
\left(\otimes_{t=1}^TK^t\right)v=z \iff K^T V \left(\otimes_{t=1}^{T-1}K^t\right)\tr=Z.
\]
Recall that by definition, $L^tK^tv\geq0$ for all $v\geq0$. Hence,
\begin{align*}
    &\forall v\geq0,\:L^1K^1v\geq0\implies \forall V\geq0,\: L^2K^2 V (L^1K^1)\tr\geq0 \iff \\
    &\forall v\geq0,\: (L^1K^1\otimes L^2K^2)v\geq0 \implies \forall V\geq0,\: L^3K^3 V (L^1K^1\otimes L^2K^2)\tr\geq0\iff\\
    &\forall v\geq0,\: (\otimes_{t=1}^3 L^tK^t)v\geq0 \implies \forall V\geq0,\: L^4K^4 V (\otimes_{t=1}^3 L^tK^t)\tr\geq0\implies\\
    &\dots \implies \forall v\geq0,\: (\otimes_{t=1}^T L^tK^t)v\geq0 \iff \forall v\geq0,\: L_TK_Tv\geq0.
\end{align*}
Hence,
\[
\{K_Tv\::\:v\geq0\}\subseteq \{z\::\:L_Tz \geq0\}.
\]

\subsection{Proof of Theorem~\ref{thm:sufficiencyminimaltensorequalmaximaltensor}}
The proof of the first statement follows directly from Theorem~1 in \citet{aubrun2022monogamy}. The proof of the moreover statement follows from Corollary~4 in \citet{aubrun2021entangleability}.
\par 
For self-containment we prove here directly the sufficiency in the moreover statement. Without loss of generality, assume that $K^t$ has full column rank for all $t$ except maybe $t=1$. By assumptions of the theorem, $K^t$ is proper and, thus, of full row rank for all $t$. Hence, $K_T$ has full row rank and we can represent any $z$ as a weighted sum of columns of $K_T$ (some weights may be negative). That is, $\{v\::\:K_Tv=z\}$ is nonempty for any $z$. Take any $z$ such that $L_Tz\geq 0$. We want to show that there exists $v\geq0$ such that $K_Tv=z$. Towards a contradiction, assume that $v\not\geq 0$ for all $v\in\{v\::\:K_Tv=z\}$. Take any $v\in\{v\::\:K_Tv=z\}$. Then
\begin{align*}
    &L_TK_Tv\geq0\implies L^TK^T V (L_{T-1}K_{T-1})\tr\geq0,
\end{align*}
where $V$ is constructed the same way it was constructed in the proof of Proposition~\ref{thm:weylmiknowskyrecursive}.
Since $K^T$ is invertible (full row and column rank), we observe that $V (L_{T-1}K_{T-1})\tr\geq0$. Take any row of $V$ that has a negative component and call the transpose of this row $v$. Then
\begin{align*}
    &L_{T-1}K_{T-1}v\geq0\implies L^{T-1}K^{T-1} V (L_{T-2}K_{T-2})\tr\geq0.
\end{align*}
Hence, $L_{T-1}K_{T-1}V\tr\geq0$. Repeating this step finitely many times we end up having $L^1K^1v\geq0$ for some $v\not\geq 0$. Since $K^1$ may not have full column rank, we only can conclude that there exists $v^*\geq0$ such that $K^1v=K^1v^*$. We can construct such $v^*\geq0$ for all possible subvectors of the original $v$. Combine these $v^*$s into $\bar{v}\geq0$. By definition, $K_T\bar{v}=z $ and $\bar{v}\geq0$. The latter is not possible since it was assumed that $v\not\geq 0$ for all $v\in\{v\::\:K_Tv=z\}$. The contradiction completes the proof.

\subsection{Proof of Theorem~\ref{thm: WM stable rho}}
First we show necessity of stability. By definition of DRUM, there exists a distribution over $\mathcal{U}$, $\mu$, such that
\[
\rho\left(\left(x_{i_t|j_t}\right)_{t\in \mathcal{T}}\right)=\int \prod_{t\in \mathcal{T}} \Char{\argmax_{y\in B^t_{j_t}}u^t(y)= x^t_{i_t|j_t}}d\mu(u)
\]
for all $\rand{i},\rand{j}$. Fix some $t'\in\mathcal{T}$, $x_{\rand{i}|\rand{j}}$, and $j_{t'}\in\mathcal{J}^{t'}$. Note that 
\begin{align*}
&\sum_{i\in\mathcal{I}^{t'}_{j_{t'}}}\rho\left(x_{\rand{i}|\rand{j}}\right)=\\
&\sum_{i\in \mathcal{I}_{j_{t'}}^{t'}} \int\Char{\argmax_{y\in B^{t'}_{j_{t'}}}u^{t'}(y)= x^{t'}_{i|j_{t'}}} \prod_{t\in \mathcal{T}\setminus\{t'\}} \Char{\argmax_{y\in B^t_{j_t}}u^t(y)= x^t_{i_t|j_t}}d\mu(u)=\\
&\int\sum_{i\in \mathcal{I}_{j_{t'}}^{t'}} \Char{\argmax_{y\in B^{t'}_{j_{t'}}}u^{t'}(y)= x^{t'}_{i|j_{t'}}} \prod_{t\in \mathcal{T}\setminus\{t'\}} \Char{\argmax_{y\in B^t_{j_t}}u^t(y)= x^t_{i_t|j_t}}d\mu(u)=\\
&\int \prod_{t\in \mathcal{T}\setminus\{t'\}} \Char{\argmax_{y\in B^t_{j_t}}u^t(y)= x^t_{i_t|j_t}}d\mu(u),
\end{align*}
where the last equality follows from $\argmax_{y\in B^{t'}_{j_{t'}}}u^{t'}(y)$ being a singleton and $\{x^{t'}_{i|j_{t'}}\}_{i\in\mathcal{I}_{j_{t'}}^{t'}}$ being a partition. The right-hand side of the last expression does not depend on the choice of $j_{t'}$. Stability follows from $t'$ and $x_{\rand{i}|\rand{j}}$ being arbitrary. 

Next, we show that any stable $\rho$ belongs to a linear span of columns of $A_T$. That is, the system $A_Tv=\rho$ always has a solution and the cone generated by $A_T$ is proper when restricted to stable $\rho$. Hence, Theorem~\ref{thm: WM stable rho} follows from Theorem~\ref{thm:sufficiencyminimaltensorequalmaximaltensor}. 

Before we formally show the existence of a solution, let us first replicate the proof in the simple setup with 2 time periods. Recall that 
\[
A^t=\left(\begin{array}{ccc}
     1&1&0  \\
     0&0&1  \\
     1&0&0  \\
     0&1&1  
\end{array}\right).
\]
First, construct the matrix $A^{t*}$ by removing the last row from $A^t$. That is, for every budget except the first one (the first two rows), remove the row that corresponds to the last patch of that budget (rows 3 and 4 correspond to the second budget). As a result,
\[
A^{t*}=\left(\begin{array}{ccc}
     1&1&0  \\
     0&0&1  \\
     1&0&0  \\
\end{array}\right).
\]
Put the removed row in the matrix $A^{t-}$. That is, $A^{t-}=(0\: 1\: 1)$. Note that $A^{t-}=G^tA^{t^*}$, where $G^{t}=(1\: 1\: -1)$. Moreover, $A_T=A^1\otimes A^2$ (see Table 3) can be partitioned into the matrix $A_T^{*}$ that contains the rows generated by rows of $A^{1*}$ and $A^{2*}$ ($A_T^*=A^{1*}\otimes A^{2*}$), and the matrix $A_T^{-}$ that contains the rest of the rows. That is,
\begin{align*}
    A_T=\left(\begin{array}{c}
     A^{1*}\\
     A^{1-}
\end{array}\right)\otimes\left(\begin{array}{c}
     A^{2*}\\
     A^{2-}
\end{array}\right)=\left(\begin{array}{c}
     A^{1*}\otimes A^{2*} \\
     A^{1*}\otimes A^{2-} \\
     A^{1-}\otimes A^{2*} \\
     A^{1-}\otimes A^{2-} 
\end{array}\right)=\left(\begin{array}{c}
     A_T^{*}\\
     A_T^{-}
\end{array}\right).
\end{align*}
Let $\rho^*$ and $\rho^{-}$ be the parts of $\rho$ that correspond to rows of $A_T^{*}$ and $A^{-}_T$. Since every element of $\rho$ corresponds to some choice path, $\rho^*$ does not contain choice paths that contain either $x^{1}_{2|2}$ or $x^{2}_{2|2}$ (we removed one row from $A^1$ and one row from $A^2$). Similarly, $\rho^{-}$ contains all choice paths where at least in one time period $t$ a patch was removed from $A^t$.

Note that $A^{t*}$, $t\in\mathcal{T}$, has full row rank. Hence, $A^*_T$, as a Kronecker product of full row rank matrices, is of full row rank as well. Thus, $v^*=A_T^{*\prime}(A_T^{*}A_T^{*\prime})^{-1}\rho^*$ exists and solves $A^*_Tv=\rho^*$. If we show that $A_T^{-}v^*=\rho^{-}$, then $v^*$ solves $A_Tv=\rho$ as well. Note that, 
\begin{align*}
&A^{1*}\otimes A^{2-} v^*=\left(A^{1*}\otimes G^2A^{2*} \right)v^*=
\left(\begin{array}{ccc}
     G^2&0&\dots\\
     0&G^2&\dots\\
     \dots&\dots&\dots\\
     \dots&0&G^2\\
\end{array}\right)\left(A^{1*}\otimes A^{2*}\right)v^*=diag(G^2)\rho^*,
\end{align*}
where $diag(L)$ is a block-diagonal matrix with matrix $L$ being on the main diagonal. The vector $\rho^*$ has 9 elements with the first 3 elements corresponding to choice paths that have $x^1_{1|1}$ and all possible patches that were not removed from $t=2$. That is, the first 3 elements of $\rho^*$ are $\rho\left(\left(x^1_{1|1},x^2_{1|1}\right)\right)$, $\rho\left(\left(x^1_{1|1},x^2_{2|1}\right)\right)$, and $\rho\left(\left(x^1_{1|1},x^2_{1|2}\right)\right)$ (the patch $x^2_{2|2}$ was removed). Thus, the first element of $diag(G^2)\rho^*$ is
\[
\rho\left(\left(x^1_{1|1},x^2_{1|1}\right)\right)+\rho\left(\left(x^1_{1|1},x^2_{2|1}\right)\right) - \rho\left(\left(x^1_{1|1},x^2_{1|2}\right)\right)=\rho\left(\left(x^1_{1|1},x^2_{2|2}\right)\right),
\]
where the equality follows from stability of $\rho$. Similarly, the second element of $diag(G^2)\rho^*$ is
\[
\rho\left(\left(x^1_{2|1},x^2_{1|1}\right)\right)+\rho\left(\left(x^1_{2|1},x^2_{2|1}\right)\right) - \rho\left(\left(x^1_{2|1},x^2_{1|2}\right)\right)=\rho\left(\left(x^1_{2|1},x^2_{2|2}\right)\right),
\]
and the third element is $\rho\left(\left(x^1_{1|2},x^2_{2|2}\right)\right)$. So $v^*$ solves the equations with only $x^2_{2|2}$ dropped.
 
Next, consider $A^{1-}\otimes A^{2*} v^*$. Note that all objects we work with (e.g., $A_T$ and $A_T^*$) are defined as a function of $\mathcal{T}$. Hence, if we push the time period $t$ to the very end (i.e., $1,\dots,t-1,t+1,\dots,T,t$), we still can define all objects for the new order of time labels. Let $W^t$ (with inverse $W^{t,-1}$, which pushes the last element of $\mathcal{T}$ to $t$-th position) be a transformation that recomputes all objects for the time span where $t$ is pushed to the end. For example, $W^1$ pushes the label $t=1$ to the end of $\mathcal{T}$ (i.e., $\mathcal{T}$ becomes $\{2,1\}$). Transformation $W^t$ satisfies the following three properties: $W^t[C]=C$ if $C$ does not depend on $\mathcal{T}$; $W^t[CD]=W^t[C]W^t[D]$ for any matrices $C$ and $D$; and $W^t[\otimes_{t'\in\mathcal{T}}A^{t'*}]=\otimes_{t'\in\mathcal{T}\setminus\{t\}}A^{t'*}\otimes A^{t*}$. Hence,
\begin{align*}
&A^{1-}\otimes A^{2*} v^*=W^{1,-1}\left[W^{1}\left[\left(A^{1-}\otimes A^{2*} \right)v^*\right]\right]=
W^{1,-1}\left[\left(A^{2*}\otimes A^{1-} \right)W^1\left[v^*\right]\right]=\\
&W^{1,-1}\left[diag(G^1)\left(A^{2*}\otimes A^{1*}\right)W^1\left[v^*\right]\right]=
W^{1,-1}\left[diag(G^1)W^{1}\left[W^{1,-1}\left[\left(A^{2*}\otimes A^{1*}\right)W^1\left[v^*\right]\right]\right]\right]=\\
&W^{1,-1}\left[diag(G^1)W^{1}\left[\left(A^{1*}\otimes A^{2*}\right)v^*\right]\right]=
W^{1,-1}\left[diag(G^1)W^{1}\left[\rho^*\right]\right].
\end{align*}
In words, $W^{1}\left[\rho^*\right]$ changes labels so that $t=1$ is the last one and reshuffles elements of $\rho^*$, then $diag(G^1)W^{1}\left[\rho^*\right]$ computes probabilities of choice paths where $x^2_{2|2}$ were dropped. Finally, $W^{1,-1}$ returns the original labeling. So the result is the subvector of $\rho^{-}$ where $x^1_{2,2}$ is dropped (relabelling changes $x^2_{2|2}$ to $x^1_{2|2}$). So $v^*$ solves the equations where only $x^1_{2|2}$ is dropped. 

Let $Y^{t}$ be an operator such that $Y^{t}[\cdot]=W^{t,-1}\left[diag(G^{t}) W^{t}[\cdot]\right]$. That is, $Y^t$ pushes $t$ to the end, multiplies the resulting object by $diag(G^t)$ and then pushes label $t$ back to its spot. Using operator $Y^t$ we can deduce that
\begin{align*}
&A^{1-}\otimes A^{2-} v^*=Y^1\left[Y^2\left[\rho^*\right]\right]=\\
&\rho\left(\left(x^1_{2|2},x^2_{1|1}\right)\right)+\rho\left(\left(x^1_{2|2},x^2_{2|1}\right)\right)-\rho\left(\left(x^1_{2|2},x^2_{1|2}\right)\right)=\rho\left(\left(x^1_{2|2},x^2_{2|2}\right)\right),
\end{align*}
where the last equality follows from stability of $\rho$.
Hence, the equation where both $x^1_{2|2}$ and $x^2_{2|2}$ were dropped is also solved by $v^*$. 
\par
Next, we generalize the above arguments for arbitrary $T$ and $A^t$. Consider the following modification of $A^t$, $t\in\mathcal{T}$. From every menu, except the first one, we pick the last alternative and remove the corresponding row from $A^t$. Let $A^{t*}$ denote the resulting matrix. Thus, matrix $A^t$ can be partitioned into $A^{t*}$ and $A^{t-}$, where rows of $A^{t-}$ correspond to alternatives removed from $A^t$. Consider the first row of $A^{t-}$. It corresponds to the last alternative from the second menu at time $t$. Note that the sum of all rows that correspond to the same menu is equal to the row of ones. Hence, the first row of $A^{t-}$ is equal to the sum of the rows that correspond to menu $1$ minus the sum of the remaining rows in menu $2$. That is, the first row of $A^{t-}$ can be written as
\[
(1,\dots,1,-1,\dots,-1,0,\dots,0)A^{t*}.
\]
Similarly, the second row of $A^{t-}$ can be written as
\[
(1,\dots,1,0,\dots,0,-1,\dots,-1,0,\dots,0)A^{t*}.
\]
In matrix notation, we can rewrite $A^{t-}$ as $A^{t-}=G^tA^{t*}$, where $G^t$ is the matrix with the $k$-th row having the elements that correspond to the alternatives from the first menu at time $t$ are equal to $1$, the elements that correspond to the alternatives from the $k$-th menu are equal to $-1$, and the rest of elements are equal to $0$. 

Next note that, up to a permutation of rows, $A_T$ can be partitioned into $A_T^*=\otimes_{t\in\mathcal{T}}A^{t*}$ and matrices of the form $\otimes_{t\in\mathcal{T}}C^t$, where $C^t\in\{A^{t*},A^{t-}\}$, with $C^t=A^{t-}$ for at least one $t$. We will stack all these matrices into $A_T^{-}$. Next, let $\rho^*$ denote the subvector of $\rho$ that corresponds to choice paths that do not contain any of the alternatives removed from $A^t$, $t\in\mathcal{T}$. Thus, $\rho=(\rho^{*\prime},\rho^{-\prime})\tr$, where $\rho^{-}$ corresponds to all elements of $\rho$ that contain at least one of the removed alternatives. As a result, we can split the original system into two: $A^*_Tv=\rho^*$ and $A^{-}_Tv=\rho^{-}$.

Consider the system $A^*_Tv=\rho^*$. We formally prove later that $A^{t*}$ has full row rank for all $t$.  Then $A^*_T$ is also of full row rank and, hence, $A^*_TA^{*\prime}$ is invertible and $v^*=A^{*\prime}\left(A^*_TA^{*\prime}\right)^{-1}\rho^*$ solves the system. If we show that
\[
A^{-}_Tv^*=\rho^{-},
\]
then we prove that $A_Tv=\rho$ always has a solution, which will complete the proof. 

Note that $A^{-}_T$ consists of the blocks of the form $\otimes_{t\in\mathcal{T}}C^t$, where $C^t\in\{A^{t*},A^{t-}\}$ and $C^t=A^{t-}$ for at least one $t$. 
Next note that for any $A$, $B$, and $C$
\[
A\otimes(BC)=diag(B)(A\otimes C),
\]
where $diag(B)$ is the block-diagonal matrix constructed from $B$. Indeed,
\begin{align*}
    A\otimes(BC)=\left(\begin{array}{ccc}
         A_{11}BC&A_{12}BC& \dots \\
         A_{21}BC&A_{22}BC&  \dots\\
         \dots & \dots& \dots 
    \end{array}\right)=\left(\begin{array}{ccc}
         B&0& \dots \\
         0&B&  \dots\\
         \dots & \dots&B 
    \end{array}\right)\left(\begin{array}{ccc}
         A_{11}C&A_{12}C& \dots \\
         A_{21}C&A_{22}C&  \dots\\
         \dots & \dots& \dots 
    \end{array}\right) \\
    =diag(B)(A\otimes C)
\end{align*}

First, consider $\otimes_{t\in\mathcal{T}}C^t$, where $C^t\in\{A^{t*},A^{t-}\}$ and $C^t=A^{t-}$ for only one $t$. Hence,
\begin{align*}
&\otimes_{t'\in\mathcal{T}}C^{t'}v^*=
W^{t,-1}\left[diag(G^t)W^{t}\left[\rho^*\right]\right]=Y^t[\rho^*].
\end{align*}
Note that because $\rho$ is stable, $diag(G^{T})\rho^*$ is the subvector of $\rho^{-}$ that corresponds to choice paths that contain one of the removed alternatives from the last period only. So, $W^{t}\left[\rho^*\right]$ first pushes the period $t$ to the very end, then $diag(G^t)W^{t}\left[\rho^*\right]$ computes the elements of $\rho^-$, and finally $W^{t,-1}\left[diag(G^t)W^{t}\left[\rho^*\right]\right]$ moves the time period $t$ back to its place.

Next, consider $\otimes_{t\in\mathcal{T}}C^t$, where $C^t\in\{A^{t*},A^{t-}\}$ and $C^t=A^{t-}$ and $C^{t'}=A^{t'-}$  for two distinct $t,t'$. Similarly to the previous case,
\begin{align*}
&\otimes_{t'\in\mathcal{T}}C^{t'}v^*=
W^{t,-1}\left[diag(G^t)W^{t}\left[W^{t',-1}\left[diag(G^{t'})W^{t'}\left[\rho^*\right]\right]\right]\right]=Y^t[Y^{t'}[\rho^*]]=Y^t\circ Y^{t'}[\rho^*],
\end{align*}
where $Y^t\circ Y^{t'}$ denotes the composite operator.
Again, $W^{t',-1}\left[diag(G^{t'})W^{t'}\left[\rho^*\right]\right]$ computes the subvector of $\rho^-$ that corresponds to choice paths where an alternative from only one time $t'$ was missing. Applying to the resulting vector $W^{t,-1}\left[diag(G^t)W^{t}\left[\cdot\right]\right]$ computes the subvector of $\rho^{-}$ with alternatives missing from $t$ and $t'$ only. Repeating the arguments for all possible rows of $A_T^{-}$, we obtain that 
\begin{align*}
&\otimes_{t'\in\mathcal{T}}C^{t'}v^*=\circ_{t':C^{t'}=A^{t'-}}Y^{t'}[\rho^*]
\end{align*}
and, thus,
$A^{-}_T v^*=\rho^{-}$. Hence, $v^*$ is a solution to $A_Tv=\rho$. 

It is left to show that $A_t^{*}$ is a       full row rank matrix for all $t$. To do so, we first prove the same result for a more general version of static RUM with ``virtual'' budgets introduced in Section~\ref{sec: BM charachterization}.

Let $\bar{\mathcal{R}}^{t}$ be the set of all linear orders on $\mathbf{X}^t$. For any $j_t\in\bar{\mathcal{J}}^t,i_t\in\mathcal{I}^t_{j_t}$, and $\succ\in\bar{\mathcal{R}}^t$ let
\[
\bar{a}^t_{\succ}=\left(\Char{x^t_{i_t|j_t}\succ x,\:\forall x\in B^t_{j_t}}\right)_{j_t\in\bar{\mathcal{J}}^t,i_t\in\mathcal{I}^t_{j_t}}
\]
be the vector of 0s and 1s that denote the best patch in every virtual budget. Analogously to $A^t$, let $\bar{A}^t$ denote the matrix which columns are $\left\{a^t_{\succ}\right\}_{\succ\in\bar{\mathcal{R}}^t}$, and $\bar{A}^{t*}$ be the matrix constructed from $\bar{A}^t$ by removing the rows that correspond to the last patch in every budget but the first one.

\begin{lemma}\label{lemma: BM full rank}
	$\bar{A}^{t*}$ has full row rank.
\end{lemma}   
\begin{proof}
Take $\mathcal{T}=\{t\}$. By Corollary~2 in \citet{saito2017axiomatizations} or Theorem~2 in \citet{dogan2022every}, for any $\bar{\rho}\geq0$ such that the sum over patches in any budget is equal to 1, there exists $\nu$ such that
\[
\bar{A}^t\nu=\bar{\rho}.
\]
Since $\bar{A}^t\alpha \nu=\alpha \bar{A}^t\nu=\alpha \bar{\rho}$ for any $\alpha\in\Real$, $\bar{A}^t\nu$ can be any positive or any negative vector such that the sum over all patches in each budget does not depend on a budget (i.e. satisfies stability). Moreover, since any vector can be written as the sum of a positive and a negative vector, $\bar{A}^t\nu$ can represent any $\bar{\rho}$ such that sums over budgets are budget independent. Hence, if we remove the last row from every budget except the first one, we obtain that for any vector $\bar{\rho}^*$, there exists $\nu$ such that $\bar{A}^{t*}\nu=\bar{\rho}^{*}$. Thus, $\bar{A}^{t*}$ is of full row rank. Indeed, if it was not, then there would exist $\xi\neq0$ such that $\xi\tr \bar{A}^{t*}\nu=0\cdot \nu=0$ for all $\nu$. Therefore, we would have $\xi\tr \bar{A}^{t*}\nu=\xi\tr \bar{\rho}^*=0$ for all $\bar{\rho}^*$. This can only be true if $\xi = 0$, contradicting $\xi\neq0$.
\end{proof}

\subsection{Proof of Theorem~\ref{thm:DRUM2x2}}

\textbf{Necessity.} Suppose that $\rho$ is rationalized by DRUM.

\textbf{Necessity of stability.} Follows from Theorem~\ref{thm: WM stable rho}.

\noindent\textbf{Necessity of $\mathrm{D}$-monotonicity.} Note that $A_T=A^{1}\otimes A_{T-1}$, where  
\[
A^1=\left( 
\begin{array}{ccc}
    1 & 1 & 0  \\
    0 & 0 & 1  \\
    1 & 0 & 0  \\
    0 & 1 & 1 
\end{array}
\right).
\]
Note that, since in every time period there are only $2$ budgets, if $\rho$ is rationalized by DRUM, then by Theorem~\ref{thm:main}, there exists component-wise nonnegative $\nu$ (i.e. $\nu\geq 0$) such that
$
A_T\nu=\rho.
$
We next show that this $\nu\geq 0$ together with stability, which we already showed to be satisfied, implies $\mathrm{D}$-monotonicity. 

First, note that we can partition $\nu$ into $3$ vectors ($\nu^1_1$, $\nu^1_2$, and $\nu^1_{3}$) and $\rho$ into $4$ vectors ($\rho^1_{1|1}$, $\rho^1_{2|1}$, $\rho^1_{1|2}$, and $\rho^1_{2|2}$) such that
\[
\left( \begin{array}{c}
    \rho^1_{1|1}\\
    \rho^1_{2|1}\\
    \rho^1_{1|2}\\
    \rho^1_{2|2}
\end{array}
\right)=\rho=A_T\nu=A_1\otimes A_{T-1}\nu=\left( \begin{array}{ccc}
    A_{T-1} & A_{T-1} & 0 \\
    0 & 0 & A_{T-1}  \\
    A_{T-1} & 0 & 0 \\  
    0 & A_{T-1} & A_{T-1}  
\end{array}
\right) \left( \begin{array}{c}
    \nu^1_{1}\\
    \nu^1_{2}\\
    \nu^1_{3}
\end{array}
\right)=\left( \begin{array}{c}
    A_{T-1}(\nu^1_{1}+\nu^1_{2})\\
    A_{T-1}\nu^1_{3}\\
    A_{T-1}\nu^1_{1}\\
    A_{T-1}(\nu^1_{2}+\nu^1_{3})
\end{array}
\right).
\]
In this representation, $\rho^1_{i|j}$ correspond to all choice paths that contain patch $x^1_{i|j}$. Subtracting the third line from the first one, and the second line from the fourth one in the last system of equations, we obtain that
\begin{align*}
    \rho^1_{1|1}-\rho^1_{1|2}=\rho^1_{2|2}-\rho^1_{2|1}=A_{T-1}\nu^1_2\geq 0,
\end{align*}
where the last inequality follows from $\nu\geq 0$ and $A_{T-1}$ consisting of zeros and ones. Thus, $\mathrm{D}\left(x^1_{i'_1|j'_1}\right)\left[\rho(x_{\rand{i}|\rand{j}})\right]\geq 0$ if $x^1_{i_1'|j_1'}>^D x^1_{i_1|j_1}$. 

Applying the above arguments to $\rho^1_{1|1}-\rho^1_{1|2}=A_{T-1}\nu^1_2$, we obtain that 
\begin{align*}
    (\rho^1_{1|1,1|1}-\rho^1_{1|2,1|1})-(\rho^1_{1|1,1|2}-\rho^1_{1|2,1|2})=(\rho^1_{2|2,2|2}-\rho^1_{2|1,2|2})-(\rho^1_{2|2,2|1}-\rho^1_{2|1,2|1})=A_{T-2}\nu^2_2\geq 0,
\end{align*}
where $\rho^1_{i|j,i'|j'}$ corresponds to all choice paths that contain patches $x^1_{i|j}$ and $x^2_{i'|j'}$. Thus, $\mathrm{D}\left(x^2_{i_2'|j_2'}\right)\mathrm{D}\left(x^1_{i_1'|j_1'}\right)\left[\rho(x_{\rand{i}|\rand{j}})\right]\geq 0$ if $x^1_{i_1'|j_1'}>^D x^1_{i_1|j_1}$ and $x^2_{i_2'|j_2'}>^D x^2_{i_2|j_2}$. Repeating these steps we can get that for all $K\leq T$
\[
\mathrm{D}\left(x^K_{i_K'|j_K'}\right)\dots\mathrm{D}\left(x^2_{i_2'|j_2'}\right)\mathrm{D}\left(x^1_{i_1'|j_1'}\right)\left[\rho(x_{\rand{i}|\rand{j}})\right]\geq 0
\]
if $x^t_{i_t'|j_t'}>^D x^t_{i_t|j_t}$ for all $t=1,\dots,K$. Note that for any permutation of time periods the matrix $A_T$ does not change. Hence, the above steps can be performed for any permutation of $x_{\rand{i}|\rand{j}}$ and $\mathrm{D}$-monotonicity is satisfied.   

\textbf{Sufficiency.} Assume that $\rho$ is stable and $\mathrm{D}$-monotone. 
Define $H_L=H^{1}\otimes H_{L-1}$ and $P_{A_T}=P_{A_1}\otimes P_{A_{T-1}}$, where 
\[
H_1=(A_1\tr A_1)^{-1}A_1\tr=\left( \begin{array}{cccc}
    0.25 & 0.25 & 0.75 & -0.25  \\
    0.5 & -0.5 & -0.5 & 0.5  \\
    -0.25 & 0.75 & 0.25 & 0.25  
\end{array}
\right)
\]
and
\[
P_{A_1}=A_1(A_1\tr A_1)^{-1}A_1\tr=\left( \begin{array}{cccc}
    0.75 & -0.25 & 0.25 & 0.25  \\
    -0.25 & 0.75 & 0.25 & 0.25  \\
    0.25 & 0.25 & 0.75 & -0.25  \\
    0.25 & 0.25 & -0.25 & 0.75  
\end{array}
\right).
\]
If we show that $\nu=H_T\rho$ satisfies (i) $\nu\geq 0$ and (ii) $A_T\nu=\rho$, then by Theorem~\ref{thm:main} $\rho$ is rationalized by DRUM.  

\textbf{Step 1: $\nu\geq 0$.} Note that
\[
\left( \begin{array}{c}
    \nu^1_{1}\\
    \nu^1_{2}\\
    \nu^1_{3}
\end{array}
\right)=\nu=H_T \rho=H_1\otimes H_{T-1}\rho=\left( \begin{array}{cccc}
    0.25H_{T-1} & 0.25H_{T-1} & 0.75H_{T-1} & -0.25H_{T-1}  \\
    0.5H_{T-1} & -0.5H_{T-1} & -0.5H_{T-1} & 0.5H_{T-1}  \\
    -0.25H_{T-1} & 0.75H_{T-1} & 0.25H_{T-1} & 0.25H_{T-1}  
\end{array}
\right) \left( \begin{array}{c}
    \rho^1_{1|1}\\
    \rho^1_{2|1}\\
    \rho^1_{1|2}\\
    \rho^1_{2|2}
\end{array}
\right).
\]
Applying stability (i.e., $\rho^1_{1|2}+\rho^1_{2|2}=\rho^1_{1|1}+\rho^1_{2|1}$), we can conclude that
\[
\left( \begin{array}{c}
    \nu^1_{1}\\
    \nu^1_{2}\\
    \nu^1_{3}
\end{array}
\right)=\left( \begin{array}{c}
    H_{T-1}\rho^1_{1|2}\\
    H_{T-1}(\rho^1_{1|1}-\rho^1_{1|2})\\
    H_{T-1}\rho^1_{2|1}
\end{array}
\right).
\]
If we next apply the above steps to $\nu_1^1=H_{T-1}\rho^1_{1|2}$, then we can obtain that
\[
\left( \begin{array}{c}
    \nu^1_{11}\\
    \nu^1_{12}\\
    \nu^1_{13}
\end{array}
\right)=\left( \begin{array}{c}
    H_{T-2}\rho^1_{1|2,1|2}\\
    H_{T-2}(\rho^1_{1|1,1|2}-\rho^1_{1|2,1|2})\\
    H_{T-2}\rho^1_{1|2,2|1}
\end{array}
\right),
\]
where $\rho^1_{i|j,i'|j'}$ corresponds to all choice paths that contain patches $x^1_{i|j}$ and $x^2_{i'|j'}$. If, instead, we apply it to $\nu_2^1=H_{T-1}(\rho^1_{1|1}-\rho^1_{1|2})$, then we obtain
\[
\left( \begin{array}{c}
    \nu^1_{21}\\
    \nu^1_{22}\\
    \nu^1_{23}
\end{array}
\right)=\left( \begin{array}{c}
    H_{T-2}(\rho^1_{1|1,1|2}-\rho^1_{1|2,1|2})\\
    H_{T-2}((\rho^1_{1|1,1|1}-\rho^1_{1|2,1|1})-(\rho^1_{1|1,1|2}-\rho^1_{1|2,1|2}))\\
    H_{T-2}(\rho^1_{1|1,2|1}-\rho^1_{1|2,2|1})
\end{array}
\right).
\]
Repeating the above steps $T$ times, we obtain that every component of $\nu$ is either equal to $\rho\left(\left(x^t_{1|2}\right)_{t\in\mathcal{T}}\right)\geq 0$, or $\rho\left(\left(x^t_{2|1}\right)_{t\in\mathcal{T}}\right)\geq 0$, or
\[
\mathrm{D}\left(x^{\rand{t}}_{\rand{i}'|\rand{j}'}\right)\left[\rho\left(x_{\rand{i}|\rand{j}}\right)\right]
\geq 0,
\]
for some $\rand{t}\in\rands{\mathcal{T}}$ and some $x_{\rand{i}|\rand{j}}$. 
The last inequality follows from  $x^t_{i_t'|j_t'}>^D x^t_{i_t|j_t}$ for all $t\in \rand{t}$ and $\mathrm{D}$-monotonicity. Hence, the proposed $\nu$ is nonnegative.

\textbf{Step 2: $A_T\nu=\rho$.} Note that
\[
A_T\nu=P_{A_T}\rho=\left( \begin{array}{cccc}
    0.75P_{A_{T-1}} & -0.25P_{A_{T-1}} & 0.25P_{A_{T-1}} & 0.25P_{A_{T-1}}  \\
    -0.25P_{A_{T-1}} & 0.75P_{A_{T-1}} & 0.25P_{A_{T-1}} & 0.25P_{A_{T-1}}  \\
    0.25P_{A_{T-1}} & 0.25P_{A_{T-1}} & 0.75P_{A_{T-1}} & -0.25P_{A_{T-1}}  \\
    0.25P_{A_{T-1}} & 0.25P_{A_{T-1}} & -0.25P_{A_{T-1}} & 0.75P_{A_{T-1}}  
\end{array}
\right) \left( \begin{array}{c}
    \rho^1_{1|1}\\
    \rho^1_{2|1}\\
    \rho^1_{1|2}\\
    \rho^1_{2|2}
\end{array}
\right).
\]
Since stability implies that $\rho^1_{1|1}+\rho^1_{2|1}=\rho^1_{1|2}+\rho^1_{2|2}$, we obtain
\[
A_T\nu=P_{A_T}\rho=\left( \begin{array}{c}
    P_{A_{T-1}}\rho^1_{1|1}\\
    P_{A_{T-1}}\rho^1_{2|1}\\
    P_{A_{T-1}}\rho^1_{1|2}\\
    P_{A_{T-1}}\rho^1_{2|2}
\end{array}
\right).
\]
Repeating the above step one more time we obtain that 
\[
P_{A_{T-1}}\rho^1_{i|j}=\left( \begin{array}{c}
    P_{A_{T-2}}\rho^1_{i|j,1|1}\\
    P_{A_{T-2}}\rho^1_{i|j,2|1}\\
    P_{A_{T-2}}\rho^1_{i|j,1|2}\\
    P_{A_{T-2}}\rho^1_{i|j,2|2}
\end{array}
\right),
\]
where $i,j\in\{1,2\}$. Repeating the above steps $T$ times for each subvector, we obtain
\[
P_{A_T}\rho=\rho.
\]
Hence, $A\nu=\rho$.
\subsection{Proof of Theorem~\ref{thm: static RUM-BM}}
\begin{proof}
    \emph{(i) implies (ii).}
    If $\rho$ is consistent with RUM, then there exists an increasing random utility function $u^t$ distributed according to $\mu$ such that $\mu(\argmax_{y\in B^t_{j_t}}u^t(y)= x_{i_t|j_t})=\rho(x_{i_t|j_t})$ for all $j_t\in\mathcal{J}^t$ and $i_t\in\mathcal{I}^t_{j_t}$. Using this random $u^t$ we can extend $\rho$ to $\bar{\rand{J}}^t$, so the BM inequalities are satisfied and the constructed $\bar{\rho}$ agrees with $\rho$. It is left to show that $\bar{\rho}$ is IU-consistent. Towards a contradiction, assume that there exists a menu, $B^t_{j_t}$, and $x_{i_t|j_t}$ in it such that $\bar{\rho}(x_{i_t|j_t})>0$, and that for all $x_{i_t|j_t}$ there exists some  $S\subseteq B^t_{j_t}$  such that $S >^{t} x_{i_t|j_t}$. This is impossible since $u^t$ is assumed to be a monotone function on $>^{t}$, so no monotone function would choose a point in $x_{i_t|j_t}$ when better points are available in other patches. This contradiction completes the proof.
    \par 
    \emph{(ii) implies (i).}
    Let $\bar{\mathcal{R}}^t$ be the set of linear orders on $X^t$. By the result in \citet{falmagne1978representation}, we know that there is a $\nu\in\Delta\left(\bar{\mathcal{R}}^t\right)$ such that 
    \[
    \overline{\rho}(x_{\rand{i}|\rand{j}})=\sum_{\succ\in \bar{\mathcal{R}}^t}\nu(\succ)\Char{x_{\rand{i}|\rand{j}}\succ y_{\rand{i}|\rand{j}}\quad \forall y_{\rand{i}|\rand{j}}, \rand{i}\in\rand{I}_{\rand{j}}}.
    \]
    It only remains to show that the stochastic demand generated by the mixture of linear orders on the extended set of menus assigns a zero measure to linear orders that are not extensions of  $>^{t}$. Since $\rho$ is IU-consistent, $\nu(\succ)=0$ for any $\succ \in \bar{\mathcal{R}}^t$ that is not an extension of the order $>^{t}$. To show this is true, we prove the contrapositive. Namely, if $\nu(\succ)>0$ for some $\succ \in \bar{\mathcal{R}}^t$ that is not an extension of the order $>^{t}$, then there exist $S,S'\subseteq \mathbf{X}^t$ such that $S>^{t}S'$ yet there is a $y\in S'
    $ such that $y\succ x$ for all $x\in S$. Thus, IU-consistency fails for the virtual budget $\{y,x\}$.  
    \par 
    \emph{(ii) is equivalent to (iii).} The statement follows from the definition of matrix $\bar{H}^t$.
\end{proof}
\subsection{Proof of Theorem~\ref{thm: DRUM-BM}}
\begin{proof}
    \emph{(i) implies (ii).}
    Direct from arguments analogous to those made in Theorem~\ref{thm: static RUM-BM}, Theorem~\ref{thm:weylmiknowskyrecursive}, and Theorem~\ref{thm: WM stable rho}. 
    \par
    \emph{(ii) implies (i).} 
    We break the proof into two steps. 
    \par 
    First step. Let $\bar{\mathcal{R}}$ be the set of linear order profiles in $\times_{t\in\mathcal{T}}\mathbf{X}^t$, with typical element $(\succ^t)_{t\in\mathcal{T}}$. For any $\bar{\rho}$ such that satisfy $
        \bar{\rho}\in \bigcap_{k_1,\cdots,k_T\geq1}\left\{\Gamma^{\rands{\phi}^*\prime}_{\rand{k}}z\::\:\left(\otimes_{t\in\mathcal{T}}\bar{H}^{t,\otimes_{k_t}}\right)z\geq0\right\}$ and  stability, we can use the results in Theorem~\ref{thm:sufficiencyminimaltensorequalmaximaltensor} and the results in Theorem~\ref{thm: WM stable rho}, to ensure that there exists a $\nu\in \Delta(\bar{\mathcal{R}})$ such that
    \[
    \overline{\rho}(x_{\rand{i}|\rand{j}})=\sum_{(\succ^t)_{t\in\mathcal{T}}\in \mathcal{R}^*}\nu((\succ^t)_{t\in\mathcal{T}})\Char{x^t_{i_t|j_t}\succ^t y\quad \forall y\in B^t_{j_t}\quad\forall t\in\mathcal{T}}.
    \]
    Since $\rho$ is IU-consistent, $\nu((\succ^t)_{t\in\mathcal{T}})=0$ for any $(\succ^t)_{t\in\mathcal{T}} \in \bar{\mathcal{R}}$ that contains some element $\succ^t$ that is not an extension of the  order $>^{t}$. To show this is true, we prove the contrapositive. Namely, if $\nu((\succ^t)_{t\in\mathcal{T}})>0$ for some $(\succ^t)_{t\in\mathcal{T}} \in \bar{\mathcal{R}}$ that is not an extension of the order $>^{t}$, then there exist nonempty $S,S'\subseteq \mathbf{X}^t$ such that $S>^{t}S'$ yet for an element $y\in S'$ $y\succ x$ for all $x\in S$. Thus, IU-consistency fails for the virtual budget path that contains the budget $\{y,x\}$, for any selection of $x\in S$, at time $t$. 
\end{proof}
\subsection{Proof of Proposition~\ref{prop:constantDRUMimpliesSARP}}
\begin{proof}
We provide here the proof of Proposition~\ref{prop:constantDRUMimpliesSARP}. Assume towards a contradiction that $\rho$ is rationalized by DRUM with $\mu$ that satisfies constancy and SARPD is violated for some $\rand{j}$. Hence, there exist some $y^{t_1}$, $y^{t_N}$, and some $u\in U$ such that $u\left(y^{t_1}_{i_t|j_t}\right)>u\left(x^{t_{N}}_{i_{t_n}|j_{t_n}}\right)$. However, the violation of SARPD implies that  $u\left(x^{t_1}_{i_{t_1}|j_{t_1}}\right)>u\left(x^{t_1}_{i_{t_1}|j_{t_1}}\right)$ which is impossible. SARPD rules out the possibility that some individuals in the population violate SARP.  When constancy is relaxed, we need to obtain cross-sectional variation (i.e., more than one budget path) to test DRUM.     
\end{proof}

\subsection{Proof of Proposition~\ref{prop:DRUMimpliesmarginalRUM}}
\begin{proof}

Let, $\rho^{-t}((x_{i_{\tau}|j_{\tau}}^{\tau})_{\tau\in\mathcal{T}\setminus{\{t\}}})=\left(\sum_{i\in\mathcal{I}^t_j}\rho\left(x_{\rand{i}|\rand{j}}\right)\right)$. We also define the vector 
\[
\rho^{-\tau}=(\rho((x_{i_{t}|j_{t}}^{t})_{t\in\mathcal{T}\setminus\{\tau\}}))_{\mathbf{j}\in\mathbf{J},\mathbf{i}\in\mathbf{I}_{j}}.
\]
Note that $\rho^{-1}$ is of the same length that $\rho^1_{i|j}$ for any patch $x_{i|j}^1$. We let $\mathcal{R}_t$ be the set of linear orders at time $t\in\mathcal{T}$. The scalar $a_{t,r_t,i_k,j_k}$ is the entry of matrix $A_t$ for column corresponding to $r_t$ and row corresponding to $i_k,j_k$.
\begin{lemma} If the vector representation of $P$, $\rho$, is consistent with DRUM, then for every finite sequence of patches (including repetitions), $k$, $\{(i_k,j_k)\}$ such that $j_k\in\mathcal{J}^t$ and $i_k\in\mathcal{I}_{{j}_k}^t$ 
\[
\sum_{k}\rho^1_{i_k|j_k}\leq \rho^{-1} \max_{r_t\in\mathcal{R}_t}\sum_{k}a_{t,r_t,i_k,j_k}.
\]  
\end{lemma}
The condition above implies the fact that marginals, conditionals are consistent with RUM. Assume that $\rho$ is interior (i.e., rule out zero probabilities on choice paths), then the condition above implies that the marginal probability 
\[
\rho((x_{i_t|j_t}^t)_{t\in\mathcal{T}})|(x_{i_{\tau}|j_{\tau}}^{\tau})_{\tau \in\mathcal{T}\setminus{\{1\}}})=\frac{\rho((x_{i_t|j_t}^t)_{t\in\mathcal{T}})}{\rho((x_{i_{\tau}|j_{\tau}}^{\tau})_{\tau \in\mathcal{T}\setminus{\{1\}}})},
\] 
is consistent with (static) RUM. In that case the condition above is just the ASRP of \citet{mcfadden1990stochastic}. It is easy to see that the same reasoning can be done recursively and for any permutation of time, so all conditional probabilities of choice, as defined above, are consistent with (static) RUM if the vector representation $\rho$ is consistent with DRUM.    
\end{proof}

\section{Monte Carlo Simulations: Statistical Test of DRUM}\label{appendix: montecarlo}

Here we provide a Monte Carlo study to evaluate the performance of KS's test when applied to DRUM in finite samples. We consider both the demand setup and the binary menus setup. 
\subsection{Power Analysis: Demand Setup}
We consider the simple setup with $K = T = J^t = 2$. We set the number of DMs per choice path to ${N_{\rand{i}|\rand{j}}} \in \{50,500,5000\}$ and the number of simulations for each data generating process (DGP) to $1000$. The critical value for each test statistic is computed using $999$ bootstrap samples. As recommended in KS, the tuning parameter $\tau_{N}$ is set to $\tau_{N} = \sqrt{\log(4N_{\rand{i}|\rand{j}})/4N_{\rand{i}|\rand{j}}}$ (given that there are four choice paths in every budget path, $4N_{\rand{i}|\rand{j}}$ is the sample size of each budget path).

First, we consider a dynamic random Cobb-Douglas utility model. The utility function is given by
\begin{equation*}
    u_t(y_t) = y_{1,t}^{\alpha_t} y_{2,t}^{1-\alpha_t},
\end{equation*}
where $\alpha_t \in (0,1)$. Budgets in both periods are the same and correspond to prices $(2,1)\tr$ and $(1,2)\tr$ with an expenditure of $1$. We consider two DGPs for random $\alpha=(\alpha_1,\alpha_2)\tr$. 
\begin{align*}
    &\text{DGP1: } \alpha_1\sim U[0,1];\quad \alpha_{2} = \max\{ \min\{ 0.9\alpha_{1} + \epsilon_{1},1\},0\},\: \epsilon_{1}\sim N(0,25)\\
    &\text{DGP2: } \alpha_t=\mathrm{arctan}(\varepsilon_t)/\pi+1/2,\: t=1,2;\quad \varepsilon=(\varepsilon_1,\varepsilon_2)\tr\sim N(0,V)
\end{align*}
where 
\[
V=\left(\begin{array}{cc}
     1&0.5  \\
     0.5&1 
\end{array}\right).
\] 
Both DGPs are consistent with DRUM. The rejection rates at the $5$ percent significance level for all three sample sizes and both DGPs are presented in Table~\ref{table: null}.
\begin{table}[htbp]
\centering
\begin{tabular} {l|l|ccp{2cm}p{2cm}p{2cm}p{2cm}p{2cm}p{2cm}} DGP & $N_{i|j}$ & Rejection rate, \% \\ \hline
\multirow{2}{*}{DGP1} & 50 & 3.4 \\
& 500 & 4.3 \\
& 5000 & 5.1 \\ \hline
\multirow{ 2}{*}{DGP2} & 50 & 3.7  \\
& 500 & 4.6 \\
& 5000 & 5.4   \\
\end{tabular}
\caption{Every entry represents the rejection rate at the $5$ percent significance level and is computed from $1000$ simulations and $999$ bootstraps per simulation.} \label{table: null}
\end{table}
The rejection rates are close to $5$ percent even for small sample sizes. To analyze the finite sample power of the test, we consider the DGP described in Table~\ref{tab:marginalRumNotDrum}. Recall that this $\rho$ fails both $\mathrm{D}$-monotonicity and stability. The rejection rate is $100$ percent for all sample sizes. It is remarkable that $\rho$ in Table~\ref{tab:marginalRumNotDrum} has marginal probabilities consistent with RUM. Yet, even at small sample sizes such as $N_{\rand{i}|\rand{j}}=50$, the rejection rate is $100$ percent. These simulations show that KS's test for DRUM has good size and power properties in finite samples in the demand setup.

\subsection{Power Analysis: Mimicking the Empirical Application}

We provide a Monte Carlo study to evaluate the performance of KS's test in a simulated environment mimicking our application with binary menus. This exercise is important because the number of observations per budget path is moderate.  Hence, the asymptotic performance of the statistical test derived in the previous section may not translate to our application.  We consider three DGPs given by
\begin{align*}
    \rho_{1}^{t} &=
    \begin{bmatrix}
           1/5 \\
           4/5 \\
           4/5 \\ 
           1/5 \\
           1/5 \\
           4/5
    \end{bmatrix} , \;
    \rho_{2}^{t} = \begin{bmatrix}
           1/5 \\
           4/5 \\
           1/2 \\ 
           1/2 \\
           1/5 \\
           4/5
    \end{bmatrix} , \;
         \rho_{3}^{t} = \begin{bmatrix}
           1/4 \\
           3/4 \\
           2/4 \\ 
           2/4 \\
           1/4 \\
           3/4
    \end{bmatrix}.
\end{align*}

\noindent
Recall that the $\mathcal{H}$-representation of RUM is given by $H^{t}\rho \geq 0$, where $H^{t}$ is given by Table \ref{table:H binary} in our application. It is easy to check that the following hold:
\begin{align*}
    H^{t}\rho_{1}^{t} &= [-0.4, 1.4, 1.4, -0.4, -0.4, 1.4]'; \\
    H^{t}\rho_{2}^{t} &= [-0.1, 1.1, 1.1, -0.1, -0.1, 1.1]'; \\
    H^{t}\rho_{3}^{t} &= [0, 1, 1, 0, 0, 1]'.
\end{align*}

\noindent
The dynamic extension of $H^{t}$ is obtained from the Kronecker product of $H^{t}$, $H = \otimes_{t \in \mathcal{T}} H^{t}$. Likewise, the dynamic version of $\rho_{i}^{t}$ is obtained from the Kronecker product of $\rho_{i}^{t}$, $\rho_i = \otimes_{t \in \mathcal{T}} \rho_{i}^{t}$, $i \in \{1,2,3\}$. Note that the first two DGPs are inconsistent with RUM while the third DGP is consistent with RUM. In the same way, the dynamic version of the first two DGPs are inconsistent with DRUM while the third one is consistent with DRUM. Specifically, observe that the size of the violations of DRUM are larger in the first DGP than in the second DGP and that the third DGP is a knife-edge case.

For the current analysis, we consider the same setup as in our application with $K = T = J^{t} = 3$. We set the number of consumers per budget path to $N_{\rand{j}} \in \{10, 175, 350\}$. This choice is intended to be representative of the number of consumers per budget path in our application and to be informative about the small sample performance of the statistical test. We set the number of simulations for each DGP to $1000$. The critical value for each test statistic is computed using $999$ bootstrap samples. As recommended by KS, the tuning parameter $\tau_N$ is set to $\tau_N = \sqrt{\log(N_\rand{j})/N_\rand{j}}$.

The results are obtained using the test of KS that is based on the $\mathcal{V}$-representation of the model. The rejection rates at the $5$ percent significance level for all three sample sizes and for each DGP are presented in Table \ref{table: null2}. As expected, false positives are less likely under the first DGP than the second DGP. Also, the third DGP shows that false negatives quickly attain the desired target level as the sample size grows. Overall, the results of Table \ref{table: null2} show that the statistical test performs very well even in small samples. In that sense, the nonrejection of DRUM in our application is unlikely to be the byproduct of a lack of power.

\begin{table}[htbp]
\centering
\begin{tabular} {l|l|ccp{2cm}p{2cm}p{2cm}p{2cm}p{2cm}p{2cm}} DGP & $N_{j}$ & Rejection rate, \% \\ \hline
\multirow{2}{*}{DGP1} & 10 & 100 \\
& 175 & 100 \\
& 350 & 100 \\ \hline
\multirow{ 2}{*}{DGP2} & 10 & 25.3  \\
& 175 & 99.5 \\
& 350 & 100 \\ \hline
\multirow{ 2}{*}{DGP3} & 10 & 13.6  \\
& 175 & 6.0 \\
& 350 & 5.3   \\
\end{tabular}
\caption{Every entry represents the rejection rate at the $5$ percent significance level and is computed from $1000$ simulations and $999$ bootstraps per simulation.} \label{table: null2}
\end{table}




\end{document}